\documentclass[a4paper]{article}
\usepackage[noend, noline, ruled]{algorithm2e}
\usepackage{tikz}
\usepackage{paralist,expdlist}
\usepackage[margin=1in]{geometry}
\usepackage{amsmath,amsthm}
\usepackage{amssymb}
\usepackage{authblk}
\usepackage{array}
\usepackage[utf8]{inputenc}
\usepackage[T1]{fontenc}
\usepackage{enumitem}
\usepackage{stmaryrd}
\usepackage{xcolor, colortbl}
\usepackage{thmtools}
\usepackage{mathtools}
\usepackage{thm-restate}
\usepackage{microtype}
\usepackage{lmodern}
\usetikzlibrary{snakes}

\usepackage{lineno}

\usepackage[colorlinks]{hyperref}
\usepackage[capitalise]{cleveref}

\crefname{subappendix}{Appendix}{Appendices}
\crefname{appendix}{Appendix}{Appendices}
\crefname{section}{Section}{Sections}
\crefname{figure}{Figure}{Figures}

\setenumerate{itemsep=0pt,topsep=2pt,parsep=0pt,partopsep=0pt}
\setitemize{itemsep=0pt,topsep=2pt,parsep=0pt,partopsep=0pt}

\usetikzlibrary{decorations.pathreplacing,decorations.pathmorphing}

\newtheorem{theorem}{Theorem}[section]
\newtheorem{fact}[theorem]{Fact}
\newtheorem{lemma}[theorem]{Lemma}
\newtheorem{corollary}[theorem]{Corollary}
\newtheorem{proposition}[theorem]{Proposition}
\newtheorem{observation}[theorem]{Observation}

\theoremstyle{definition}
\newtheorem{definition}[theorem]{Definition}
\newtheorem{construction}[theorem]{Construction}

\newtheorem{claim}[theorem]{Claim}
\theoremstyle{remark}
\newtheorem{example}[theorem]{Example}

\crefname{section}{Section}{Sections}
\crefname{construction}{Construction}{Constructions}
\crefname{claim}{Claim}{Claims}
\crefname{fact}{Fact}{Facts}

\usepackage{tikz-qtree}
\usetikzlibrary{decorations.pathreplacing,decorations.pathmorphing}
\usetikzlibrary{shapes,arrows}
\usetikzlibrary{shapes.multipart}
\usetikzlibrary{decorations.text}
\usetikzlibrary{calc}

\newcommand{\IPMFull}{\textsc{Internal Pattern Matching Queries}\xspace}
\newcommand{\IPM}{\textsc{IPM Queries}\xspace}
\newcommand{\IPMOne}{\textsc{IPM Query}\xspace}
\newcommand{\RIPM}{\textsc{Restricted IPM Queries}\xspace}
\newcommand{\RIPMOne}{\textsc{Restricted IPM Query}\xspace}
\newcommand{\PQ}{\textsc{Period Queries}\xspace}
\newcommand{\BQ}{\textsc{Prefix-Suffix Queries}\xspace}
\newcommand{\BQone}{\textsc{Prefix-Suffix Query}\xspace}
\newcommand{\PEQ}{\textsc{Periodic Extension Queries}\xspace}
\newcommand{\FC}{\textsc{Cyclic Equivalence Queries}\xspace}
\newcommand{\BLCPFull}{\textsc{Bounded Longest Common Prefix (LCP) Queries}\xspace}
\newcommand{\BLCP}{\textsc{Bounded LCP Queries}\xspace}
\newcommand{\ILCP}{\textsc{Interval Longest Common Prefix (LCP) Queries}\xspace}
\newcommand{\ILCPFull}{\textsc{Interval Longest Common Prefix (LCP) Queries}\xspace}
\newcommand{\LSC}{\textsc{LZ Substring Compression Queries}\xspace}
\newcommand{\GSC}{\textsc{Generalized LZ Substring Compression Queries}\xspace}
\newcommand{\LCEQFull}{\textsc{Longest Common Extension Queries}\xspace}
\newcommand{\LCEQ}{\textsc{LCE Queries}\xspace}
\newcommand{\LCE}{\operatorname{LCE}\xspace}
\newcommand{\LZ}{\operatorname{LZ}}
\newcommand{\Zz}{\mathbb{Z}_{\ge 0}}

\newcommand{\pillar}{\ensuremath{\mathtt{PILLAR}}\xspace}

\newcommand{\TM}{\mathsf{TM}}

\newcommand{\SPEC}{\mathsf{SpRuns}}

\newcommand{\R}{\mathsf{Samples}}
\newcommand{\N}{\mathsf{NHP}}
\newcommand{\HP}{\mathsf{HP}}
\newcommand{\Y}{\mathsf{Y}}
\newcommand{\B}{\mathbf{B}}

\newcommand{\ceil}[1]{\left\lceil #1 \right\rceil}
\newcommand{\floor}[1]{\left\lfloor #1 \right\rfloor}

\newcommand{\eps}{\varepsilon}
\newcommand{\sub}{\subseteq}
\newcommand{\integ}{\mathbb{Z}}

\newcommand{\lcs}{\operatorname{lcs}}
\newcommand{\lcp}{\operatorname{lcp}}
\newcommand{\lcpinf}[2]{\lcp({#1}^\infty, {#2})}
\newcommand{\dd}{\mathinner{\ldotp\ldotp}}
\DeclareMathOperator{\per}{per}
\newcommand{\val}{\mathsf{val}}
\newcommand{\aalph}{\mathit{alph}}

\newcommand{\pred}{\mathrm{pred}}
\newcommand{\suc}{\mathrm{succ}}
\newcommand{\rank}{\mathrm{rank}}

\newcommand{\osiemsiedem}{\left(\tfrac{8}{7}\right)}
\newcommand{\Occ}{\mathit{Occ}}

\newcommand{\Rot}{\mathsf{rot}}
\newcommand{\ROT}{\mathsf{Rot}}

\newcommand{\maybeqed}{}
\newcommand{\RUNS}{\mathsf{RUNS}}
\newcommand{\Runs}{\mathsf{R}}
\newcommand{\MaxReg}{\mathsf{MaxReg}}
\newcommand{\Mid}{\mathsf{Mid}}

\newcommand{\code}{\mathsf{code}}

\newcommand{\Oh}{\mathcal{O}}

\renewcommand{\S}{\mathsf{Sync}}
\renewcommand{\S}{\mathsf{Sync}}
\newcommand{\F}{\mathsf{F}}
\newcommand{\SF}{\mathsf{SyncFr}}

\newcommand{\rsucc}{\mathit{rsucc}}

\newcommand{\BLOCK}{\mathsf{Block}}
\newcommand{\run}{\mathsf{run}}
\newcommand{\id}{\mathsf{id}}
\newcommand{\len}{\mathsf{len}}

\newsavebox{\mybox}
\newenvironment{dsproblem}[1]
{\begin{center}\begin{lrbox}{\mybox}\begin{minipage}{0.96\columnwidth}{\textsc{#1}}\\}
{\end{minipage}\end{lrbox}\fbox{\usebox{\mybox}}\end{center}}

  \newcommand{\defdsproblem}[2]{
  \begin{dsproblem}{#1}
#2
  \end{dsproblem}
  }

\title{Internal Pattern Matching Queries in~a Text and Applications\thanks{A preliminary version of this work was presented at the 26th Annual {ACM-SIAM} Symposium on Discrete Algorithms (SODA 2015)~\cite{DBLP:conf/soda/KociumakaRRW15} and included in the Ph.D.\ thesis of the first author~\cite{phd}.}}

\author[1]{Tomasz Kociumaka}
\author[2,3]{Jakub Radoszewski}
\author[2]{Wojciech Rytter}
\author[2]{Tomasz Wale\'n}
 
\affil[1]{Max Planck Institute for Informatics,  Saarbrücken, Germany}
\affil[ ]{\texttt{tomasz.kociumaka@mpi-inf.mpg.de}}
 \affil[2]{University of Warsaw, Warsaw, Poland}
 \affil[ ]{\texttt{[jrad, rytter, walen]@mimuw.edu.pl}}
 \affil[3]{Samsung R\&D Warsaw, Warsaw, Poland}

\date{\empty}

\begin{document}
\maketitle

\begin{abstract}
We consider several types of internal queries, that is, questions about \emph{fragments} of a given text $T$
specified in constant space by their locations in $T$.
Our main result is an optimal data structure for Internal Pattern Matching (IPM) queries which,
given two fragments $x$ and $y$, ask for a representation of all fragments contained in $y$ and matching $x$ exactly.
This problem can be viewed as an internal version of the fundamental Exact Pattern Matching problem:
we are looking for occurrences of one substring of $T$ within another substring.
Our data structure answers IPM queries in time proportional to the quotient $|y|/|x|$ of fragments' lengths,
which is required due to the information content of the output.
If $T$ is a text of length $n$ over an integer alphabet of size $\sigma$,
then our data structure occupies $\Oh(n/ \log_\sigma n)$ machine words (that is, $\Oh(n\log \sigma)$ bits)
and admits an $\Oh(n/ \log_\sigma n)$-time construction algorithm.

We show the applicability of IPM queries for answering internal queries corresponding to other classic string processing problems.
Among others, we derive optimal data structures reporting the periods of a fragment and testing the cyclic equivalence of two fragments.
Since the publication of the conference version of this paper (SODA 2015), IPM queries
have found numerous further applications, following the path paved by the classic Longest Common Extension (LCE) queries of Landau and Vishkin (JCSS, 1988).
In particular, IPM queries have been implemented in grammar-compressed and dynamic settings
and, along with LCE queries, constitute elementary operations of the \pillar model, developed by Charalampopoulos, Kociumaka, and Wellnitz
(FOCS 2020) to design universal approximate pattern-matching algorithms.

All our algorithms are deterministic, whereas the data structure in the conference version of the paper only admits randomized construction in $\Oh(n)$ expected time.
To achieve this, we provide a novel construction of \emph{string synchronizing sets} of Kempa and Kociumaka (STOC 2019).
Our method, based on a new \emph{restricted} version of the recompression technique of Jeż (\emph{J.\ ACM}, 2016),
yields a hierarchy of $\Oh(\log n)$ string synchronizing sets covering the whole spectrum of fragments' lengths.
\end{abstract}

\section{Introduction}
In this paper, we consider \emph{internal queries}, which ask to solve instances of a certain string processing problem with input strings given as fragments of a fixed string $T$ represented by their endpoints. The task is to preprocess $T$, called the \emph{text}, into a data structure that efficiently answers certain types of internal queries.
In more practical applications, the text $T$ can be interpreted as a corpus containing all the data gathered for a given study
(such as the concatenation of genomes to be compared).
On the other hand, when internal queries arise during the execution of algorithms, then $T$ is typically the input to be processed.

The setting of internal queries does not include problems like text indexing, where the text is available in advance but the pattern is explicitly provided at query time.
This restricts the expressibility of internal queries but at the same time allows for greatly superior query times, which do not need to account for reading any strings.
Another difference is in the typical usage scenario: Data structures for indexing problems are primarily designed to be interactively queried by a user.
In other words, they are meant to be constructed relatively rarely, but they need to be stored for prolonged periods of time.
As a result, space usage (including the multiplicative constants) is heavily optimized, while the efficiency of construction procedures is of secondary importance.
On the other hand, internal queries often arise during bulk processing of textual data.
The relevant data structures are then built within the preprocessing phase of the enclosing algorithms, 
so the running times of the construction procedures are counted towards the overall processing time. 
In this setting, efficient construction is as significant as fast queries. 

The origins of internal queries in texts can be traced back to the invention of the \LCEQFull (\LCEQ),
originally introduced by Landau and Vishkin~\cite{DBLP:journals/jcss/LandauV88}.
In such a query, given two fragments of the text $T$, one is to compute the length of their longest common prefix.
The classic data structure for \LCEQ takes linear space, supports constant-time queries, and can be constructed in linear time~\cite{DBLP:journals/siamcomp/HarelT84,DBLP:journals/jcss/LandauV88,DBLP:journals/jacm/Farach-ColtonFM00}.
Although this trade-off is optimal for large alphabets, a text of length $n$ over an alphabet of size $\sigma$ can, in general, be stored in $\Oh(n\log \sigma)$ bits (which is $\Oh(n / \log_\sigma n)$ machine words) and read in $\Oh(n / \log_\sigma n)$ time.
Only recently, Kempa and Kociumaka~\cite{phd,Kempa2019} showed that such complexity and preprocessing time are sufficient for constant-time \LCEQ. 

Components with optimal deterministic construction and constant-time queries are particularly valued,
because algorithms can use them essentially for free,
i.e., with no negative effect on the overall complexity in the standard theoretical setting.
We develop data structures with such guarantees for internal versions of several natural problems in text processing, including Exact Pattern Matching.

We always denote the input text by $T$ and its length by $n$.
We also make a standard assumption (cf.\ \cite{DBLP:journals/jacm/Farach-ColtonFM00}) that the characters of $T$ are (or can be identified with) integers $[0\dd \sigma)$
\footnote{ For $i,j\in \integ$, we
denote $[i\dd j]=\{k \in \integ : i \le k \le j\}$, $[i\dd j)=\{k \in \integ : i \le k < j\}$, and $(i\dd j]={\{k \in \integ: i < k \le j\}}$. },
where the $\sigma = n^{\Oh(1)}$, that is, $T$ is over a polynomially-bounded \emph{integer alphabet}.
Our results are designed for the standard word RAM model with machine words of $\omega \ge \log n$ bits%
\footnote{Throughout this paper, $\log$ denotes the base-2 logarithm (any other base is explicitly provided in the subscript).}.
In this model, the text $T$ can be represented using $\Oh(n/\log_\sigma n)$ machine words, that is, $\Oh(n \log \sigma)$ bits, in a
so-called packed representation; see~\cite{DBLP:journals/tcs/Ben-KikiBBGGW14}.

Before we formally state our results, let us recall a few elementary notations related to strings;
a more thorough exposition is provided in \cref{chp:prelim}.
For an alphabet $\Sigma$ and an integer $m \in \Zz$, by $\Sigma^m$ we denote the set of length-$m$ strings over $\Sigma$.
For a string $w$, by $|w|$ we denote its length and by $w[0],\ldots,w[|w|-1]$ its subsequent characters.
We call $w[i \dd j)$ a \emph{fragment} of $w$; it can be interpreted as a range $[i \dd j)$ of positions in~$w$.
The fragment is called a \emph{prefix} if $i=0$ and a \emph{suffix} if $j=|w|$.
A fragment $w[i \dd j)$ corresponds to a substring of $w$, that is, a string composed of letters $w[i],\ldots,w[j-1]$.
Let us note that many fragments can correspond to (be occurrences of) the same substring.
The \emph{concatenation} of two strings $u,v$ is denoted $u\cdot v$ or~$uv$.
We identify strings of length 1 with the underlying characters, which lets us write $w=w[0]\cdots w[|w|-1]$.
The length of the longest common prefix of two strings $v,w$ is denoted by $\lcp(v,w)$.
For a string $w$ and an integer $k\in \Zz$, we denote the concatenation of $k$ copies of $w$ by $w^k$.
An integer $p$ is a \emph{period} of a string $w$ if $p \in [1 \dd |w|]$ and $w[i]=w[i+p]$ for all $i \in [0 \dd |w|-p)$.
The smallest period of $w$ is denoted as $\per(w)$.

\subsection{Internal Pattern Matching (IPM) Queries}
Our main contribution here is a data structure
for the internal version of the Exact Pattern Matching problem,
which asks for the occurrences of one fragment within another fragment.

\defdsproblem{\textsc{Internal Pattern Matching (IPM) Queries}}{Given fragments $x$ and $y$ of 
 $T$ satisfying $|y|< 2|x|$, report the fragments matching $x$ and contained in $y$ (represented as an arithmetic progression of starting positions).}
Let us note that if $|y|\ge2|x|$, then one can  ask $\Oh(|y|/|x|)$ \IPM 
(for the occurrences of $x$ within fragments of length $2|x|-1$ contained in $y$, with overlaps of at least $|x|-1$ characters between the subsequent fragments) and output $\Oh(|y|/|x|)$ arithmetic progressions.
We impose the restriction $|y|<2|x|$ so that the output can be represented in constant space using the following folklore fact:

\begin{fact}[Breslauer and Galil~\cite{DBLP:journals/algorithmica/BreslauerG95}, Plandowski and Rytter~\cite{DBLP:conf/icalp/PlandowskiR98}]\label{fct:single}
  Let $x$, $y$ be strings satisfying $|y|< 2|x|$.
  The set of starting positions of the occurrences of $x$ in $y$ forms a single arithmetic progression.
\end{fact}

We design an optimal data structure for \IPM.

\begin{restatable}[Main result]{theorem}{thmipm}\label{thm:ipm}
  For every text $T\in [0\dd \sigma)^n$, there exists a data structure of size $\Oh(n/\log_\sigma n)$ which answers \IPM in $\Oh(1)$ time.  The data structure can be constructed in $\Oh(n / \log_\sigma n)$ time given the packed representation of $T$.
\end{restatable}

Linear-time solutions of the exact pattern matching problem~\cite{morris1970linear,DBLP:journals/siamcomp/KnuthMP77,DBLP:journals/cacm/BoyerM77} are among the foundational results of string algorithms.
Although this alone can motivate the study of \IPM, the significance of our result is primarily due to a growing collection of applications.
In this work, we show that several further internal queries can be answered efficiently using \cref{thm:ipm}; these are described in detail in \cref{subsec:app}.
Other applications of \IPM, developed after the conference version of the paper~\cite{DBLP:conf/soda/KociumakaRRW15}, are listed in~\cref{sec:alien}.

\subsection{Applications of IPM Queries}\label{subsec:app}
\paragraph{\underline{Period Queries}.}
One of the central notions of combinatorics on words is that of a \emph{period} of a string.
Already Knuth, Morris, and Pratt~\cite{morris1970linear,DBLP:journals/siamcomp/KnuthMP77}
provided a linear-time procedure listing all periods of a given string (as a side result of their linear-time pattern-matching algorithm).

The sorted sequence of periods of a length-$m$ string can be cut into $\Oh(\log m)$ arithmetic progressions~\cite{DBLP:journals/siamcomp/KnuthMP77}.
A complete characterization of the possible families of periods~\cite{DBLP:journals/jct/GuibasO81} further shows
that the size of such a representation ($\Oh(\log^2 m)$ bits) is asymptotically tight.
Hence, we adopt it in the internal version of the problem of finding all periods of a string,
formally specified below.

\defdsproblem{\PQ}{Given a fragment $x$ of $T$, report all periods of $x$ (represented by disjoint arithmetic progressions).}

We have introduced \PQ in~\cite{DBLP:conf/spire/KociumakaRRW12},
presenting two solutions. 
The first data structure takes $\Oh(n\log n)$ space and answers \PQ in the optimal $\Oh(\log |x|)$ time
after $\Oh(n\log n)$-time randomized construction.
The other one is based on orthogonal range searching;
its size is $\Oh(n+S_\rsucc(n))$ and the query time is $\Oh(Q_\rsucc(n)\cdot \log |x|)$,
where $S_{\rsucc}(n)$ and $Q_{\rsucc}(n)$
are analogous values for data structures answering range successor queries; see~\cref{sec:GSC} for a definition.
The state-of-the-art trade-offs are $S_{\rsucc}(n)=\Oh(n)$ and $Q_{\rsucc}(n)=\Oh(\log^{\eps} n)$ for every constant $\eps>0$~\cite{DBLP:conf/swat/NekrichN12},
$S_{\rsucc}(n)=\Oh(n \log \log n)$ and $Q_{\rsucc}(n)=\Oh(\log \log n)$~\cite{DBLP:journals/ipl/Zhou16},
as well as $S_{\rsucc}(n)=\Oh(n^{1+\eps})$ and $Q_{\rsucc}(n)=\Oh(1)$ for every constant $\eps>0$~\cite{DBLP:journals/tcs/CrochemoreIKRTW12}.
The first two of these data structures can be constructed in time $C_{\rsucc}(n)=\Oh(n\sqrt{\log n})$~\cite{DBLP:conf/soda/BelazzouguiP16,Gao2020},
and the third one in time $C_{\rsucc}(n)=\Oh(n^{1+\eps})$~\cite{DBLP:journals/tcs/CrochemoreIKRTW12}.

In this paper, we develop an asymptotically optimal data structure answering \PQ.

\begin{restatable}{theorem}{thmapppq}\label{thm:app-pq}
For every text $T\in [0\dd \sigma)^n$, there exists a data structure of size $\Oh(n/\log_\sigma n)$ which answers \PQ in $\Oh(\log |x|)$ time.
The data structure can be constructed in $\Oh(n/\log_\sigma n)$ time given the packed representation of $T$.
\end{restatable}

Our query algorithm is based on the intrinsic relation between periods and \emph{borders} (i.e., substrings being both prefixes and suffixes) of a string.
In fact, to answer each \textsc{Period Query}, it combines the results of the following \BQ
used with $x=y$ to determine the borders of~$x$.

\defdsproblem{\BQ}{Given fragments $x$ and $y$ of $T$ and a positive integer $d$, report all suffixes of $y$ of length in $[d \dd 2d)$
that also occur as prefixes of $x$ (represented as an arithmetic progression of their lengths).}

In other words, we prove the following auxiliary result.

\begin{restatable}{theorem}{thmappbq}\label{thm:app-bq}
For every text $T\in [0\dd\sigma)^n$, there exists a data structure of size $\Oh(n/\log_\sigma n)$ which answers \BQ in $\Oh(1)$ time. The data structure can be constructed in $\Oh(n/\log_\sigma n)$ time given the packed representation of $T$.
\end{restatable}

In many scenarios, very long periods ($p=m-o(m)$ for a string of length $m$) are irrelevant.
The remaining periods correspond to borders of length $\Theta(m)$ and thus can be retrieved with just a constant number of \BQ.
The case of $p \le \tfrac12 m$ is especially important
since fragments $x$ with periods not exceeding $\tfrac12|x|$, called \emph{periodic} fragments, can be uniquely extended to \emph{maximal repetitions}, also known as \emph{runs} (see~\cref{chp:prelim}).
We denote the unique run extending a periodic fragment $x$ by $\run(x)$.
If $x$ is not periodic, we leave $\run(x)$ undefined, which we denote as $\run(x)=\bot$.

\defdsproblem{\PEQ}{Given a fragment $x$ of $T$, compute the run $\run(x)$ extending $x$.}

\noindent
\cref{thm:app-bq}, along with the optimal data structure for \LCEQ~\cite{Kempa2019}, imply the following result.

\begin{restatable}{theorem}{thmrun}\label{thm:run}
  For every text $T\in [0\dd\sigma)^n$, there exists a data structure of size $\Oh(n/\log_\sigma n)$ which answers \PEQ in $\Oh(1)$ time. The data structure can be constructed in $\Oh(n/\log_\sigma n)$ time given the packed representation of $T$.
\end{restatable}

Bannai et al.~\cite{DBLP:journals/siamcomp/BannaiIINTT17} presented an alternative implementation of \PEQ with $\Oh(n)$-time construction and $\Oh(1)$-time queries.
The underlying special case of \PQ also generalizes \textsc{Primitivity Queries}
(asking if a fragment $x$ is \emph{primitive}, i.e., whether it does not match $u^k$ for any string $u$ and integer $k\ge 2$),
earlier considered by Crochemore et al.~\cite{DBLP:journals/tcs/CrochemoreIKRRW14},
who developed a data structure of size $\Oh(n+S_\rsucc(n))$ with  $\Oh(Q_{\rsucc}(n))$-time query algorithm.

\paragraph{\underline{Cyclic Equivalence Queries}.}
We consider a \emph{rotation} operation which moves the last character of a given string $u$ to the front.
Formally, if $|u|=m$, then $\Rot(u)=u[m-1] u[0] \cdots u[m-2]$.
In general, for $j\in \integ$, we define the $\Rot^j$ function as the $j$th function power of $\Rot$.
Two strings $u$ and $v$ are called cyclically equivalent if $u=\Rot^j(v)$ for some integer $j$.
A classic linear-time algorithm for checking cyclic equivalence of strings $u$ and $v$ consists in pattern matching for $u$ in $v^2$~\cite{DBLP:journals/ipl/Manacher76}; there is also a simple linear-time constant-space algorithm~\cite{Jewels}.
We consider the following queries.

\defdsproblem{\FC}{
Given two fragments $x$ and $y$ of $T$, return $\{j\in \integ : \Rot^j(x)=y\}$ (represented as an arithmetic progression).
}

\begin{restatable}{theorem}{thmfc}\label{thm:app-fc}
  For every text $T\in [0\dd \sigma)^n$, there exists a data structure of size $\Oh(n/\log_\sigma n)$ which answers \FC in $\Oh(1)$ time. The data structure can be constructed in $\Oh(n/\log_\sigma n)$ time given the packed representation of $T$.
\end{restatable}

\paragraph{\underline{Bounded LCP Queries}.}
Keller et al.~\cite{DBLP:journals/tcs/KellerKFL14} used the following queries in the solution to their \GSC problem:
\defdsproblem{\BLCPFull}{Given two fragments $x$ and $y$ of $T$, find the longest prefix $p$ of $x$ which  occurs in $y$.}

Our result for \IPM can be combined with the techniques of~\cite{DBLP:journals/tcs/KellerKFL14} in a more efficient implementation of \BLCP.
Compared to the original version, the resulting data structure,
specified below, has a $\log \log |p|$ factor instead of a $\log |p|$ factor in the query time.

\begin{restatable}{theorem}{thmblcp}\label{thm:blcp}
For every text $T$ of length $n$ over an alphabet $[0 \dd n^{\Oh(1)})$, there exists a data structure of size $\Oh(n+S_{\rsucc}(n))$
which answers \BLCP in $\Oh(Q_{\rsucc}(n)\log\log |p|)$ time.
The data structure can be constructed in $\Oh(n+C_{\rsucc}(n))$ time.
\end{restatable}

In \cref{chp:app}, we formally define \GSC and discuss how \cref{thm:blcp} lets us improve and extend the results of~\cite{DBLP:journals/tcs/KellerKFL14}.

\paragraph{\underline{Earlier Versions of Our Results}.}
Weaker versions of \cref{thm:ipm,thm:app-pq,thm:app-bq,thm:run,thm:app-fc} (with $\Oh(n)$ space and construction time) and \cref{thm:blcp} were published in the conference version of the paper~\cite{DBLP:conf/soda/KociumakaRRW15} and the Ph.D.\ thesis of the first author~\cite{phd}. Moreover, the construction algorithms provided in~\cite{DBLP:conf/soda/KociumakaRRW15} were Las-Vegas randomized, with linear bounds on the \emph{expected} running time only.

\subsection{Further Applications of IPM Queries}\label{sec:alien}
Since the publication of the conference version of this work~\cite{DBLP:conf/soda/KociumakaRRW15}, 
the list of internal queries, predominantly implemented using \IPM, has grown to include
shortest unique substrings~\cite{DBLP:journals/algorithms/Abedin0PT20},
longest common substring~\cite{DBLP:journals/corr/abs-1804-08731},
suffix rank and selection~\cite{WaveletSuffixTree,DBLP:conf/cpm/Kociumaka16},
BWT substring compression~\cite{WaveletSuffixTree},
shortest absent string~\cite{DBLP:journals/tcs/BadkobehCKP22},
dictionary matching~\cite{DBLP:journals/algorithmica/Charalampopoulos21},
string covers~\cite{DBLP:conf/spire/CrochemoreIRRSW20},
masked prefix sums~\cite{DBLP:conf/spire/Das0KMW22},
circular pattern matching~\cite{CPM2023}, and
longest palindrome~\cite{DBLP:conf/walcom/MitaniMSH23}.

Furthermore, \IPM have been used in efficient algorithms for many problems such
as approximate pattern matching~\cite{DBLP:conf/focs/Charalampopoulos20,DBLP:conf/focs/Charalampopoulos22},
approximate circular pattern matching~\cite{DBLP:journals/jcss/Charalampopoulos21,DBLP:conf/esa/Charalampopoulos22},
RNA folding~\cite{DBLP:conf/icalp/0001KS22},
and computing string covers~\cite{DBLP:conf/esa/RadoszewskiS20}.
Additionally, \IPM have found further indirect applications that are based on the internal queries from \cref{subsec:app}:
\PEQ have been applied for
approximate period recovery~\cite{DBLP:journals/tcs/AmirALS18,DBLP:journals/algorithmica/AmirBKLS22},
dynamic repetition detection~\cite{DBLP:conf/esa/AmirBCK19},
identifying two-dimensional maximal repetitions~\cite{DBLP:conf/esa/AmirLMS18},
enumeration of distinct substrings~\cite{DBLP:conf/spire/Charalampopoulos20},
and pattern matching with variables~\cite{DBLP:journals/toct/FernauMMS20,DBLP:conf/spire/KosolobovMN17},
whereas
\BQ have been applied for detecting gapped repeats and subrepetitions~\cite{DBLP:journals/jda/KolpakovPPK17,DBLP:journals/mst/GawrychowskiIIK18},
in the dynamic longest common substring problem~\cite{DBLP:journals/corr/abs-1804-08731},
and for computing the longest unbordered substring~\cite{DBLP:journals/corr/abs-1805-09924}.

The fundamental role of \IPM as a building block for the design of string algorithms
motivated their efficient implementation in the compressed and dynamic settings~\cite{DBLP:conf/focs/Charalampopoulos20}.
In particular, the \pillar model, introduced in~\cite{DBLP:conf/focs/Charalampopoulos20}
with the aim of unifying approximate pattern-matching algorithms across different settings, includes \IPM as one of the primitives.
It also includes \LCEQ, defined as $\LCE(i,i')=\lcp(w[i\dd |w|),w[i'\dd |w|))$, and \LCEQ on reversed strings,
as well as the following basic primitives:
\begin{itemize}
  \item $\mathtt{Extract}(w,\ell,r)$: Given a string $w$ and integers $0\le \ell \le r \le |w|$, retrieve the string $w[\ell\dd r)$.
  \item $\mathtt{Access}(w,i)$: Given a string $w$ and a position $i\in [0\dd |w|)$, retrieve the character $w[i]$.
  \item $\mathtt{Length}(w)$: Retrieve the length $|w|$ of the string $w$.
\end{itemize}
The argument strings of \pillar primitives are represented as fragments of one or more strings in a given text collection $\mathcal{X}$.

Using an earlier version of \cref{thm:ipm}, providing $\Oh(n)$-time deterministic construction~\cite{phd}, it has been observed \cite[Theorem 7.2]{DBLP:conf/focs/Charalampopoulos20} that, after $\Oh(n)$-time preprocessing of a collection $\mathcal{X}$ of strings of total length $n$,
each \pillar operation can be performed in $\Oh(1)$ time.
We improve upon this result using \cref{thm:ipm} to implement \IPM and
the following implementation of \LCEQ.
\begin{proposition}[Kempa and Kociumaka~\cite{Kempa2019}]\label{prop:lce}
For every text $T$ of length $n$ over alphabet $[0\dd \sigma)$, there exists a data structure of size $\Oh(n/\log_\sigma n)$ which answers \LCEQ in $\Oh(1)$ time. The data structure can be constructed in $\Oh(n/\log_\sigma n)$ time given the packed representation of $T$.
\end{proposition}

We apply elementary bit-wise operations for $\mathtt{Access}$ queries
(the $\mathtt{Extract}$ and $\mathtt{Length}$ queries are straightforward since $w=X[\ell\dd r)$ is represented by a pointer to $X\in \mathcal{X}$ and the two endpoints $\ell$ and $r$).
This gives the following result.

\begin{theorem}\label{thm:pillar}
A collection $\mathcal{X}$ of strings of total length $n$ over alphabet $[0\dd \sigma)$
can be preprocessed in $\Oh(|\mathcal{X}|+n/\log_\sigma n)$ time so that each \pillar operation can be performed in $\Oh(1)$ time.
\end{theorem}

Since the approximate pattern-matching algorithms of \cite{DBLP:conf/focs/Charalampopoulos20,DBLP:conf/focs/Charalampopoulos22,DBLP:conf/esa/Charalampopoulos22} are implemented in the \pillar model,
\cref{thm:pillar} immediately improves the running times for strings over small alphabets.
In particular, given a pattern $p\in [0\dd \sigma)^m$, a text $t\in [0\dd \sigma)^n$, and a threshold $k$, the occurrences of $p$ in $t$ with at most $k$ mismatches (substitutions) can be reported in time $\Oh(n/\log_\sigma n + (n/m)\cdot \min(k\sqrt{m\log m},k^2))$,
whereas the occurrences with at most $k$ edits (insertions, deletions, and substitutions) can be reported in time $\Oh(n/\log_\sigma n + (n/m)\cdot k^3\sqrt{\log m\log k})$. In both cases, the improvement is that the $\Oh(n)$ term, dominating the complexity of previous state-of-the-art solutions~\cite{DBLP:conf/stoc/ChanGKKP20,DBLP:conf/focs/Charalampopoulos22} for small values of $k$, is replaced by $\Oh(n/\log_\sigma n)$. 
A similar phenomenon applies to circular pattern matching with at most $k$ mismatches~\cite{DBLP:conf/esa/Charalampopoulos22}, where we achieve $\Oh(n/\log_\sigma n + (n/m)\cdot k^3 \log \log k)$ time for the reporting version of the problem and $\Oh(n/\log_\sigma n + (n/m)\cdot k^2 \log k/ \log \log k)$ for the decision version.

\subsection{Related Queries}\label{sec:related}
Internal queries are not the only problems in the literature involving fragments of a static text.
Other variants include
Interval LCP queries~\cite{DBLP:conf/soda/CormodeM05,DBLP:journals/tcs/KellerKFL14},
Range LCP queries~\cite{DBLP:conf/cocoon/Abedin0HNSST18,RangeLCP,DBLP:conf/spire/AmirLT15,DBLP:conf/spire/PatilST13},
Substring Hashing queries~\cite{DBLP:conf/cpm/FarachM96,DBLP:conf/esa/Gawrychowski11,DBLP:conf/esa/GawrychowskiLN14},
fragmented Pattern Matching queries~\cite{DBLP:journals/talg/AmirLLS07,DBLP:conf/esa/GawrychowskiLN14},
and Cross-Document Pattern Matching queries~\cite{DBLP:journals/jda/KopelowitzKNS14}, to mention a few.
In particular, Interval LCP queries can be used to solve the decision version of \IPM.
Keller et al.~\cite{DBLP:journals/tcs/KellerKFL14}
showed how to answer the decision version of \IPM in $\Oh(Q_{\rsucc}(n))$ time
using a data structure of size $\Oh(n+S_{\rsucc}(n))$ that can be constructed in $\Oh(n+C_{\rsucc}(n))$ time.
The aforementioned query time is valid for arbitrary lengths $|x|$ and $|y|$,
so the efficiency of this data structure is incomparable to our~\cref{thm:ipm}.

\subsection{Technical Contributions}\label{sec:techniques}
Below, we briefly introduce the most important technical contributions of our work.
We start with a high-level overview of our \IPM data structure for large alphabets,
that is, $\sigma = n^{\Theta(1)}$.
In this setting, the space and construction time bounds of \cref{thm:ipm} simplify to $\Oh(n)$.

\paragraph{String Matching by Deterministic Sampling}\label{sec:samp}
The idea of (deterministic) sampling is a classic technique originally developed for parallel string matching algorithms~\cite{DBLP:journals/siamcomp/Vishkin91} and later applied in other contexts such as quantum string matching~\cite{DBLP:journals/jda/HariharanV03}.
Algorithms using this approach to find the occurrences of a pattern $x$ in a text $y$ follow a three-phase scheme.
In the preprocessing phase, a subset of characters of $x$, called the \emph{sample}, is determined.
Then, in the filtering phase, the algorithm selects all fragments $y[i\dd i+|x|)$ of $y$ that match the sample.
Finally, in the verification phase, the algorithm checks whether each candidate $y[i\dd i+|x|)$ matches the whole pattern $x$.
The efficiency of this scheme hinges on two properties of the chosen samples: (1) they must be simple enough to support efficient filtering,
and (2) they must carry enough information to leave few candidates for the verification phase.
For example, Vishkin~\cite{DBLP:journals/siamcomp/Vishkin91} chooses a sample of size $\Oh(\log |x|)$ so that the starting positions of the candidates surviving the filtering phase form $\Oh(|y|/|x|)$ arithmetic progressions whose difference is the smallest period of~$x$.

Our data structure answering \IPM, described in \cref{chp:ipm}, uses contiguous samples, that is, the sample of a pattern $x$ can be interpreted as a fragment of $T$ contained in $x$. Moreover, we minimize the total number of samples across all the fragments of $T$ rather than the size of each sample,
and we aim to choose the samples \emph{consistently}, so that any fragment matching a sample is a sample itself.
These constraints are infeasible for highly periodic ($\HP$) patterns, whose length is much larger than the smallest period (we set $\frac13|x|$ as the cut-off point for the period), and such patterns are handled in \cref{chp:per} using different techniques outlined later on.
As for non-highly-periodic ($\N$) patterns, we choose $\Oh(n)$ samples in total so that the smallest period of the sample of $x$ is $\Omega(|x|)$.
The small number of samples allows storing them explicitly, whereas the large period guarantees that the sample of $x$ has $\Oh(1)$ occurrences within any fragment $y$ of length $|y|<2|x|$ (due to \cref{fct:single}).
Thus, at query time, our data structure uses the precomputed set of samples (stored in appropriate deterministic dictionaries~\cite{DBLP:conf/icalp/Ruzic08,DBLP:conf/focs/PatrascuT14}) to identify the sample of $x$ and its occurrences in $y$.
Then, we use \LCEQ to test which of these occurrences extend to occurrences of~$x$.

\paragraph{Sample Selection: String Synchronizing Sets Hierarchy and Restricted Recompression}
The challenge of implementing the strategy outlined above is to consistently pick $\Oh(n)$ samples among $\Theta(n^2)$ fragments of $T$.
The natural first step is to restrict the selection to $\Theta(n\log n)$ fragments whose lengths form a geometric series,
but any further reduction in the number of samples requires non-trivial symmetry breaking.

In the conference version of this paper~\cite{DBLP:conf/soda/KociumakaRRW15}, we employed a strategy reminiscent of the \emph{minimizers} technique popular in bioinformatics~\cite{DBLP:journals/bioinformatics/RobertsHHMY04,Wood_2014,DBLP:journals/bioinformatics/DeorowiczKGD15,DBLP:journals/bioinformatics/Li18}
and known under different names in many other applications~\cite{DBLP:conf/sigmod/SchleimerWA03,DBLP:conf/icdm/SorokinaGWG06,DBLP:journals/ce/ButakovS09,DBLP:journals/tissec/PonecGWB10}.
In that approach, candidate fragments are consistently assigned uniformly random weights, and the sample of $x$ is defined as the minimum-weight fragment of a certain length (such as $2^{\floor{\log |x|}-1}$) contained in $x$. 
With minor adaptations (necessary to avoid highly periodic samples), this scheme yields $\Oh(n/2^{k})$ length-$2^k$ samples in expectation for every $k\in [0\dd \floor{\log n}]$.
Subsequently, our sample selection algorithm was adapted for answering \LCEQ~\cite{DBLP:conf/soda/BirenzwigeGP20,phd} and for the Burrows--Wheeler transform construction~\cite{Kempa2019}. The latter paper contributed the following clean notion, which has since been applied in many further contexts (see e.g.~\cite{DBLP:conf/focs/KempaK20,DBLP:conf/esa/Charalampopoulos21,DBLP:conf/stoc/KempaK22,DBLP:conf/soda/JinN23}).

\begin{restatable}[\textbf{Synchronizing Set}; Kempa and Kociumaka~\cite{Kempa2019}]{definition}{defsss}\label{def:sss}
  Let $T$ be a string of length $n$ and let $\tau \in [1\dd \floor{\frac{n}{2}}]$.
  A set $\S\sub [0\dd n-2\tau]$ is a
  \emph{$\tau$-synchronizing set} of $T$ if it satisfies the following conditions:
  \begin{description}
\item[Consistency:] For all $i,j\in [0\dd n-2\tau]$, if $i\in \S$ and the fragments $T[i\dd i+2\tau)$ and $T[j\dd j+2\tau)$ match, then $j\in \S$.
\item[Density:]
For every $i\in [0\dd n-3\tau+1]$, we have $[i\dd i+\tau)\cap \S = \emptyset\;\Longleftrightarrow\; \per(T[i\dd i+3\tau-1)) \le \tfrac13\tau$.
\end{description}
We say that the elements of a fixed $\tau$-synchronizing set $\S$ are $\tau$-\emph{synchronizing positions} and the fragments $T[s\dd s+2\tau)$ for $s\in \S$ are $\tau$-\emph{synchronizing fragments}; the set of $\tau$-synchronizing fragments is denoted~$\SF$.
\end{restatable}

Kempa and Kociumaka~\cite{Kempa2019} obtained the following result building upon our original sample-selection algorithm~\cite{DBLP:conf/soda/KociumakaRRW15}
and its derandomized version presented in~\cite{phd}.
\begin{restatable}[{\cite[Proposition 8.10 and Theorem 8.11]{Kempa2019}}]{proposition}{prpsynch}\label{prp:synch}
For every text $T \in [0\dd \sigma)^n$ with $\sigma=n^{\Oh(1)}$ and $\tau\in [1\dd \floor{\frac{n}{2}}]$, there exists a $\tau$-synchronizing set of size $\Oh(n/\tau)$ that can be constructed in $\Oh(n)$ time. Moreover, if $\tau \leq \frac15 \log_\sigma n$ and $T$ is given in a packed representation, then the construction time can be improved to $\Oh(n/\tau)$.
\end{restatable}
It is not hard to argue (see \cref{lem:repr}) that one can pick as samples all the length-$1$ fragments as 
well as all the $2^{k-1}$-synchronizing fragments for all $k\in [1\dd \floor{\log n}]$ and some fixed $2^{k-1}$-synchronizing sets.
Unfortunately, it takes $\Oh(n \log n)$ time to construct these synchronizing sets using the algorithm of \cref{prp:synch}.
Consequently, the deterministic version of \cref{thm:ipm} published in~\cite{phd} follows a sophisticated approach relying on just two string synchronizing sets, constructed in $\Oh(n)$ time each.
In this paper, we apply a more natural strategy and show that the entire hierarchy of  $2^{k-1}$-synchronizing sets
can be constructed in deterministic $\Oh(n)$ time. Formally, our new result reads as follows:

\begin{restatable}[\bf Construction of Synchronizing Sets Hierarchy]{theorem}{thmsss}\label{thm:sss}
  After $\Oh(n)$-time preprocessing of a text $T \in [0\dd \sigma)^n$ with $\sigma=n^{\Oh(1)}$,
  given $\tau\in [1\dd \floor{\frac{n}{2}}]$, one can construct
  in $\Oh(\frac{n}{\tau})$ time a $\tau$-synchronizing set $\S$ of size  $|\S|< \tfrac{70n}{\tau}$.
  Moreover, no $\tau$-synchronizing fragment induced by $\S$ has period $p\le \frac{\tau}{3}$.
\end{restatable}

As opposed to~\cite{Kempa2019}, our construction does not use minimizers, but we rely on 
locally consistent parsing, a concept dating back to mid-1990s~\cite{FirstConsistentParsing,SymmetryBreakingST,DBLP:conf/focs/SahinalpV96,DBLP:journals/algorithmica/MehlhornSU97}. Specifically, we adapt the \emph{recompression} technique of Jeż~\cite{DBLP:journals/talg/Jez15,DBLP:journals/jacm/Jez16},
which results in a simple and efficient parsing scheme.

All locally consistent parsing algorithms construct a sequence of factorizations of $T$, where
level-$0$ phrases consist of single letters of $T$, the only level-$q$ phrase (for some $q=\Oh(\log n)$) consists of the entire text $T$,
and, for each $k \in [1\dd q]$, the level-$k$ phrases are concatenations of level-$(k-1)$ phrases.
In this context, local consistency means that whether or not two subsequent level-$(k-1)$ phrases are grouped into the same level-$k$ phrase 
is a local decision that depends only on a few neighboring  level-$(k-1)$ phrases.
Unfortunately, the lengths of level-$(k-1)$ phrases can vary significantly between regions of the text, and thus it is not possible to make sure that the symmetry-breaking decisions are made based on fixed-size contexts, as required by the consistency property of string synchronizing sets (\cref{def:sss}).

Consequently, we alter the original recompression algorithm and introduce a small but consequential restriction: phrases that are deemed too long for their level are never merged with their neighbors. Although this trick does not eliminate very long phrases, it lets us quantify local consistency in terms of fixed-size contexts. In the following lemma, proved in \cref{sec:recompression}, the parameter $\alpha_k$ can be interpreted as the level-$k$ context size.
\begin{restatable}{lemma}{lemcons}\label{lem:cons}
Let $\B_k$ denote the set of starting positions of level-$k$ phrases (except for the leftmost one). 
Then, for every $i,j\in [\alpha_k\dd n-\alpha_k]$, if $i\in \B_k$ and the fragments $T[i-\alpha_k\dd i+\alpha_k)$ and $T[j-\alpha_k\dd j+\alpha_k)$ match, then $j\in \B_k$.
\end{restatable}

Restricted recompression has already been used in~\cite{DBLP:journals/tit/KociumakaNP23} to prove that every text admits a run-length straight-line program of a particular size.
Restricted variants of other locally consistent parsing schemes have been used in~\cite{DBLP:conf/latin/KociumakaNO22} and \cite{DBLP:conf/stoc/KempaK22} to derive efficient compressed and dynamic text indexes. In both cases, the counterpart of \cref{lem:cons} is what differentiates restricted variants from the original implementations of the locally consistent parsing schemes.

In \cref{chp:lce}, we provide further intuition and explain how to define the synchronizing sets in terms of the sets $\B_k$ of phrase boundaries.
\cref{sec:recompression} provides the formal definition and analysis of the restricted recompression technique.
The proof of \cref{thm:sss} is given in \cref{sec:correctness}.

\paragraph{\IPM in Highly Periodic Patterns}
As mentioned above, our deterministic sampling strategy applies only to non-highly-periodic ($\N$) patterns.
In \cref{chp:per}, we develop a complementary data structure that is responsible for handling highly-periodic ($\HP$) patterns.

For this, we build upon the fact that the structure of highly-periodic fragments can be encoded by \emph{maximal repetitions} (also known as runs)~\cite{DBLP:journals/dam/Main89,DBLP:conf/focs/KolpakovK99,DBLP:journals/siamcomp/BannaiIINTT17}.
In particular, we rely on the notion of \emph{compatibility}~\cite{DBLP:journals/tcs/CrochemoreIKRRW14}: two strings are compatible whenever their 
string periods are cyclically equivalent.
If $x$ is periodic, every matching fragment $x'$ can be extended to a run compatible with $x$.
Moreover, due to the assumption $|y|<2|x|$, if $x'$ is contained within $y$, then the run must contain the middle position of~$y$.
Consequently, we develop a simple new data structure that, for any position in $T$, in $\Oh(1)$ time  
lists all runs with a certain minimum length (such as $|x|$) and maximum period (such as $\frac13|x|$). 
Next, we use the techniques of~\cite{DBLP:journals/tcs/CrochemoreIKRRW14} to eliminate runs incompatible with $x$ 
and find the occurrences of $x$ within each compatible run.

\paragraph{\IPM in Texts over Small Alphabets}
The aforementioned techniques allow answering \IPM using $\Oh(n)$ space and $\Oh(n)$ construction time.
In the case of small alphabets, that is, $\sigma = n^{o(1)}$, both complexities can be improved to $\Oh(n/\log_\sigma n)$.
For this, in \cref{sec:packed}, we reduce \IPM in $T\in [0\dd \sigma)^n$ to \IPM in a text $T'$ of length $\Oh(n/\log_\sigma n)$
over an alphabet of size $n^{\Theta(1)}$. 
We use a $\tau$-synchronizing set $\S$, constructed using \cref{prp:synch} for appropriate $\tau=\Theta(\log_\sigma n)$,
to partition $T$ into blocks, and we encode these blocks in $T'$.
By the density property of $\S$, each block has length at most $\tau$ or shortest period at most $\frac13\tau$.
Moreover, the consistency property implies that the block boundaries within $x$ match the block boundaries within any occurrence of $x$.
Consequently, we retrieve the fragments $\Phi(x)$ and $\Phi(y)$ of $T'$, encoding the blocks contained within $x$ and $y$, respectively,
identify the occurrences of $\Phi(x)$ in $\Phi(y)$, and apply \LCEQ (\cref{prop:lce}) through \cref{lem:peralg} below to check which of them correspond to occurrences of $x$ in $y$.
A major challenge in implementing this strategy is that $\Phi(y)$ can be much longer than $\Phi(x)$ despite $|y|<2|x|$.
In particular, $\Phi(x)$ can be empty if $|x|=\Oh(\tau)$ (in that case, we precompute the answers) or if $\per(x)\le \frac13\tau$ (in that case, we reuse the techniques for $\HP$ patterns). In the remaining cases, we show that it suffices to trim $\Phi(y)$ to a carefully defined fragment of length $\Oh(|\Phi(x)|)$.

\paragraph{Applications of \IPM}
\cref{chp:app} demonstrates the usage of \IPM for answering \BQ, \FC, and \BLCP; it also covers applications of all these queries.
The common feature of our solutions is that we make a constant number of \IPM to list candidates (suffixes for \BQ, rotations for \FC, and previous occurrences for \BLCP), and then verify each of them using \LCEQ.
In the non-periodic case, verifying $\Oh(1)$ candidates does not constitute any significant challenge. 
Otherwise, we exploit the structural insight of \cref{fct:single}: the sequence $\mathbf{p}=(p_i)_{i=0}^{k-1}$
of reported starting positions forms a \emph{periodic progression}, meaning that $T[p_0\dd p_1)=\cdots = T[p_{k-2}\dd p_{k-1})$,
and thus, for any position $q$, the answers to queries $\LCE(p_i,q)$ can be obtained in bulk using just $\Oh(1)$ $\LCE$ queries. 
Formally, our main auxiliary result reads as follows:

\begin{restatable}{lemma}{peralg}\label{lem:peralg}
  Consider a text $T$ equipped with a data structure answering \LCEQ in $\Oh(1)$ time.
  Given a fragment $v$ of $T$ and a collection of fragments $u_i=T[p_i\dd r)$ represented with a periodic progression $\mathbf{p}=(p_i)_{i=0}^{k-1}$ and a position $r \ge p_{k-1}$,
  the following queries can be answered in $\Oh(1)$ time:
  \begin{enumerate}[label=(\alph*)]
    \item\label{it:perpref} Report indices $i\in [0\dd k)$ such that $u_i$ matches a prefix of $v$, represented as a subinterval of $[0\dd k)$.
     \item\label{it:permax} Report indices $i\in [0\dd k)$ maximizing $\lcp(u_i,v)$, represented as a subinterval of $[0\dd k)$.
  \end{enumerate}
\end{restatable}

\section{Strings and Periodicity}\label{chp:prelim}

We consider \emph{strings} over an \emph{alphabet} $\Sigma$, i.e., finite sequences of \emph{characters} from the set~$\Sigma$.
The set of all strings over $\Sigma$ is denoted by $\Sigma^*$, and $\Sigma^+$ is the set of non-empty strings over $\Sigma$.
Occasionally, we also work with the family $\Sigma^\infty$ of \emph{infinite strings} indexed by non-negative integers.

\subsection{Basic Notations on Strings}\label{sec:comb}
Let us recall that, for a string $w$, by $|w|$ we denote its length, and by $w[0],\ldots,w[|w|-1]$ its subsequent characters.
By $\varepsilon$ we denote the empty string.
By $\aalph(w)$ we denote the set $\{w[0],\ldots,w[|w|-1]\}$.
The string $w^R = w[|w|-1]\cdots w[0]$ is called the \emph{reverse} of $w$.
A string $u$ is called a \emph{substring} of $w$ if $u=w[i] \cdots w[j-1]$
for some $i,j\in [0\dd |w|]$ with $i\le j$.
In this case, we say that $u$ \emph{occurs} in $w$ at position~$i$, and we denote by $w[i\dd j)=w[i\dd j-1]$
the \emph{occurrence} of $u$ at position $i$.
By $\Occ(u,w)$ we denote the set of starting positions of all occurrences of $u$ in $w$.

We call $w[i\dd j)$ a \emph{fragment} of $w$; formally, a fragment can be interpreted as a tuple
consisting of (a pointer to) the string $w$ and the two endpoints $i,j\in [0\dd |w|]$ with $i\le j$.
If a fragment $w[i\dd j)$ is an occurrence of a string $u$, then we write $u\cong w[i\dd j)$ and say that $w[i\dd j)$ \emph{matches} $u$.
Similarly, if $w[i\dd j)$ and $w[i'\dd j')$ are occurrences of the same string, we denote this by $w[i\dd j)\cong w[i'\dd j')$,
and we say that these fragments match.
On the other hand, the equality of fragments $w[i\dd j)=w[i'\dd j')$ is reserved for occasions when $w[i\dd j)$ is the same fragment as $w[i'\dd j')$ (i.e., $i=i'$ and $j=j'$). 

We assume that a fragment $x=w[i\dd j)$ of a string $w$ inherits some notions from the underlying substring: the \emph{length} $|x|=j-i$, the characters $x[i']=w[i+i']$ for $i'\in [0\dd |x|)$, and the \emph{subfragments} $x[i'\dd j')=w[i+i'\dd i+j')$ for $i',j'\in [0\dd |x|)$ with $i'\le j'$.
A fragment $w[i\dd j)$ also has a natural interpretation as a \emph{range} $[i\dd j)$ of positions in~$w$.
This lets us consider \emph{disjoint} or \emph{intersecting} (\emph{overlapping}) fragments and define the containment relation ($\sub$) on fragments of $w$. 
Moreover, for $i,j,k\in [0\dd |w|]$ with $i\le j \le k$, the fragments $w[i\dd j)$ and $w[j\dd k)$ are called \emph{consecutive} and $w[i\dd k)=w[i\dd j)\cdot w[j\dd k)$ is assumed to be their \emph{concatenation}.
If fragments $w[i\dd j)$ and $w[i'\dd j')$ intersect, we denote their \emph{intersection} $w[\max(i,i')\dd \min(j,j'))$ by $w[i\dd j)\cap w[i'\dd j')$.
 
A fragment $x$ of $w$ of length $|x|<|w|$ is called a \emph{proper} fragment of $w$.
A fragment $w[i\dd j)$ is a \emph{prefix} of $w$ if $i=0$ and a \emph{suffix} of $w$ if $j=|w|$.
We extend the notions of a prefix and a suffix to the underlying substrings.

We denote by $\prec$ the natural order on $\Sigma$
and extend this order in the standard way to the \emph{lexicographic} order on $\Sigma^{*}$.

The notion of concatenation $uv$ extends to $u\in \Sigma^*$ and $v\in \Sigma^\infty$, resulting in $uv\in \Sigma^{\infty}$.
Also, the notion of the longest common prefix of two strings naturally extends to $\Sigma^*\cup \Sigma^\infty$.
For a string $w\in \Sigma^+$, we also introduce the infinite power $w^\infty\in \Sigma^\infty$, i.e., the concatenation of infinitely many copies of $w$.

\subsection{Periodic Structures in Strings}\label{sec:per_struct}
An integer $p\in [1\dd |w|]$ is a \emph{period} of a string $w\in \Sigma^+$ if $w[i]=w[i+p]$ holds for all $i\in  [0\dd |w|-p)$.
We call $w$ \emph{periodic} if its smallest period satisfies $\per(w)\le \tfrac12 |w|$.
A \emph{border} of a string $w$ is a substring of $w$ which occurs both as a prefix and as a suffix of~$w$.
Note that $p$ is a period of $w$ if and only if $w$ has a border of length $|w|-p$.
Periods of a string $w$ satisfy the following periodicity lemma.
\begin{lemma}[Periodicity Lemma~\cite{MR0162838,fine1965uniqueness}]\label{lem:per}
Let $w$ be a string with periods $p$ and $q$.
If $p+q-\gcd(p,q)\le |w|$, then $\gcd(p,q)$ is also a period of $w$.
\end{lemma}

For an integer $k\ge 2$, the string $w^k$ is called a \emph{power} of $w$ (with \emph{root} $w$).
A string $u\in \Sigma^+$ is \emph{primitive} if it is not a power,
i.e., $u\ne w^k$ for every integer $k\ge 2$ and every root~$w$.
For a string $u$, the shortest string $w$ that satisfies $u=w^k$ for some $k \in \mathbb{Z}_{+}$ is called the \emph{primitive root} of $u$.
Primitive strings enjoy a synchronizing property, which is an easy consequence of \cref{lem:per}.

\begin{lemma}[{see~\cite[{Lemma~1.11}]{AlgorithmsOnStrings}}]\label{lem:synchr}
A non-empty string $u$ is primitive if and only if it occurs exactly twice in $u^2$ (as a prefix and as a suffix).
\end{lemma}

By \cref{lem:synchr}, a string $w\in \Sigma^n$ is primitive if and only if it has exactly $n$ distinct rotations.

\paragraph{Maximal Repetitions (Runs)}\label{sec:runs}
A \emph{run} (maximal repetition)~\cite{DBLP:journals/dam/Main89,DBLP:conf/focs/KolpakovK99} in a string $w$ is a periodic fragment
$\gamma=w[i\dd j]$ which can be extended neither to the left nor to the right without increasing the smallest period $p=\per(\gamma)$,
that is, $w[i-1] \ne w[i+p-1]$ and $w[j+1] \ne w[j-p+1]$ provided that the respective positions exist.
We assume that runs are stored together with their periods so that $\per(\gamma)$ can be retrieved in constant time.
We denote the set of all runs in a string $w$ by $\RUNS(w)$.

\begin{figure}[htpb]
\begin{center}
\begin{tikzpicture}[xscale=0.5,yscale=0.5]

\foreach \i/\x in {0/b,1/a,2/a,3/b,4/a,
                   5/b,6/a,7/a,8/b,9/a,
                   10/b,11/b}{
  \draw (\i,0) node[above] {\texttt{\x}};
  \draw (\i,-0.5) node {\scriptsize \i};
}

\foreach \ud in {0,5,9}{
  \begin{scope}[xshift=\ud cm,yshift=0.2cm]
  \clip (0.4,0.5) rectangle (2.6,1);
  \foreach \dx in {-1,0,1,2}{
    \draw[xshift=\dx cm] (0.5, 0.5) sin (1, 0.8) cos (1.5, 0.5);
  }
  \end{scope}
}

\begin{scope}
  \clip (6.4,1.1) rectangle (10.6,1.8);
  \foreach \dx in {-2,0,2,4}{
    \draw[xshift=\dx cm] (6.5, 1.2) sin (7.5, 1.7) cos (8.5, 1.2);
  }
\end{scope}

\begin{scope}[yshift=0.4cm]
  \clip (-0.6,1.4) rectangle (10.6,2.6);
  \foreach \dx in {-5,0,5,10}{
    \draw[xshift=\dx cm] (-0.5, 1.5) sin (2, 2.5) cos (4.5, 1.5);
  }
\end{scope}

\begin{scope}
  \clip (1.4,-1.7) rectangle (6.6,-1);
  \foreach \dx in {-2,0,2,4}{
    \draw[xshift=\dx cm] (1.5, -1) sin (2.5, -1.5) cos (3.5, -1);
  }
\end{scope}

\begin{scope}
  \clip (3.4,-2.6) rectangle (9.6,-1.8);
  \foreach \dx in {-3,0,3,6}{
    \draw[xshift=\dx cm] (3.5, -1.8) sin (5, -2.5) cos (6.5, -1.8);
  }
\end{scope}

\end{tikzpicture}
\end{center}
\vspace*{-0.7cm}
\caption{\label{fig:run}
  Runs in $w=\mathtt{baababaababb}$.
}
\end{figure}

\begin{example}
  For a string $w=\mathtt{baababaababb}$, we have \[\RUNS(w)=\{w[1\dd 3),
  \, w[6\dd 8),\, w[10\dd 12),\, w[2\dd 7),\, w[7\dd 11),\, 
  w[4\dd 10),\, w[0\dd 11)\};\]
  see \cref{fig:run}.
  We have three runs with period 1:
  $w[1\dd 3)\cong \mathtt{aa}$, $w[6\dd 8)\cong \mathtt{aa}$, and $w[10\dd 12)\cong \mathtt{bb}$;
  two runs with period 2: $w[2\dd 7) \cong \mathtt{ababa}$ and $w[7\dd 11)\cong \mathtt{abab}$;
  one run with period 3: $w[4\dd 10)\cong \mathtt{abaaba}$;
  and one run with period 5: $w[0\dd 11)\cong \mathtt{baababaabab}$.
\end{example}

Our results rely on the following asymptotic bounds related to runs.

\begin{proposition}[\cite{DBLP:conf/focs/KolpakovK99,DBLP:journals/siamcomp/BannaiIINTT17}]\label{fct:runs}
Given a text $T$ of length $n$, the set $\RUNS(T)$ of all runs in $T$ (together with their periods) can be computed in $\Oh(n)$ time. In particular, $|\RUNS(T)|=\Oh(n)$.
\end{proposition}

\begin{figure}[th]
\begin{center}
\begin{tikzpicture}[scale=0.7]

\filldraw[fill=black!15] (3,0) rectangle (7,.5);
\draw (1, 0) -- (10, 0);
\draw (1, .5) -- (10, .5);
\draw[dashed] (0,0) -- (1,0);
\draw[dashed] (0,.5) -- (1,.5);
\draw[dashed] (10,0) -- (11,0);
\draw[dashed] (10,.5) -- (11,.5);

\begin{scope}
\clip(3, 0.5) rectangle (7, 00.85);
\foreach \x in {3, 4.6, 6.2} {
	\draw (\x, 0.5) sin (\x+.8, .8) cos (\x + 1.6, 0.5);
}
\end{scope}

\begin{scope}
\clip(1.7, 0.5) rectangle (3, 00.85);
\foreach \x in {1.4} {
	\draw[densely dashed] (\x, 0.5) sin (\x+.8, .8) cos (\x + 1.6, 0.5);
}
\end{scope}
\begin{scope}
\clip(7, 0.5) rectangle (9.7, 0.85);
\foreach \x in {6.2, 7.8, 9.4} {
	\draw[densely dashed] (\x, 0.5) sin (\x+.8, .8) cos (\x + 1.6, 0.5);
}
\end{scope}

\draw (5, 0.25) node {$u$};

\draw[dotted] (1.7, .7) -- (1.7, -.3);
\draw[dotted] (9.7, .7) -- (9.7, -.3);
\draw [decorate,decoration={brace,amplitude=12pt}] (9.7, 0) -- node[below=11pt] {$\gamma$} (1.7, 0);
\draw[dotted] (3, .5) -- (3,1);
\draw[dotted] (4.6, .5) -- (4.6, 1);
\draw[<->] (4.6, .9) -- node[above] {$p$} (3, .9);

\end{tikzpicture}
\end{center}
\vspace*{-0.7cm}
\caption{\label{fig:run-extension}
A run $\gamma$ extending a fragment $u$, that is, $\gamma=\run(u)$, satisfies $\per(u)=\per(\gamma)=p \le \tfrac12|u| \le \tfrac12|\gamma|$.
}
\end{figure}

We say that a run $\gamma$ \emph{extends} a fragment $x$ if $x\sub \gamma$ and $\per(x)=\per(\gamma)$; see \cref{fig:run-extension}.
Every periodic fragment can be extended to a run with the same period. Moreover, the following easy consequence of \cref{lem:per} implies that this extension is unique. For the fixed text $T$, we denote the unique run extending a periodic fragment $x$ by $\run(x)$.
\begin{fact}\label{fct:uni}
Let $\gamma\ne \gamma'$ be runs in a string $w$.
If $p=\per(\gamma)$ and $p'=\per(\gamma')$,
then $|\gamma \cap \gamma'| < p+p'-\gcd(p,p')$.
\end{fact}
\begin{proof}
For a proof by contradiction, suppose that  $|\gamma \cap \gamma'| \ge p+p'-\gcd(p,p')$.
By \cref{lem:per}, this means that $\gcd(p,p')$ is a period of the intersection $\gamma\cap \gamma'$, which we denote $w[\ell\dd r]$.
Since $\gamma\ne \gamma'$, exactly one of these runs must contain position $\ell-1$ or $r+1$.
Due to symmetry, we may assume without loss of generality that $\gamma$ contains position $\ell-1$.
Observe that positions $\ell+p-1$ and $\ell+p'-1$ are located within $\gamma\cap \gamma'$, so $w[\ell+p-1]=w[\ell+p'-1]$,
because $\gcd(p,p')$ divides $|p-p'|$.

On the other hand, $w[\ell-1]=w[\ell+p-1]$ (because $\gamma$ contains position $\ell-1$) and $w[\ell-1]\ne w[\ell+p'-1]$ (by maximality of $\gamma'$).
Thus, $w[\ell-1]=w[\ell+p-1]=w[\ell+p'-1]\ne w[\ell-1]$, which is a contradiction that concludes the proof.
\end{proof}

Fragments $x$ satisfying $\run(x)=\gamma$ admit the following elegant characterization.

\begin{observation}\label{obs:char}
Consider a text $T$ and a run $\gamma \in \RUNS(T)$.
A fragment $x$ of $T$ satisfies $\run(x)=\gamma$ if and only if $x$ is contained in $\gamma$
and $|x|\ge 2\per(\gamma)$.
\end{observation}

A string $x$ is called \emph{highly periodic} if $\per(x) \le \frac13|x|$; otherwise, it is called \emph{non-highly-periodic}.
Highly periodic and non-highly-periodic strings are further denoted as $\HP$-strings and $\N$-strings, respectively.

\section{Synchronizing Sets Hierarchy}\label{chp:lce}
Recall that a $\tau$-\emph{synchronizing set} consists of the starting positions of selected length-$2\tau$ fragments.
It is defined as follows to satisfy consistency and density conditions.

\begin{definition}[\textbf{Synchronizing Set}; Kempa and Kociumaka~\cite{Kempa2019}]
  Let $T$ be a string of length $n$ and let $\tau\in [1\dd \floor{\frac{n}{2}}]$.
  A set $\S\sub [0\dd n-2\tau]$ is a \emph{$\tau$-synchronizing set} of $T$ if it satisfies the following conditions:
  \begin{description}
\item[Consistency:] For all $i,j\in [0\dd n-2\tau]$, if $i\in \S$ and $T[i\dd i+2\tau)\cong T[j\dd j+2\tau)$, then $j\in \S$.
\item[Density:]
For every $i\in [0\dd n-3\tau+1]$, we have $[i\dd i+\tau)\cap \S = \emptyset\;\Longleftrightarrow\; \per(T[i\dd i+3\tau-1)) \le \tfrac13\tau$.
\end{description}
We say that the elements of a fixed $\tau$-synchronizing set $\S$ are $\tau$-\emph{synchronizing positions} and the fragments $T[s\dd s+2\tau)$ for $s\in \S$ are $\tau$-\emph{synchronizing fragments}; the set of $\tau$-synchronizing fragments is denoted~$\SF$.
\end{definition}

\begin{example}
  Let  $\TM_t$ be the Thue--Morse word~\cite{Thue2} of length $2^t$ and $\overline{\TM_t}$ be its bitwise negation.
  For $i\in [0\dd 2^t)$, the character $\TM_t[i]$ is the \emph{pop-count} of $i$ modulo 2 (the parity of ones in the binary representation of $i$).
  For $k\in [0\dd t)$, the following construction yields a $2^{k}$-synchronizing set of $\TM_t$:
  \[\S=\left(\Occ(\TM_k,\TM_t)\cup \Occ(\overline{\TM_k},\TM_t)\right)\cap \left[0\dd 2^t-2^{k+1}\right].\]
  We illustrate the case of $t=5$ and $k=2$, where $\TM_{2}=0110$.
  Then,  $\S\,=\, \{0,4,6,8,12,16,20,22,24\}$ with the $2^2$-synchronizing positions underlined:
  $\TM_5\,=\,
  \underline{0}\,1\,1\,0\,\underline{1}\,0\,\underline{0}\,1\,\underline{1}\,0\,0\,1\,
  \underline{0}\,1\,1\,0\,\underline{1}\,0\,0\,1\,\underline{0}\,1\,\underline{1}\,0\,
  \underline{0}\,1\,1\,0\,1\,0\,0\,1$.
  \end{example}

We say that a fragment is \emph{$p$-periodic} if its smallest period is at most $p$; otherwise
we say it is \emph{$p$-non-periodic}.
Let us note that if $[i\dd i+\tau)\cap \S=\emptyset$ holds for a $\tau$-synchronizing set $\S$,
then the density condition stipulates that $T[i\dd i+3\tau-1)$ is $\frac13\tau$-periodic.
Observe that the starting positions of all the $\frac13\tau$-non-periodic length-$2\tau$ fragments of $T$
form a $\tau$-synchronizing set. However, this set is too large to be useful. 
Kempa and Kociumaka~\cite{Kempa2019} showed how to efficiently construct $\tau$-synchronizing sets
of optimal size $\Oh(n/\tau)$.

\prpsynch*

We present a simultaneous linear-time construction of $\tau$-synchronizing sets for
a whole \emph{hierarchy} of geometrically increasing $\tau$'s.
More formally, we show that, after linear-time preprocessing, a 
$\tau$-synchronizing set for a given $\tau$ can be constructed just in $\Oh(n/\tau)$ time.

\thmsss*

In this section, we give an overview of the construction behind \cref{thm:sss}.
As indicated in \cref{sec:techniques}, we use a modification of the recompression technique~\cite{DBLP:journals/jacm/Jez16} to construct a sequence $(\F_0,\F_1,\dots,\F_q)$ of factorizations of $T$, where
$\F_0$ consists of single letters of $T$, phrases of $\F_k$ are concatenations of phrases of $\F_{k-1}$ for $k \in [1\dd q]$,
and $\F_q$ contains one phrase spanning across entire $T$.
The factorizations are represented by the sets $\B_k$ of \emph{phrase boundaries}: the starting positions of phrases of $\F_k$ except for the leftmost phrase.

\paragraph*{\bf Local Consistency}
The local consistency of recompression is characterized as follows: whether or not two subsequent phrases of $\F_{k-1}$ are concatenated into the same phrase of $\F_k$ depends solely on the \emph{names} of these two phrases (where matching phrases have equal names). 
Unfortunately, since the phrases can get arbitrarily long, we cannot conclude that the resulting set $\B_k$ of phrase boundaries is locally consistent:
for any fixed $k>1$, whether or not a given $i$ position belongs to $\B_k$ may depend on a context of unbounded size.
Consequently, in \cref{sec:recompression}, we develop \emph{restricted} recompression, where two subsequent phrases of $\F_{k}$ may be concatenated into the same phrase of $\F_{k+1}$ only if their lengths do not exceed
\[d_k = \osiemsiedem^{\floor{k/2}}.\]
As a result, every phrase of $\F_k$ has length at most $\frac74 d_k$ or primitive root of length at most $d_k$ (see \cref{fct:recompr}).
Moreover, this guarantees that whether or not a given $i$ position belongs to $\B_k$ depends only on its context of length $\alpha_k \le 16d_{k-1}$.
Formally, we define
\[\alpha_k = \begin{cases} 1 & \text{if }k=0,\\
   \alpha_{k-1}+\floor{d_{k-1}} & \text{otherwise},\end{cases}\]
so that \cref{lem:cons} can be proved in \cref{sec:recompression}.
Our construction further guarantees that $|\B_k| = \Oh(n/d_k)$ and hence $q=\Oh(\log n)$.

\begin{figure}[h]
\centering
\begin{tikzpicture}[xscale=0.6,yscale=0.6,font=\small]
  \def\tauW{3.0}
  \def\boxH{0.4}

  \def\curr{0}
  \foreach [count=\i]
    \w/\c
    in {2/A, 4/B, 1/C, 6/D, 1/C, 1/C, 6/D, 4/E, 1/F} {
    \coordinate (l\i) at (\curr,0);
    \coordinate (b\i) at (\curr+\w,0);
    \coordinate (r\i) at (\curr+\w,0.8);
    \draw (l\i) rectangle (r\i) node[midway] {$\c$};
    \xdef\curr{\number\numexpr\curr+\w\relax}
  }
  \foreach \i in {1,...,8} {
    \node at (r\i) [above] {$b_\i$};
  }

  \foreach \i/\y in {6/3, 3/2, 5/2, 2/1, 4/1, 7/1} {
    \coordinate (bl\i) at ($(b\i)+(-\tauW,-\y)$);
    \coordinate (bll\i) at ($(b\i)+(-\tauW+0.2,-\y+0.2)$);
    \coordinate (br\i) at ($(b\i)+(\tauW,-\y)$);
    \coordinate (bc\i) at ($(b\i)+(0,-\y)$);
    \coordinate (bu\i) at ($(b\i)+(0,-\y+\boxH)$);

    \draw[densely dotted, violet] (bu\i) -- (l\i -| bu\i);
    \draw[thin, blue, -latex] (bll\i) -- (l\i -| bll\i);

    \draw[fill=white] (bl\i) rectangle ($(br\i)+(0,\boxH)$);
    \draw[very thick, violet] (bc\i) -- (bu\i);
    \node[blue] at (bll\i) {$\bullet$};
    \node[blue, above] at (bll\i |- r\i) {$s_\i$};
  }

  \draw[|<->|] ([yshift=-0.2cm]bl6)--([yshift=-0.2cm]bc6) node[midway,below] {$\tau$};
  \draw[|<->|] ([yshift=-0.2cm]bc6)--([yshift=-0.2cm]br6) node[midway,below] {$\tau$};
\end{tikzpicture}
\caption{Assume $\F_k=(A,B,C,D,C,C,D,E,F)$, where the symbols $A,B,C,D,E,F$ are names of phrases.
The set $\B_k$ consists of starting positions of phrases.
The $\tau$-synchronizing set results by 
going back from the start of each phrase by the distance $\tau$ (provided that a length-$2\tau$ fragment starts there).
The set $\S$ consists of positions in $T$ marked by a blue dot. 
The figure shows the case without highly periodic fragments.
}\label{fig:ss}
\end{figure}

Since the density condition in \cref{def:sss} depends on the notion of $\frac13\tau$-periodicity,
our construction needs to capture it as well.
For an integer $\tau\in [1\dd n]$, we define the set of \emph{$\tau$-runs} in $T$ as 
\[\RUNS_{\tau}(T) = \{\gamma \in \RUNS(T) : |\gamma|\ge \tau \text{ and } \per(\gamma)\le \tfrac13 \tau\}.\]

\begin{figure}[h]
\centering
\begin{tikzpicture}[scale=0.7,font=\small]
  \def\tauW{2.0}
  \def\boxH{0.4}

  \coordinate (s) at (0,0);
  \draw (s) rectangle +(10, 0.5);

  \begin{scope}[xshift=2cm,yshift=0.5cm]
    \clip (0,0) rectangle (6.6,0.5);
    \foreach \dx in {0,...,10}{
      \coordinate (c\dx) at (\dx,0);
      \draw (\dx, 0) sin (\dx+0.5, 0.3) cos (\dx+1, 0);
    }
  \end{scope}

  \coordinate (r2) at ($(c6)+(0.75-2*\tauW,-2.5)$);
  \coordinate (rr2) at ($(r2)+(2*\tauW,\boxH)$);
  \draw[thick,fill=white] (r2) rectangle (rr2);
  \node[blue] (s2) at ($(r2)+(0.2,0.45*\boxH)$) {$\bullet$};
  \draw[densely dotted] (rr2) -- (rr2 |- s);
  \draw[thin, blue, -latex] (s2) -- (s2 |- s);

  \coordinate (r1) at ($(c0)+(-0.25,-1.5)$);
  \coordinate (rr1) at ($(r1)+(2*\tauW,\boxH)$);
  \draw[thick,fill=white] (r1) rectangle (rr1);
  \node[blue] (s1) at ($(r1)+(0.2,0.45*\boxH)$) {$\bullet$};
  \draw[thin, blue, -latex] (s1) -- (s1 |- s);

  \draw[|<->|] ([yshift=-0.2cm]r1 -| rr1)--([yshift=-0.2cm]r1) node[midway,below] {$2\tau$};
  \draw[|<->|] ([yshift=-0.2cm]r2 -| rr2)--([yshift=-0.2cm]r2) node[midway,below] {$2\tau$};

\end{tikzpicture}
\caption{Synchronizing positions and synchronizing fragments induced by a $\tau$-run.
The first/last positions of such fragments are one position to the left/right of a $\tau$-run. The set $\S$ includes the first positions
of synchronizing fragments (blue dots).}\label{Sept23}
\end{figure}

The \emph{middle point} of a fragment $T[i \dd i+2\ell)$ is defined as $i+\ell$. 
We define our $\tau$-synchronizing set based on the set $\B_k$ of phrase boundaries at the lowest level 
$k$ such that $\tau \ge 16d_{k-1} \ge \alpha_k$. Formally, we define \[k(\tau)=\max\{j \in \Zz\,:\, j=0 \text{ or }16d_{j-1} \le \tau\}.\]

\begin{construction}[Sets of synchronizing positions and fragments]\label{def:syncfr}
The set $\SF$ of $\tau$-synchronizing fragments consists of all $\frac{\tau}{3}$-non-periodic 
fragments of length $2\tau$ such that 
\begin{enumerate}[label=(\alph*)]
\item\label{it:a} their middle point is in $\B_{k(\tau)}$, or
\item\label{it:b} their first position is one position to the left of a $\tau$-run, or
\item\label{it:c} their last position is one position to the right of a $\tau$-run.
\end{enumerate}
The set $\S$ consists of the first positions
of fragments in $\SF$.
\end{construction}

Schematic illustrations can be found on \cref{fig:ss,Sept23,fig:general}.
The procedure constructing the sets $\B_k$ is described and analyzed in \cref{sec:recompression}.
Then, in \cref{sec:correctness}, we prove that \cref{def:syncfr} indeed yields $\tau$-synchronizing sets
and that these sets can be constructed efficiently.

\begin{figure}[h]
\centering
\vspace*{-0.5cm}
\begin{tikzpicture}[xscale=0.7,yscale=0.7]
  \def\tauW{1.8}
  \def\perW{1.2}
  \def\boxH{0.4}

  \def\curr{0}
  \foreach [count=\i]
    \w/\c
    in {1/A, 1/B, 2/C, 3/D, 3/C, 5/C, 5/D, 1/E} {
    \coordinate (l\i) at (\curr,0);
    \coordinate (b\i) at (\curr+\w,0);
    \coordinate (r\i) at (\curr+\w,0.8);
    \draw (l\i) rectangle (r\i) node[midway] {};
    \xdef\curr{\number\numexpr\curr+\w\relax}
  }
  \foreach \i in {1,...,7} {
    \node at (r\i) [above,yshift=0.1cm] {$b_\i$};
  }

  \begin{scope}[xshift=7.6cm,yshift=0.8cm]
    \clip (-\perW,0) rectangle (10,2);
    \foreach \dx in {0,...,10}{
      \coordinate (c\dx) at (\dx*\perW,0);
      \draw (\dx*\perW, 0) sin ((\dx*\perW+0.5*\perW, 0.3) cos ((\dx*\perW+1*\perW, 0);
    }
    \coordinate (b10) at ($(c0)+(-0.4+\tauW,-0.8)$);
    \coordinate (b11) at ($(c8)+(+0.5-\tauW,-0.8)$);
  \end{scope}

  \foreach \i/\y in {3/2, 2/1, 4/3} {
    \coordinate (bl\i) at ($(b\i)+(-\tauW,-\y)$);
    \coordinate (bll\i) at ($(b\i)+(-\tauW+0.2,-\y+0.2)$);
    \coordinate (br\i) at ($(b\i)+(\tauW,-\y)$);
    \coordinate (bc\i) at ($(b\i)+(0,-\y)$);
    \coordinate (bu\i) at ($(b\i)+(0,-\y+\boxH)$);

    \draw[densely dotted, violet] (bu\i) -- (l\i -| bu\i);
    \draw[thin, blue, -latex] (bll\i) -- (l\i -| bll\i);

    \draw[fill=white] (bl\i) rectangle ($(br\i)+(0,\boxH)$);
    \draw[very thick, violet] (bc\i) -- (bu\i);
    \node[blue] at (bll\i) {$\bullet$};
  }

  \foreach \i/\y in {10/1, 11/2} {
    \coordinate (bl\i) at ($(b\i)+(-\tauW,-\y)$);
    \coordinate (bll\i) at ($(b\i)+(-\tauW+0.2,-\y+0.2)$);
    \coordinate (br\i) at ($(b\i)+(\tauW,-\y)$);
    \coordinate (bc\i) at ($(b\i)+(0,-\y)$);
    \coordinate (bu\i) at ($(b\i)+(0,-\y+\boxH)$);

    \draw[thin, blue, -latex] (bll\i) -- (l1 -| bll\i);

    \draw[fill=white] (bl\i) rectangle ($(br\i)+(0,\boxH)$);
    \draw[very thick, violet] (bc\i) -- (bu\i);
    \node[blue] at (bll\i) {$\bullet$};
  }
  \draw[densely dotted] (br11)--($(br11 |- r1)+(0,0.2)$);

  \node[blue, above] at (bll2 |- b1) {$s_1$};
  \node[blue, above] at (bll3 |- b1) {$s_2$};
  \node[blue, above] at (bll4 |- b1) {$s_3$};
  \node[blue, above] at (bll10 |- b1) {$s_4$};
  \node[blue, above] at (bll11 |- b1) {$s_5$};

  \draw[|<->|] ([yshift=-0.2cm]bl4)--([yshift=-0.2cm]bc4) node[midway,below] {$\tau$};
  \draw[|<->|] ([yshift=-0.2cm]bc4)--([yshift=-0.2cm]br4) node[midway,below] {$\tau$};

\end{tikzpicture}
\caption{
  Illustration of synchronizing fragments of length $2\tau$ in a general situation: the non-periodic case and the case of synchronizing fragments generated by $\tau$-runs.
  The set of $\tau$-synchronizing positions is $\S=\{s_1,\ldots,s_5\}$, and the set 
  of phrase boundaries is $\B_k=\{b_1,b_2,\ldots,b_7\}$.
}\label{fig:general}
\end{figure}

\section{Restricted Recompression: Hierarchy of Phrase Boundaries}\label{sec:recompression}

\newcommand{\Symb}{\mathbf{S}}
\newcommand{\Act}{\mathbf{A}}
\newcommand{\Left}{\mathbf{L}}
\newcommand{\Right}{\mathbf{R}}
\newcommand{\rle}{\mathsf{\mathbf{RunShrink}}}
\newcommand{\pc}{\mathsf{\mathbf{PairShrink}}}
\newcommand{\Zp}{\mathbb{Z}_{+}}

In this subsection, we introduce and analyze a version of the recompression technique~\cite{DBLP:journals/jacm/Jez16}, called here \emph{restricted recompression}.

\subsection{Definition of Restricted Recompression}
Given a string $T\in \Sigma^+$, we create a sequence of factorizations $(\F_0,\ldots,\F_q)$,
where each factorization $\F_k$ decomposes $T$ into $m_k$ \emph{phrases} $T[f_{k,i}\dd f_{k,i+1})$ for $i\in [0\dd m_k)$.
Phrases are not allowed to grow too fast, so our version of recompression is called \emph{restricted}.

We identify the factorization $\F_k$ with a string $T_k\in \Symb^+$ of phrase \emph{names} such that $T_k[i]=T_k[j]$ holds if and only if $T[f_{k,i}\dd f_{k,i+1})\cong T[f_{k,j}\dd f_{k,j+1})$. The alphabet $\Symb$ of \emph{symbols} is defined as the least fixed point of the following equation:
 \[\Symb = \Sigma \cup (\Symb \times \Symb)\cup (\Symb \times \mathbb{Z}_{\ge 2} ).\]
The phrase names  can be converted into strings using an \emph{expansion function} $\val: \Symb \to \Sigma^+$:
\[\val(S) = \begin{cases}
  S & \text{if $S\in \Sigma$},\\
  \val(S_1)\cdot \val(S_2) & \text{if $S=(S_1,S_2)$ for $S_1,S_2\in \Symb$},\\
  \val(S')^m & \text{if $S=(S',m)$ for $S'\in \Symb$ and $m\in \mathbb{Z}_{\ge 2}$},
\end{cases}\]
that is extended to a \emph{morphism} $\val : \Symb^*\to \Sigma^*$ by setting $\val(S_1\cdots S_s)=\val(S_1)\cdots \val(S_s)$ for $S_1,\ldots,S_s\in \Symb$.
Hence, $\val(T_k[j]) \cong T[f_{k,j}\dd f_{k,j+1})$ and $\val(T_k)=T$.

The sets of \emph{phrase boundaries} corresponding to a factorization $\F_k$ can be formally defined as
$\B_k=\{f_{k,1},\ldots,f_{k,m_k-1}\}$.
We have $\B_k \subseteq \B_{k-1}$ for all $k \in [1 \dd q]$.
Note that if $\F_k$ consists of just a single phrase, then $\B_k=\emptyset$; see \cref{fig:schematic}.

\begin{figure}[h]
  \centering
  \begin{tikzpicture}[scale=0.7]

    \def\smallTriangle#1{
      \draw (#1)--+(-0.4,-0.5)--+(+0.4,-0.5)--cycle;
    }
    \def\smallEdge#1{
      \draw (#1)--+(-0,-0.5);
    }

    \def\leftTree#1#2#3#4{
        \begin{scope}[shift={#2}]
        \node[above] at (0,0) {#1};
        \draw (0,0)--(-0.5,-0.5);
        \draw (0,0)--(0.5,-0.5);
        \ifx&#3&%
            \smallEdge{-0.5,-0.5}  
        \else
            \smallTriangle{-0.5,-0.5}
        \fi
        \ifx&#4&%
            \smallEdge{0.5,-0.5}  
        \else
            \smallTriangle{0.5,-0.5}
        \fi
        \end{scope}
    }
    
    \draw[very thick,-latex, blue] (5,-1.5)--+(0,-1);

    \begin{scope}
        \foreach \x in {0.85,3.2,5,7.2,9} {
        \draw[thick,violet] (\x,0)--+(0,-1.2);
        }
        \leftTree{$S_1$}{(0,0)}{x}{}
        \leftTree{$S_2$}{(2,0)}{x}{x}
        \leftTree{$S_3$}{(4,0)}{}{x}
        \leftTree{$S_4$}{(6,0)}{x}{x}
        \leftTree{$S_5$}{(8,0)}{}{x}
        \leftTree{$S_6$}{(10,0)}{x}{x}
    \end{scope}
    
    \begin{scope}[xshift=1cm,yshift=-3cm]
        \foreach \x in {2.1, 3.85, 8.1} {
        \draw[thick,violet] (\x,0)--+(0,-1.7);
        }
        
        \begin{scope}[xshift=0cm]
        \node[above] at (0,0) {$(S_1,S_2)$};
        \draw(0,0)--(-1,-0.5);
        \draw(0,0)--(1,-0.5);
        \leftTree{}{(-1,-0.5)}{x}{}
        \leftTree{}{(1,-0.5)}{x}{x}
        \end{scope}
        
        \begin{scope}[xshift=2.8cm]
        \node[above] at (0,0) {$S_3$};
        \draw(0,0)--(0,-0.5);
        \leftTree{}{(0,-0.5)}{}{x}
        \end{scope}
        
        \begin{scope}[xshift=6cm]
        \node[above] at (0,0) {$(S_4,S_5)$};
        \draw(0,0)--(-1,-0.5);
        \draw(0,0)--(1,-0.5);
        \leftTree{}{(-1,-0.5)}{x}{x}
        \leftTree{}{(1,-0.5)}{}{x}
        \end{scope}
        
        \begin{scope}[xshift=9.2cm]
        \node[above] at (0,0) {$S_6$};
        \draw(0,0)--(0,-0.5);
        \leftTree{}{(0,-0.5)}{x}{x}
        \end{scope}
    
    \end{scope}
    
\end{tikzpicture}
  \caption{
  Schematic view of one iteration of restricted recompression
  (with phrase boundaries indicated).}\label{fig:schematic}
\end{figure}

Next, we describe operations being the basic building blocks of restricted recompression.
They are modifications of operations used in standard recompression~\cite{DBLP:journals/jacm/Jez16}.

\begin{definition}[{\bf Restricted run-length encoding}]\label{it:rle}Given a string $U\in \Symb^*$ and a subset $\Act \sub \Symb$, we define an operation $\rle_{\Act}(U)$ that returns a string in $\Symb^*$ as shown in \cref{algo:RLE}.

\begin{minipage}[t]{.93\textwidth}\vspace{0pt}
\begin{algorithm}[H]
\ForEach{maximal unary run $U[i\dd i+m)\cong A^m$ \KwSty{in} $U$ with $m\ge 2$ and $A\in \Act$}{replace $U[i\dd i+m)$ with $(A,m) \in \Symb$;}
\Return{$U$}\;
\caption{$\rle_{\Act}(U)$}\label{algo:RLE}
\end{algorithm}
\end{minipage}
\end{definition}

\begin{definition}[{\bf Restricted pair compression}]\label{it:pair}
Given a string $U\in \Symb^*$ and disjoint subsets $\Left,\Right \sub \Symb$, we define an operation
$\pc_{\Left,\Right}(U)$ that returns a string in $\Symb^*$ as shown in \cref{algo:PC}.

\begin{minipage}[t]{.93\textwidth}\vspace{0pt}
\begin{algorithm}[H]
\ForEach{fragment $U[i\dd i+1]$ \KwSty{in} $U$ with $U[i]\in \Left$ and $U[i+1]\in \Right$}{
replace $U[i\dd i+1]$ with $(U[i],U[i+1]) \in \Symb$\;
}
\Return{$U$}\;
\caption{$\pc_{\Left,\Right}(U)$}\label{algo:PC}
\end{algorithm}
\end{minipage}
\end{definition}

\begin{figure}[h]
\centering
\begin{center}
\begin{tikzpicture}[xscale=0.6,yscale=0.5]
\foreach \x in {0,1,2.5,4,7,9,10.5,11.5,12.5,13.5,16.5,19.5,21,22.5,24,26}{
  \draw[thick] (\x,0) -- (\x,1);
}
\foreach \x/\j in {1/1,2.5/2,4/3,7/4,9/5,10.5/6,11.5/7,12.5/8,13.5/9,16.5/10,19.5/11,21/12,22.5/13,24/14}{
  \draw[xshift=0.15cm] (\x,0) node[below] {$b_{\j}$};
  \filldraw[xshift=0.15cm] (\x,0.5) circle (0.07cm);
}
\foreach \x/\c in {0.5/D,1.75/A,3.25/B,5.5/E,8/C,9.75/B,11/D,12/D,13/D,15/E,18/E,20.25/A,21.75/B,23.25/A,25/C}{
  \draw (\x,0.5) node {$\c$};
}
\draw[thick] (0,0) -- (26,0)  (0,1) -- (26,1);

\draw[very thick,-latex] (-0.5,0.5) -- (-0.5,4);

\begin{scope}[yshift=3.5cm]
\foreach \x in {0,1,4,7,10.5,13.5,16.5,19.5,22.5,24,26}{
  \draw[thick] (\x,0) -- (\x,1);
}
\foreach \x/\j in {1/1,4/3,7/4,10.5/6,13.5/9,16.5/10,19.5/11,22.5/13,24/14}{
  \draw[xshift=0.15cm] (\x,0) node[below] {$b_{\j}$};
  \filldraw[xshift=0.15cm] (\x,0.5) circle (0.07cm);
}
\foreach \x/\c in {0.5/D,5.5/E,15/E,18/E,23.25/A,25/C}{
  \draw (\x,0.5) node {$\c$};
}
\draw (2.5,0.5) node {$(A,B)$};
\draw (8.875,0.5) node {$(C,B)$};
\draw (12,0.5) node {$(D,3)$};
\draw (21,0.5) node {$(A,B)$};
\draw[thick] (0,0) -- (26,0)  (0,1) -- (26,1);
\end{scope}
\end{tikzpicture}
\end{center}
\caption{Detailed illustration of one iteration of computing the
restricted recompression,
moving from $T_{k}\,=\,D\,A\,B\,E\,C\,B\,D\,D\,D\,E\,E\,A\,B\,A\,C$
to $T_{k+2}\,=\,D\,(A,B)\,E\,(C,B)\,(D,3)\,E\,E\,(A,B)\,A\,C$.
The dots correspond to phrase boundaries.
The set of phrase boundaries is changed as follows: $\B_{k}=\{b_1,b_2,\ldots,b_{14}\}\ \longrightarrow\
\B_{k+2}=\{b_1,b_3,b_4,b_6,b_{9},b_{10},b_{11},b_{13},b_{14}\}$.
We have $\Left_{k}=\{A,C\}$ and $\Right_{k}=\{B\}$. Further, $|\val(A)|,\ldots,|\val(D)|\le d_k < |\val(E)|$, so $EE$ is not shrunk.}\label{fig:parse}
\end{figure}

\begin{definition}[{\bf Restricted recompression}]\label{constr:Jez}
Given a string $T\in \Sigma^+$, the strings $T_k$ for $k\in \Zz$ are constructed as shown in \cref{algo:RECOMPR}; see also \cref{fig:parse}.

\begin{minipage}[t]{.93\textwidth}\vspace{0pt}
\begin{algorithm}[H]
$T_0:=T$;\ \ $k:=0$\;
\While{$|T_k|>1$}{
  $\Act_k := \{S\in \Symb : |\val(S)| \le d_k\}$\;
  $T_{k+1}:=\rle_{\Act_{k}}(T_{k})$\;
  $k:=k+1$\;
  $\Act_{k} := \{S\in \Symb : |\val(S)| \le d_{k}\}$\;
  $T_{k+1}:=\pc_{\Left_{k},\Right_{k}}(T_{k})$, where $\Left_{k},\Right_{k}$ are disjoint subsets of $\Act_k$ specified later\;
  $k:=k+1$\;
}
$q:=k$;
\caption{Constructing  restricted recompression
}
\label{algo:RECOMPR}\mbox{ \ }\\
\end{algorithm}
\end{minipage}
\end{definition}

\subsection{Efficient Implementation of Restricted Recompression}
Before we proceed with an efficient construction algorithm, let us show two basic properties of restricted recompression.
We view each application of a shrinking method ($\rle$ or $\pc$) in \cref{algo:RECOMPR} as a decomposition of $T_k$ into fragments, called \emph{blocks},
such that single-character blocks stay intact, and longer blocks are \emph{collapsed} into single characters in $T_{k+1}$.
We refer to \emph{block boundaries} as positions of $T_k$ where blocks start.

A distinctive feature of recompression (compared to other locally consistent parsing techniques)
is that matching phrases are given equal names.
More generally, the following property is shown by induction.

\begin{lemma}\label{fct:cons}
For every $k\in \Zz$ and fragments $x,x'$ of $T_k$, if the fragments expand to equal strings, that is $\val(x)=\val(x')$,
then $x$ and $x'$ are formed of the same sequences of phrase names, that is $x \cong x'$.
\end{lemma} 
\begin{proof}
We proceed by induction on $k$. 
Let $x,x'$ be fragments of $T_{k}$ satisfying $\val(x)=\val(x')$.
If $k=0$, then $x\cong x'$ holds due to $\val(x)\cong x$ and $\val(x')\cong x'$.
Otherwise, let $y$ and $y'$ be the fragments of $T_{k-1}$ obtained from $x$ and $x'$, respectively, by expanding collapsed blocks.

Note that $\val(y)=\val(x)=\val(x')=\val(y')$, so the inductive assumption guarantees $y\cong y'$.
Inspecting \cref{it:rle,it:pair}, we can observe that if $T_{k-1}[i-1\dd i]\cong T_{k-1}[i'-1\dd i']$ for $i,i'\in [1\dd |T_{k-1}|)$,
then block boundaries at positions $i$ and $i'$ are placed consistently, that is,
either at both of them or at neither of them.
Consequently, block boundaries within $y$ and $y'$ are placed consistently.
Moreover, both $y$ and $y'$ consist of full blocks (since they are collapsed to $x$ and $x'$, respectively). Thus, $y$ and $y'$ are consistently partitioned into blocks. 
Matching blocks get collapsed to matching symbols both in \cref{it:rle,it:pair}, so we derive $x\cong x'$.
\end{proof}

In particular, we conclude that $T_{k}=\rle_{\Act_{k}}(T_{k})$ holds for all odd $k\in \Zz$.
\begin{corollary}\label{cor:distinct}
  For every odd $k\in \Zz$, there is no $j\in [1\dd |T_k|)$ such that $T_k[j-1]=T_k[j]\in \Act_{k}$.
  \end{corollary}
  \begin{proof}
    For a proof by contradiction, suppose that $T_k[j-1]=T_k[j]\in  \Act_{k}$ holds for some $j\in [1\dd |T_k|)$.
    By the definition of $d_k=\osiemsiedem^{\floor{k/2}}$, we have $\Act_k=\Act_{k-1}$.
  Let \[x=T_{k-1}[i-\ell\dd i)\quad \text{and}\quad x'=T_{k-1}[i\dd i+\ell')\] be 
blocks of $T_{k-1}$ collapsed to $T_k[j-1]$ and $T_k[j]$, respectively. 

Due to $\val(x)=\val(T_k[j-1])=\val(T_k[j])=\val(x')$, \cref{fct:cons} guarantees $x \cong x'$ and, in particular, $\ell=\ell'$.
    If $\ell=1$, then \[T_{k-1}[i-1]=T_k[j-1]=T_k[j]=T_{k-1}[i] \in \Act_{k}=\Act_{k-1}.\]
Otherwise, $x\cong x' \cong A^\ell$ for some symbol $A\in \Act_{k-1}$,
    which means that $T_{k-1}[i-1]=T_{k-1}[i]=A\in \Act_{k-1}$.

In either case, $T_{k-1}[i-1]=T_{k-1}[i]\in \Act_{k-1}$, which means that $\rle_{\Act_{k-1}}(T_{k-1})$ does not place a block
    boundary at position $i$ in $T_{k-1}$, a contradiction.
  \end{proof}

The classification into left and right symbols is made similarly as in~\cite[Lemma 6.2]{DBLP:journals/talg/Jez15}.
We formulate an auxiliary problem and use its folklore deterministic linear-time solution employing the so-called method of conditional expectations~\cite[Section 6.3]{MU05}.
A proof of the following lemma is provided in \cref{app:AMDC} for completeness.

\defdsproblem{\textsc{Approximate Maximum Directed Cut}}{
\textbf{Input}: A directed multigraph $G=(V,E)$ without self-loops.\\
\textbf{Output}:
A partition $V = L\cup R$ such that at least $\tfrac14|E|$ arcs lead from $L$ to $R$.
}

\begin{restatable}{lemma}{AMDC}
\label{lem:maxcut}
The \textsc{Approximate Maximum Directed Cut} problem can be solved in $\Oh(|V|+|E|)$ time.
\end{restatable}

We are ready to provide a deterministic algorithm for computing phrase boundaries of restricted recompression, shown in the proposition below.
It implies, in particular, that $\B_k$ is empty for appropriately large $k=\Theta(\log n)$, so the number of iterations $q$
of the restricted recompression is $\Oh(\log n)$.

In the implementation of operation $\pc_{\Left_{k},\Right_{k}}$ in restricted recompression, we construct a multigraph with vertices
$\aalph(T_k)\cap \Act_k$, such that there is an arc from $T_k[i-1]$ to $T_k[i]$ for $i \in [1 \dd |T_k|)$ if both symbols belong to $\Act_k$, i.e., both $\val(T_k[i-1])$ and $\val(T_k[i])$ have lengths at most $d_k$.
The sets $\Left_k, \Right_k$ are the output $L,R$ of the \textsc{Approximate Maximum Directed Cut} on this multigraph.
The number of arcs from $L$ to $R$ is exactly the number of pairs of collapsed blocks; the fact that it is bounded from below
guarantees that $|\B_k|$ decreases exponentially.
The steps using $\rle$ do not need to reduce the number of phrases, but they guarantee that there are no self-loops in the constructed multigraphs
(cf.\ \cref{cor:distinct}).

\begin{proposition}\label{lem:recompr2}
  Given a string $T\in \Sigma^n$, a family
  of phrase boundaries $(\B_k)_{k=0}^q$ representing a restricted recompression $(T_k)_{k=0}^q$ can be constructed in $\Oh(n)$ time,
  with an extra guarantee that $|\B_k| \le \frac{4n}{d_{k}}$ holds for $k\in [0 \dd q]$.
\end{proposition}
\begin{proof}
For subsequent integers $k$, we represent the symbols of $T_k$ via \emph{identifier functions} $\id_k$
mapping distinct symbols to distinct identifiers in $[0\dd |T_k|)$.
Thus, we actually store strings $I_k$ such that $|I_k|=|T_k|$ and $I_k[i]=\id_k(T_k[i])$ for $i\in [0\dd |T_k|)$.
Moreover, we store arrays mapping identifiers to expansion lengths of 
the corresponding symbols: $\len_k(I_k[i]):=|\val(T_k[i])|$. 

From this representation, it is easy to derive the equality
\[\B_k = \left\{\sum_{i=0}^{j-1} \len_k(I_k[i]) : j\in[1\dd |I_k|)\right\}.\]
In order to construct $I_0$ and $\len_0$, we sort the characters of $T$ and assign them consecutive positive integer identifiers;
the lengths are obviously $\len_0(I_0[i])=1$ for all $i \in [0 \dd |I_0|)$.
This step takes $\Oh(n)$ time due to the assumption $\sigma = n^{\Oh(1)}$.
Moreover, $|\B_0| = n-1 < 4n = \frac{4n}{d_0}$ holds as claimed.

In order to construct $I_{k+1}$ and $\len_{k+1}$ for $k\ge 0$, we process $I_{k}$ and $\len_{k}$
depending on the parity of $k$.

\paragraph{Case 1:} $k$ is even, that is, $T_{k+1}=\rle_{\Act_{k}}(T_{k})$.

We scan $I_{k}$ from left to right outputting representations of subsequent symbols of $T_{k+1}$.
The representations are elements of $[0\dd |T_{k}|)\,\cup\, [0\dd |T_{k}|)^2$; later, they will be given identifiers in $[0 \dd |T_{k+1}|)$.
Each symbol $S$ in $T_{k+1}$ is represented as $(\id_{k}(S'),m)$, such that $S' \in \aalph(T_{k})$, $m \ge 2$, and $S=(S',m)$ does not occur in $T_{k}$, or as $\id_{k}(S)$ otherwise.
Suppose that $I_{k}[j\dd |I_{k}|)$ is yet to be processed.
If \[\len_{k}(I_{k}[j]) > d_{k}\quad\text{or}\quad I_{k}[j]\ne I_{k}[j+1],\]
we output $I_{k}[j]$ as the next symbol of $T_{k+1}$ and continue processing $I_{k}[j+1\dd |I_{k}|)$.

Otherwise, we determine the maximum integer $m\ge 2$ such that $I_{k}[j']=I_{k}[j]$ for $j'\in [j\dd j+m)$,
output $(I_{k}[j],m)$ as the next symbol of $T_{k+1}$, and continue processing $I_{k}[j+m\dd |I_{k}|)$.
By \cref{fct:cons}, $(T_{k}[j],m)$ does not occur in $T_{k}$ in this case.

Note that $|T_{k+1}|\le |T_{k}|$. The characters of $I_{k+1}$ are initially represented as elements of $[0\dd |T_{k}|)\cup [0\dd |T_{k}|)^2$. Thus, these characters can be sorted in $\Oh(|T_{k}|)$ time so that consecutive integer identifiers $\id_{k+1}$ are assigned to symbols of $T_{k+1}$. 

We also set 
\begin{itemize}
\item 
$\len_{k+1}(\id_{k+1}(S))=m\cdot \len_{k}(\id_{k}(S'))$ if $S$ is represented as $(\id_{k}(S'),m)$ and 

\item $\len_{k+1}(\id_{k+1}(S))=\len_{k}(\id_{k}(S))$ if $S$ is represented as $\id_{k}(S)$.
\end{itemize}

Overall, this algorithm constructs $I_{k+1}$ in $\Oh(|T_{k}|) = \Oh(|\B_{k}|)$ time.
Note that, by the inductive assumption,
\[|\B_{k+1}|\le |\B_{k}| \le \tfrac{4n}{d_{k}}=\tfrac{4n}{d_{k+1}}.\]

\paragraph{Case 2:} $k$ is odd, that is, $T_{k+1}=\pc_{\Left_{k},\Right_{k}}(T_{k})$.

We first partition $\Act_{k}\cap \aalph(T_{k})$ into $\Left_{k}$ and $\Right_{k}$.
Technically, this step results in appropriately marking $\id_{k}(T_k[i])$ depending on whether $T_k[i]\in \Left_{k}$,
$T_k[i]\in \Right_{k}$, or $T_k[i] \not\in \Act_k$.

For this, we scan the array $I_{k}$ and construct a directed multigraph $G_{k}$ with $V(G_{k})=\aalph(I_{k})$.
For each $i\in [1\dd |T_{k}|)$, we add an arc $I_{k}[i-1]\to I_{k}[i]$
provided that 
\[\len_{k}(I_{k}[i-1])\le d_{k}\quad \text{and}\quad\len_{k}(I_{k}[i])\le d_{k}.\]

By \cref{cor:distinct}, this arc is not a self-loop,
so the algorithm of \cref{lem:maxcut} yields in $\Oh(|\B_{k}|)$ time a partition $V(G_{k})=L\cup R$ with at least $\tfrac14|E(G_{k})|$
arcs from $L$ to $R$. For all symbols $S$ such that $\len_{k}(\id_{k}(S))\le d_{k}$, if $\id_{k}(S)\in L$, then we add $S$ to $\Left_{k}$,
and otherwise we add it to $\Right_{k}$.

Next, we scan $I_{k}$ from left to right outputting representations of subsequent symbols of $T_{k+1}$.
Initially, each symbol $S$ in $T_{k+1}$ is represented as $(\id_{k}(S_1),\id_{k}(S_2))$, such that $S_1,S_2 \in \aalph(T_{k})$ and $S=(S_1,S_2)$ does not occur in $T_{k}$, or as $\id_{k}(S)$ otherwise.
Suppose now that $I_{k}[j\dd |I_{k}|)$ is yet to be processed.
 If
 \[j=|I_{k}|-1, \qquad T_{k}[j]\notin \Left_{k}\quad\text{or}\quad T_{k}[j+1]\notin \Right_{k},\]
we output $I_{k}[j]$ as the next symbol of $T_{k+1}$ and continue processing 
the fragment $I_{k}[j+1\dd |I_{k}|)$.

Otherwise, we output $(I_{k}[j],I_{k}[j+1])$ as the next symbol of $T_{k+1}$ and continue processing $I_{k}[j+2\dd |I_{k}|)$.
By \cref{fct:cons}, $(T_{k}[j],T_{k}[j+1])$ does not occur in $T_{k}$ in this case.

Note that $|T_{k+1}|\le |T_{k}|$ and that the characters of $T_{k+1}$ are initially represented as elements of $[0\dd |T_{k}|)\cup [0\dd |T_{k}|)^2$.
Thus, these characters can be sorted in $\Oh(|T_{k}|)$ time so that consecutive integer identifiers $\id_{k+1}$ can be assigned to symbols of $T_{k+1}$. 

We also set 

\begin{itemize}
\item
$\len_{k+1}(\id_{k+1}(S))=\len_{k}(\id_{k}(S_1))+\len_{k}(\id_{k}(S_2))$ if $S$ is represented as $(\id_{k}(S_1),\id_{k}(S_2))$ and 

\item
$\len_{k+1}(\id_{k+1}(S))=\len_{k}(\id_{k}(S))$ if $S$ is represented as $\id_{k}(S)$.
\end{itemize}

Overall, this algorithm constructs $I_{k+1}$ in $\Oh(|T_{k}|) = \Oh(|\B_{k}|)$ time.
Moreover, we have $|\B_{k+1}| \le |\B_{k}|-\tfrac14 |E(G_{k})|$ by construction of the partition $\Act_{k}\cap \aalph(T_k) = \Left_{k} \cup \Right_{k}$.

Observe that, for $i\in [1\dd |T_{k}|)$, an arc  $I_{k}[i-1]\to I_{k}[i]$ is not added to $G_{k}$ 
only if 
$|\val(T_{k}[i-1])| > d_{k}$ or $|\val(T_{k}[i])| > d_{k}$.
There are at most $\frac{n}{d_{k}}$ indices $i\in [0\dd |T_{k}|)$ such that $|\val(T_{k}[i])| > d_{k}$ and each of them prevents at most two arcs from being added to $G_{k}$.
Thus, $|E(G_{k})| \ge |\B_{k}| - \frac{2n}{d_{k}}$, and consequently 
\[|\B_{k+1}| \le |\B_{k}|-\tfrac14 |E(G_{k})| 
\le |\B_{k}| - \tfrac14\left(|\B_{k}| - \tfrac{2n}{d_{k}}\right)
  = \tfrac34|\B_{k}| + \tfrac{n}{2d_{k}}  \le \tfrac{3n}{d_{k}} + \tfrac{n}{2d_{k}} = \tfrac{7n}{2d_{k}}  = \tfrac{4n}{d_{k+1}}.\]

The algorithm terminates after constructing the first string $I_q$ with $|I_q|=1$ and even $q$ (then, $\B_k=\emptyset$ for $k \ge q$). The overall running time is 
\[\Oh\left(n+\sum_{k=0}^q |\B_k|\right)=\Oh\left(\sum_{k=0}^q \tfrac{4n}{d_{k}}\right)=\Oh\left(\sum_{k=0}^\infty (\tfrac{7}{8})^{k/2}n\right) = \Oh(n).\qedhere \]
\end{proof}

\subsection{Properties of Restricted Recompression}
Restricted recompression guarantees that phrase lengths can be bounded as shown in the following fact.
\begin{fact}\label{fct:recompr}
For every $k\in \Zz$ and every symbol $S\in \Symb$ occurring in $T_k$, the expansion $\val(S)$ has length at most $\frac74{d_k}$ or primitive root of length at most $d_k$.
\end{fact}
\begin{proof}
  We proceed by induction on $k$. 
  Let $S$ be a symbol in $T_k$. If $k=0$, then $|\val(S)|=1 < \frac{7}{4}\cdot 1 = \frac74d_0$.
  Thus, we may assume $k>0$.

  If $S$ also occurs in $T_{k-1}$,
  then the inductive assumption shows that $\val(S)$ 
  is of length at most $\frac74{d_{k-1}} \le \frac74{d_k}$
  or its primitive root is of length at most $d_{k-1}\le d_k$. 
  
Otherwise, we have two possibilities.
  If $k$ is odd, then $S = (A,m)\in \Act_{k-1}\times \mathbb{Z}_{\ge 2}$,
  and thus the primitive root of $S$ is of length at most $|\val(A)| \le {d_{k-1}}=d_k$.
  If $k$ is even, on the other hand, then $S = (S_1,S_2)\in \Act_{k-1}\times \Act_{k-1}$,
  so $|\val(S)|=|\val(S_1)|+|\val(S_2)|\le 2{d_{k-1}}=\frac74d_k$.
\end{proof}

What is even more important, restricted recompression guarantees that the sets $\B_k$ are locally consistent.
\lemcons*
\begin{proof}
  We proceed by induction on $k$. The base case of $k=0$ is trivially satisfied
  due to $\B_0 = [\alpha_0\dd n-\alpha_0]$.
  
  Let $k>0$.
  For a proof by contradiction, suppose that \[i\in \B_k\quad \text{and}\quad T[i-\alpha_k \dd i+\alpha_k)\cong 
  T[j-\alpha_k \dd j+\alpha_k)\quad \text{yet}\quad j\notin \B_{k}.\]
  By $\alpha_k>\alpha_{k-1}$ and the inductive assumption, $i\in \B_{k} \sub \B_{k-1}$ implies $j\in \B_{k-1}$.
  
  Let us set $i',j'$ so that
  $i$ and $j$ are the first positions of the phrases induced by $T_{k-1}[i']$ and $T_{k-1}[j']$, respectively,
  that is, $i =f_{k-1,i'}=|\val(T_{k-1}[0 \dd i'))|$ and $j = f_{k-1,j'}= |\val(T_{k-1}[0 \dd j'))|$.
  Since a block boundary was not placed at $T_{k-1}[j']$,
  we have $T_{k-1}[j'-1],T_{k-1}[j']\in \Act_{k-1}$ (see \cref{it:rle,it:pair}). 
  Therefore, the phrases $T[j-\ell\dd j)\cong \val(T_{k-1}[j'-1])$ and $T[j\dd j+r)\cong \val(T_{k-1}[j'])$ around position $j$ are of length at most $\floor{d_{k-1}}$. 
  
  Since $\alpha_k = \alpha_{k-1}+\floor{d_{k-1}}$, by the inductive assumption, $j-\ell \in \B_{k-1}$ and $j+r\in \B_{k-1}$ imply $i-\ell\in \B_{k-1}$ and $i+r\in \B_{k-1}$, respectively.
  Due to \cref{fct:cons}, this yields 
  \[
  T_{k-1}[i'-1]=T_{k-1}[j'-1]\ \text{and}\ T_{k-1}[i']=T_{k-1}[j'].\]
  Consequently, a block boundary was not placed at $T_{k-1}[i']$,
  which contradicts $i\in \B_{k}$.
\end{proof}

Finally, let us show that the elements of the sequence $\alpha_k$ are bounded by the elements of the sequence $d_k$ as follows.
\begin{observation}\label{obs:d_alpha}
For every $k\in \Zz$, we have $\alpha_{k+1} \le 16d_{k}$.
\end{observation}
\begin{proof}
We denote $Z=\left\{\floor{\tfrac{k}{2}},\floor{\tfrac{k-1}{2}}\right\}$ and observe that
 \[\alpha_{k+1}= 1 + \sum_{t=0}^k \floor{d_t} \le 1+\sum_{t=0}^{k}\,\osiemsiedem^{\floor{t/2}} = 1+\sum_{z \in Z}\sum_{t=0}^z \osiemsiedem^{t}
= 1+\sum_{z \in Z} 7\cdot (\osiemsiedem^{z+1}-1)
< 8d_{k}+8d_{k-1} \le 16d_{k}.\qedhere\]
\end{proof}

\section{Details of the Synchronizing Sets Hierarchy Construction}\label{sec:correctness}

In this section, we use the properties of the hierarchy of phrase boundaries $\B_k$ and those of the family $\RUNS_\tau(T)$ of $\tau$-runs to prove that \cref{def:syncfr} yields synchronizing sets of desired size. Moreover, we show how to efficiently build these synchronizing sets.

\begin{lemma}\label{lem:sss}
Let $\S$ be as in \cref{def:syncfr} for a given $\tau\in [1\dd \floor{\frac{n}{2}}]$.
Then, $\S$ is a $\tau$-synchronizing set of size $|\S| < \tfrac{70n}{\tau}$.
\end{lemma}
\begin{proof}
We prove that $\S$ satisfies the two conditions of \cref{def:sss} and analyze the size of $\S$.

\paragraph*{\bf Consistency:}
Suppose that two length-$2\tau$ fragments $x=T[i-\tau \dd i+\tau)$ and $x'=T[j-\tau \dd j+\tau)$ satisfy $x \cong x'$ and $x \in \SF$.
We will show that if $x$ satisfies condition~\ref{it:a}--\ref{it:c} in \cref{def:syncfr}, then $x'$ satisfies the same condition and $x' \in \SF$.

If $x$ satisfies condition~\ref{it:a}, then $i\in \B_{k(\tau)}$, and we need to prove that $j\in \B_{k(\tau)}$. This statement is trivial if $k(\tau)=0$ due to $\B_0 = [1\dd n)$. Otherwise, the statement holds due to \cref{lem:cons} applied for $k=k(\tau)$ because $\alpha_k \le 16d_{k-1}\le \tau$ holds by \cref{obs:d_alpha}, and hence $T[i-\alpha_k\dd i+\alpha_k)\cong T[j-\alpha_k\dd j+\alpha_k)$. Moreover, if $x$ is $\frac{\tau}{3}$-non-periodic, then so is $x' \cong x$.

By the next claim, in conditions~\ref{it:b} and~\ref{it:c} of \cref{def:syncfr}, we do not need to explicitly mention that the respective synchronizing fragment is $\frac{\tau}{3}$-non-periodic. The claim follows from~\cref{fct:uni}.

\begin{claim}\label{fct:bc}
If the first (last) position of a length-$2\tau$ fragment of $T$ is one position to the left (right, respectively) of a $\tau$-run, then the fragment is $\frac{\tau}{3}$-non-periodic.
\end{claim}

If $x$ satisfies condition~\ref{it:b}, then
\[\per(T[i-\tau+1 \dd i])=\per(T[j-\tau+1 \dd j])\le\tfrac{\tau}{3}<\per(T[i-\tau \dd i])=\per(T[j-\tau \dd j]),\]
so there is a $\tau$-run extending the fragment $T[j-\tau+1 \dd j]$ to the right in $T$; therefore, $x'$ satisfies condition~\ref{it:b} and $x' \in \SF$
(cf.~\cref{fct:bc}).

Condition~\ref{it:c} is symmetric to condition~\ref{it:b}; in this case, there is a $\tau$-run extending the fragment $T[j-1 \dd j+\tau-2]$ to the left in $T$.

\paragraph*{\bf Density:}
Consider a position $i\in [0\dd n-3\tau+1]$.
We first show that if $[i\dd i+\tau)\cap \S = \emptyset$, then $\per(T[i\dd i+3\tau-1)) \le \tfrac{\tau}{3}$.

We start by identifying  a $\tau$-run $T[p\dd q)$ with
\[p \le i+\tau-1 \quad \text{and}\quad q\ge i+2\tau. \]
First, suppose that there exists a position $b\in [i+\tau\dd i+2\tau)\cap \B_k$.
Then, $\per(T[b-\tau\dd b+\tau))\le \frac{\tau}{3}$ since otherwise $T[b-\tau\dd b+\tau)$ would have been added to $\SF$ by condition~\ref{it:a},
and thus $b-\tau\in [i\dd i+\tau)$ would have been added to $\S$.
Hence, the run $\run(T[b-\tau\dd b+\tau))$ starts at position $p\le b-\tau \le i+\tau-1$ and ends at position $q\ge b+\tau \ge i+2\tau$.

Next, suppose that $[i+\tau\dd i+2\tau)\cap \B_k = \emptyset$.
Then, $T[i+\tau-1\dd i+2\tau)$ is contained in a single phrase of~$\F_k$.
The length of this phrase is at least $\tau+1$.
If $k=0$, then $\tau+1>1$ contradicts the fact that all phrases in $\F_0$ are of length $1$.
Otherwise, $\tau+1\ge 16d_{k-1} + 1 \ge 14d_k + 1> \frac74 d_k$, so \cref{fct:recompr} yields 
\[\per(T[i+\tau-1\dd i+2\tau))\le d_k < \tfrac{14}{3}d_k \le \tfrac{16}{3}{d_{k-1}} \le \tfrac{\tau}{3}. \]
Now, the $\tau$-run $\run(T[i+\tau-1\dd i+2\tau))$ satisfies the desired requirement.

Note that the fragments $T[p-1\dd p+2\tau-1)$ and $T(q-2\tau\dd q]$ satisfy conditions~\ref{it:b} and~\ref{it:c}, respectively (cf.\ \cref{fct:bc}).
Due to $[i\dd i+\tau)\cap \S = \emptyset$, this implies $p\le i$ and $q\ge i+3\tau-1$,
which means that $\per(T[i\dd i+3\tau-1))\le \frac{\tau}{3}$ holds as claimed.

\medskip
It remains to show, for all $i\in [0\dd n-3\tau+1]$, that
if $\per(T[i\dd i+3\tau-1)) \le \tfrac{\tau}{3}$, then $[i\dd i+\tau)\cap \S = \emptyset$;
equivalently, we need to argue that if $[i\dd i+\tau)\cap \S \ne \emptyset$, then $\per(T[i\dd i+3\tau-1)) > \tfrac{\tau}{3}$.
Indeed, if $s\in [i\dd i+\tau)\cap \S$,
then \[\per(T[i\dd i+3\tau-1))\ge \per(T[s\dd s+2\tau)) > \tfrac{\tau}{3}.\] 

\paragraph*{\bf Size:}
Observe that $|\S| \le |\B_k| + 2|\RUNS_{\tau}(T)|$. 
The second term can be bounded using \cref{fct:uni}:
\begin{claim}
  $|\RUNS_{\tau}(T)| < \frac{3n}{\tau}$.
\end{claim}
\begin{proof}
  By \cref{fct:uni}, distinct $\tau$-runs $\gamma,\gamma'$ satisfy $|\gamma \cap \gamma'|< \frac23\tau$.
  Consequently, each $\tau$-run $\gamma$ contains more than $\frac13\tau$
  (trailing) positions which are disjoint from all $\tau$-runs starting to the left of $\gamma$.
\end{proof}
As for $|\B_k|$, we rely on the upper bound of \cref{lem:recompr2}
and note that $\tau < 16d_k$ holds by the definition of $k=k(\tau)$.
Hence,
\[|\S|\le |\B_k| + 2|\RUNS_{\tau}(T)|
< \tfrac{4n}{d_{k}} + 2\cdot \tfrac{3n}{\tau}
< \tfrac{64n}{\tau}+\tfrac{6n}{\tau}=\tfrac{70n}{\tau}.\qedhere\]
\end{proof}

Before we provide an implementation of \cref{def:syncfr}, we need to make sure that the sets $\RUNS_\tau(T)$ can be built efficiently.
\begin{lemma}\label{lem:tauruns}
  After $\Oh(n)$-time preprocessing of a text $T$ of length $n$,
  given an integer $\tau\in [1\dd n]$, 
  one can construct in $\Oh(\frac{n}{\tau})$ time the set $\RUNS_{\tau}(T)$ (with runs ordered by their start positions). 
\end{lemma}
\begin{proof}
  For every $k\in \Zz$, let us define 
  \[\Runs_k = \{\gamma \in \RUNS(T) : |\gamma|\ge (\tfrac43)^k\text{ and }\per(\gamma)< \tfrac49 (\tfrac43)^k\}.\]
  By \cref{fct:uni}, distinct runs $\gamma,\gamma'\in \Runs_k$ satisfy $|\gamma \cap \gamma'| \le \frac89(\tfrac43)^k$,
  and therefore $|\Runs_k| \le 9\cdot (\tfrac34)^k n$.
  Moreover, note that $\RUNS_\tau(T)\sub \Runs_k$ holds for $k = \lfloor{\log_{4/3} \tau}\rfloor$ due to
  \[\tau \ge (\tfrac43)^k\quad \text{and} \quad \tfrac{\tau}{3}=\tfrac49 (\tfrac43)^{(\log_{4/3}\tau)-1} < \tfrac49(\tfrac43)^k.\]
 
  \paragraph{\bf Preprocessing.}  In the preprocessing phase, the algorithm computes the values $\lfloor{\log_{4/3} \tau}\rfloor$ for $\tau \in [1\dd n]$
  and the sets $\Runs_k$ for $k\in [0\dd \lfloor{\log_{4/3} n}\rfloor]$,
  with runs ordered by their start positions.
  In order to construct the sets $\Runs_k$, the algorithm builds $\RUNS(T)$ using \cref{fct:runs}
  and sorts the runs in $\RUNS(T)$ according to their start positions using bucket sort.
  Then, each $\gamma \in \RUNS(T)$ is added to the appropriate sets $\Runs_k$,
  i.e., whenever $k\in [\lfloor{\log_{4/3} (\frac94 \per(\gamma))}\rfloor\dd \lfloor{\log_{4/3} |\gamma|}\rfloor]$.
  The preprocessing time is therefore $\Oh(n + |\RUNS(T)| + \sum_{k=0}^\infty |\Runs_k|)
  = \Oh(n)$.
 \paragraph{\bf Queries.}  At query time, given an integer $\tau$, the algorithm retrieves $k= \lfloor{\log_{4/3} \tau}\rfloor$
  and iterates over $\gamma \in \Runs_k$, outputting $\gamma$ whenever $\gamma \in \RUNS_{\tau}(T)$.
  The correctness follows from  $\RUNS_\tau(T)\sub \Runs_k$, and the query time is $\Oh(1+|\Runs_k|)
  = \Oh(\frac{n}{\tau})$ due to $|\Runs_k| \le 9\cdot (\tfrac34)^k n = 12\cdot (\tfrac34)^{\lfloor \log_{4/3} \tau\rfloor + 1}n < \frac{12n}{\tau}$.
\end{proof}

We are ready to prove the main result of this section, that is, $\Oh(n)$-time construction of a data structure that allows computing
a $\tau$-synchronizing set in time $\Oh(n/\tau)$, for any $\tau\in [1\dd \floor{\frac{n}{2}}]$.

\thmsss*
\begin{proof}
  In the preprocessing phase, the algorithm constructs the sets $(\B_k)_{k=0}^q$ using \cref{lem:recompr2},
  performs the preprocessing of \cref{lem:tauruns},
  and computes
  $k(\tau) = \max\{j \in \Zz\,:\, j=0 \text{ or }16d_{j-1} \le \tau\}$ for all $\tau \in [1\dd \floor{\tfrac{n}{2}}]$.

  The query algorithm follows \cref{def:syncfr}. Synchronizing fragments satisfying conditions~\ref{it:a}--\ref{it:c}
  are constructed independently (in the left-to-right order) and then the three sorted lists are merged.
  The construction of all three lists relies on the list $\RUNS_\tau(T)$ of $\tau$-runs obtained from \cref{lem:tauruns}. The list is sorted by start positions, but since no $\tau$-run is contained within another $\tau$-run, also by end positions.

  The synchronizing positions satisfying condition~\ref{it:a} are $j-\tau$ for each $j\in \B_{k(\tau)}\cap [\tau\dd n-\tau]$ such that $T[j-\tau\dd j+\tau)$ is not contained in any $\tau$-run,
  i.e., \[\{T[\ell\dd r)\in \RUNS_{\tau}(T)\,:\, \ell \le j-\tau\text{ and } r \ge j+\tau\}\,=\, \emptyset.\]
  To check this condition, the algorithm simultaneously iterates over positions in $\B_{k(\tau)}$ and 
  the list $\RUNS_{\tau}(T)$.

  The positions satisfying condition~\ref{it:b} are $\ell-1$ for every $T[\ell\dd r)\in \RUNS_{\tau}(T)$ with $\ell \in [1\dd n-2\tau+1]$.
  (Recall \cref{fct:bc}.)
  Thus, they can be generated by iterating over the list $\RUNS_{\tau}(T)$.

  Finally, the positions satisfying condition~\ref{it:c} are $r-2\tau+1$ for every $T[\ell\dd r)\in \RUNS_{\tau}(T)$ with $r \in [2\tau-1\dd n-1]$.
  Thus, they can be generated by iterating over the list $\RUNS_{\tau}(T)$.

  Overall, constructing $\S$ costs $\Oh(1+|\B_{k(\tau)}|+|\RUNS_{\tau}(T)|)$ time.
  As observed in the proof of \cref{lem:sss}, each of these terms can be bounded by $\Oh(\frac{n}{\tau})$.
  Finally, the same lemma guarantees that $\S$ is a $\tau$-synchronizing set of size $|\S| < \tfrac{70n}{\tau}$.
\end{proof}

\section{\IPM with $\N$-Patterns}\label{chp:ipm}
\newcommand{\rx}{\hat{x}}

As discussed in \cref{sec:techniques}, in order to support \IPM in the text~$T$, we use a classic idea of pattern matching by deterministic sampling~\cite{DBLP:journals/siamcomp/Vishkin91} in a novel way.
The main trick is to select a consistent family $\R$ of \emph{samples}.
This allows answering \emph{restricted} \IPM, with $x\in \R$, using a relatively simple approach
in $\Oh(|\R|)$ space; see \cref{sec:impl}.

\begin{figure}[h]
  \begin{center}
\begin{tikzpicture}[xscale=.445,yscale=.6,font=\footnotesize]
\draw (3.5,0) rectangle (12.5, 1);
\draw[dashed] (3.5, 0) -- (1.5,0);
\draw[dashed] (3.5, 1) -- (1.5,1);
\draw[dashed] (12.5, 0) -- (14,0);
\draw[dashed] (12.5, 1) -- (14,1);

\draw (3, .5) node[anchor=mid] {$x$:};

\draw[fill=black!10] (5, .2) rectangle  (10, .8);
\draw (7.5, 0.5) node[anchor=mid]{\scriptsize $\rx$};

\draw[latex-latex] (12.5, 1.3) --  (10, 1.3);
\draw[latex-latex] (3.5, 1.3) --  (5, 1.3);

\foreach \x in {4.25, 11.1, 11.25,11.4}
	\draw (\x,1.15) -- (\x,1.45);

\draw[densely dotted] (3.5, 0) -- (3.5, 1.6);
\draw[densely dotted] (5, 0) -- (5, 1.6);
\draw[densely dotted] (10, 0) -- (10, 1.6);
\draw[densely dotted] (12.5, 0) -- (12.5, 1.6);

\begin{scope}[xshift = 14cm]
\draw (3.5,0) rectangle (21.5, 1);
\draw[dashed] (3.5, 0) -- (0,0);
\draw[dashed] (3.5, 1) -- (0,1);
\draw[dashed] (12.5, 0) -- (23,0);
\draw[dashed] (12.5, 1) -- (23,1);

\draw (3, .5) node[anchor=mid] {$y$:};

\foreach \x/\y in {4/0.09, 12/0.17} {
\draw[fill=black!10] (\x+1.5, .25) rectangle  (\x+6.5, .75);
\draw[densely dashed] (\x, \y) rectangle (\x+9, \y+0.75);
\draw[densely dotted] (\x, 0.5) -- (\x, 1.6);
\draw[densely dotted] (\x+1.5, 0.5) -- (\x+1.5, 1.6);
}

\foreach \x/\y in {4/1.3, 12/1.7}{
\foreach \z in {\x, \x+1.5, \x+6.5, \x+9}{
	\draw[densely dotted] (\z, 0.5) -- (\z, \y+0.2);
}

\draw[latex-latex] (\x, \y) --  (\x+1.5, \y);
\draw[latex-latex] (\x+9, \y) --  (\x+6.5, \y);
\foreach \z in {\x+0.75, \x+7.6, \x+7.75,\x+7.9}{
	\draw (\z,\y-0.15) -- (\z,\y+0.15);
}
}

\end{scope}
\end{tikzpicture}
\end{center}
  \caption{\label{fig:repr}
  The main idea of the query algorithm for fragments $x$ and $y$.
  We find the occurrences of $\rx \in \R$ contained in $y$ (depicted as gray rectangles).
  If there is an occurrence $x'$ of $x$ contained in $y$, then $x'$
  can be obtained by extending an occurrence of $\rx$ within $y$.
  Hence, $x'$ must be one of the fragments marked with dashed rectangles. 
  }
  \end{figure}

In the general case, we first select an arbitrary sample $\rx\in \R$ contained in $x$
and search for the occurrences of $\rx$ within $y$. 
Then, our query algorithm checks which occurrences of $\rx$ can be extended to occurrences of $x$; see \cref{fig:repr}.
In order to achieve constant query time with this approach, we need to guarantee that $\rx$ has $\Oh(1)$  occurrences in $y$. 
In general, already $x$ might have $\omega(1)$ occurrences within $y$.
Nevertheless, the following observation can be applied to bound the number of occurrences provided that $x$ does not have a short period.
In particular, if $x\in \N$ and $|y|<2|x|$, then $x$ has at most three occurrences within~$y$.

\begin{fact}[Sparsity of occurrences]\label{fct:far}
  If a substring $u$ occurs in a text $T$ at distinct positions $i,i'$, then $|i-i'|\ge \per(u)$.
  \end{fact}
  \begin{proof}
  If $i,i+d\in\Occ(u,T)$ with $d\in [1\dd |u|]$, then $u[j+d]=T[i+j+d]=T[i+j]=u[j]$ for every $j\in[0\dd |u|-d)$, i.e., $d$ is a period of $u$.
  \end{proof}

\subsection{Selection of Samples}\label{sec:repr}
Let us start with a formal definition based on the discussion above.
\begin{definition}\label{def:repr}
  A family $\R$ of fragments of $T$ is a \emph{samples family} if it satisfies the following conditions:
  \begin{enumerate}[label=(\alph*)]
    \item\label{it:repr:cons} For all fragments $x,x'$ of $T$, if $x\cong x'$ and $x\in \R$, then $x'\in \R$.

\smallskip
    \item\label{it:repr:dens} For every $x\in \N$, there exists a fragment $\rx\in \R$ contained in $x$ such that $\per(\rx)=\Theta(|x|)$.
  \end{enumerate}
\end{definition}

Let us first analyze simple ways to select samples. 
Arguably, the most naive choice is $\R=\N$.
In this case, the correctness is obvious, with each $x\in \N$ being its own sample.
However, the number of samples can only be bounded by $\Oh(n^2)$.

An easy way to dramatically reduce the number of samples is to set $\R = \{\rx \in \N : \log|\rx|\in \mathbb{Z}\}$ so that $|\R|=\Oh(n \log n)$. 
Then, as a sample of $x\in \N$, we can select an arbitrary fragment $\rx\in \N$
of length $2^{\floor{\log |x|}}$ contained in $x$. 
By the following result, such a fragment always exists.
\begin{fact}\label{fct:nonp}
  For every fragment $x\in \N$ and length $\ell\in [1\dd |x|]$,
  there is a fragment $\rx\in \N$ of length $\ell$ contained in $x$.
  \end{fact}
  \begin{proof}
  Suppose that $x[0\dd \ell)\notin \N$.
  Then, there exists a run $\gamma:=\run(x[0\dd \ell))$ with $\per(\gamma)\le \frac13\ell \le \frac13|x| < \per(x)$. 
  Hence, $x\cap \gamma$ is a proper prefix of $x$,
  i.e., $x\cap \gamma = x[0\dd i)$ for some $i\in [\ell\dd |x|)$.
  We then define $y=x(i-\ell\dd i]$.
  If $y\notin \N$, then there would be another run $\gamma':=\run(y)\ne \gamma$ in $x$ with $\per(\gamma')\le \frac13\ell$.
  Now, $|\gamma \cap \gamma'|\ge |x(i-\ell\dd i)| = \ell-1$ contradicts  \cref{fct:uni}.
  \end{proof}

We aim to select just $\Oh(n)$ samples.
For this, we use synchronizing sets of \cref{thm:sss}.
\begin{lemma}\label{lem:repr}
The family $\R = \{T[i\dd i] : i\in [0\dd n)\} \cup \bigcup_{k=1}^{\floor{\log n}} \SF_k$,
where $\SF_k$ is the set of $2^{k-1}$-synchronizing fragments of \cref{thm:sss}, satisfies \cref{def:repr}.
\end{lemma}
\begin{proof}
Let $\S_k$ denote the $2^{k-1}$-synchronizing set containing starting positions of fragments in $\SF_k$.
The consistency property~\ref{it:repr:cons} follows immediately from the 
corresponding property of synchronizing sets.

As for the existence of samples (property~\ref{it:repr:dens}),
let us fix $x\in \N$ and choose $k\in [0\dd \floor{\log n}]$
so that $3\cdot 2^{k-1}-1 \le |x|<  3\cdot 2^k-1$.

If $k=0$, then $|x|=1$, so $x\in \R$ can be chosen as its own sample.

Otherwise, let $z\in \N$ be a fragment contained in $x$ such that $|z|=3\cdot 2^{k-1}-1$; such a fragment $z$ exists due to \cref{fct:nonp}.
Observe that 
\[\per(z) \ge \tfrac13(|z|+1) = 2^{k-1} > \tfrac13 \cdot 2^{k-1}.\]
Let $z=T[i\dd i+|z|)$; by \cref{def:sss}, $\S_k\cap [i\dd i+2^{k-1})\ne \emptyset$,
so we select an arbitrary fragment $\rx:= T[s\dd s+2^k) \in \SF_k$ with $s\in [i\dd i+2^{k-1})\cap \S_k$ as the sample of $x$.
Note that $\rx$ is contained in $z$ and $z$ is contained in $x$.
Moreover, \[ \per(\rx)> \tfrac13\cdot 2^{k-1}= \tfrac{1}{18}\cdot 3\cdot 2^k > \tfrac{1}{18}|x|\]
by the extra condition of \cref{thm:sss}, and thus $\per(\rx)=\Theta(|x|)$ holds as required.
\end{proof}

We conclude with an algorithmic construction based on \cref{lem:repr}.
\begin{proposition}\label{prp:repr}
Given a text $T$ of length $n$, one can in $\Oh(n)$ time construct
a samples family $\R$ along with a data structure that,
given a fragment $x\in \N$,
  in $\Oh(1)$ time reports a sample $\rx\in \R$ contained in $x$ and satisfying $\per(\rx)>\frac{1}{18}|x|$.
Moreover, for each $m\in [1\dd n]$, we have $|\{\rx \in \R : |\rx|\ge m\}|=\Oh(n/m)$.
\end{proposition}
\begin{proof}
The construction builds $2^{k-1}$-synchronizing sets $\S_k$
for $k\in [1\dd \floor{\log n}]$ and the family $\R$ as specified in \cref{lem:repr}.
The synchronizing sets are built using \cref{thm:sss}, which gives $|\S_k| = \Oh(\frac{n}{2^{k-1}})$.
Moreover, the construction time is $\Oh(n + \sum_{k=1}^{\floor{\log n}}\frac{n}{2^{k-1}})=\Oh(n)$.
For every $m\in[1\dd n]$, the number of samples of length $m$ or more
is $\Oh(\sum_{k=\ceil{\log m}}^{\floor{\log n}}\frac{n}{2^{k-1}})=\Oh(n/m)$ as claimed.

It remains to efficiently implement assigning samples
to fragments $x\in \N$, following the approach described in the proof of \cref{lem:repr}.

The case of $|x|=1$, when $x$ is its own sample, does not require any infrastructure. 

To efficiently find samples for $|x|\in [3\cdot 2^{k-1}-1 \dd 3\cdot 2^k-1)$
and $k\in [1\dd \floor{\log n}]$, we
store \[\pred(\S_k,i)=\max\{j \le i : j\in \S_k\}\quad\text{and}\quad
\suc(\S_k,i)=\min\{j > i : j\in \S_k\}\] for each position $i$ divisible by $2^{k-1}$.
The size and construction time of this component is $\Oh(\frac{n}{2^k})$,
which is $\Oh(n)$ in total across all values of $k$.

A query for a sample of $x=T[\ell\dd r)\in \N$ is answered as follows. 
As described in the proof of \cref{lem:repr}, we have $\S_{k}\cap [\ell\dd r-2^{k}] \ne \emptyset$. 
Moreover, $r-\ell=|x|\ge 3\cdot 2^{k-1}-1$ also yields 
\[i:=2^{k-1}\ceil{\frac{\ell}{2^{k-1}}}\,\in\, [\ell\dd r-2^{k}].\]
Consequently, $\pred(\S_k,i)\in [\ell\dd r-2^{k}]$
or  $\suc(\S_k,i)\in [\ell\dd r-2^{k}]$, and the underlying fragment can be reported 
as the sample of $x$.
The query time is constant.
\end{proof}

\subsection{Implementation of the Data Structure}\label{sec:impl}
As outlined at the beginning of this section, to search for the occurrences of $x$ in $y$, we first find
in $y$ the  occurrences of the sample $\rx$ of $x$.
This step is implemented using auxiliary \RIPM specified below.
Next, we apply \LCEQ (see \cref{prop:lce}) to check which occurrences of $\rx$ can be extended to occurrences of $x$;
see also \cref{fig:repr}.

\defdsproblem{\RIPM}{
\textbf{Input}: A text $T$ and a family $\R$ of fragments of $T$.\\
\textbf{Queries}: Given a fragment $x\in \R$ and a fragment $y$ of $T$, report all fragments $x'\in \R$
contained in $y$ and matching $x$.
}

Due to the sparsity of occurrences (\cref{fct:far}),
it is relatively easy to implement \RIPM in $\Oh(\frac{|y|}{\per(x)})$ time using static dictionaries.

\begin{lemma}\label{lem:locator}
For any family $\R$ of fragments of a length-$n$ text $T$, there exists a data structure of size $\Oh(n+|\R|)$
that answers \RIPM in $\Oh(\frac{|y|}{\per(x)})$ time. It can be constructed in $\Oh(n+|\R|\log^2 \log |\R|)$ time
in general and in $\Oh(n+|\R|)$ time if $|x|=\omega^{\Oh(1)}$ for each $x\in \R$, where $\omega$ is the machine word size.
\end{lemma}
\begin{proof}
Given the family $\R$, we construct an identifier function $\id:\R\to\mathbb{Z}$ such that $\id(x)=\id(x')$ if and only if the fragments $x,x'\in \R$ match.
For this, we order the fragments $x=T[\ell\dd r)\in \R$ by the length $|x|$ and the lexicographic rank of the suffix $T[\ell\dd n)$ among the suffixes of $T$ (this rank is the $\ell$th entry of the inverse suffix array of $T$, which can be built in $\Oh(n)$ time~\cite{DBLP:journals/jacm/KarkkainenSB06}).
Matching fragments $x\in\R$ appear consecutively in this order, so we use \LCEQ (see~\cref{prop:lce}) to determine the boundaries between the equivalence classes.

We store two collections of dictionaries.
The first collection allows converting samples to identifiers.
The second collection stores non-empty answers to selected \RIPM.

\paragraph{Dictionaries of identifiers.}
We store the $\id$ function in multiple static dictionaries jointly mapping each fragment $x\in \R$  to the identifier $\id(x)$.
Specifically, for each position $\ell$ in $T$, we store a dictionary $\mathcal{D}(\ell)$ mapping $r$ to $\id(x)$
for every fragment $x=T[\ell \dd r)\in \R$. 

In the general case, we use deterministic dictionaries by Ružić~\cite{DBLP:conf/icalp/Ruzic08},
which provide $\Oh(1)$ query time, take $\Oh(m)$ space, and have $\Oh(m\log^2 \log m)$
construction time, where $m$ is the dictionary size. 
Across all positions $\ell$, this brings the overall space consumption to $\Oh(n+|\R|)$ and the overall construction time to $\Oh(n+|\R|\log^2 \log |\R|)$.
In case of short fragments $|x|=\omega^{\Oh(1)}$, we use fusion trees~\cite{DBLP:conf/focs/PatrascuT14},
which provide $\Oh(1+\log_\omega m)$ query time, take $\Oh(m)$ space, and require $\Oh(m(1+\log_\omega m))$
construction time. Since the number of fragments in $\R$
starting at any given position $\ell$ is $\omega^{\Oh(1)}$ in this case, this brings the overall space consumption and
the overall construction time to $\Oh(n+|\R|)$, whereas the query time is $\Oh(1)$.

\paragraph{Dictionaries of answers to selected queries.}
For each $k\in [0\dd \ceil{\log n}]$, we cover the text $T$ with blocks (fragments) of length $2^{k+1}-1$ with overlaps of length $2^k-1$ (the last block can be shorter)
and denote the resulting family of blocks by $\Y(k)$.
We store the non-empty answers to \RIPM for $x\in \R$ and $y\in \Y(\ceil{\log |x|})$ in multiple static dictionaries.
Specifically, for each integer $k\in [0\dd \ceil{\log n}]$ and each fragment $y\in \Y(k)$,
we store a dictionary $\mathcal{D}'(k,y)$.
For each sample $x \in \R$ that is contained in $y$ and satisfies $\ceil{\log |x|}=k$,
the dictionary $\mathcal{D}'(k,y)$ maps the identifier $\id(x)$ to the answer to a \RIPMOne for $x$ and $y$.
Note that each fragment $x$ is contained in one or two blocks $y\in Y(\ceil{\log |x|})$, 
so each $x\in \R$ appears in $\Oh(1)$ precomputed answers, and thus the total number of dictionary entries is $\Oh(|\R|)$.

In the general case, we use deterministic dictionaries by Ružić~\cite{DBLP:conf/icalp/Ruzic08},
which provide $\Oh(1)$ query time, take $\Oh(n+|\R|)$ space in total, and have $\Oh(n+|\R|\log^2 \log |\R|)$ 
overall construction time.
In case of short patterns, we use fusion trees~\cite{DBLP:conf/focs/PatrascuT14},
which provide $\Oh(1)$ query time, take $\Oh(n+|\R|)$ space in total, and have $\Oh(n + |\R|)$ construction time,
since each individual dictionary is of size $\omega^{\Oh(1)}$ in this case (because each $y\in \Y(k)$ contains $\omega^{\Oh(1)}$ fragments $x\in \R$ with $\ceil{\log |x|}=k$).

\paragraph{Query algorithm.}
To answer a query for $x=T[\ell \dd r) \in \R$ and $y$, we first compute $k=\ceil{\log |x|}$ and $\id(x)$ using $\mathcal{D}(\ell)$.
Next, we use simple arithmetics to obtain $\Oh(\frac{|y|}{|x|})$ blocks $y'\in \Y(k)$ 
that collectively contain all length-$|x|$ fragments contained in $y$.
We take the union of the precomputed answers for the identifier $\id(x)$ in dictionaries $\mathcal{D}'(k,y')$ to obtain a collection of fragments $x'\in \R$ matching $x$ and contained in one of the blocks $y'$.
By \cref{fct:far}, there are $\Oh(\frac{|y|}{\per(x)})$ such fragments $x'$, so we can filter and report those contained in $y$ spending $\Oh(1)$ time on each candidate $x'$.
\end{proof}

\begin{corollary}\label{cor:locator}
  Given the family $\R$ of samples of a length-$n$ text $T$
  constructed through \cref{prp:repr},
  one can in $\Oh(n)$ time construct a data structure that answers \RIPM
  in $\Oh(\frac{|y|}{\per(x)})$ time.
\end{corollary}
\begin{proof}
We store two instances of the data structure of \cref{lem:locator}: the first instance,
for fragments of length more than $\omega$, contains $\Oh(\frac{n}{\omega})$ samples,
and thus its construction time is $\Oh(n+\frac{n}{\omega}\log^2 \log \frac{n}{\omega})=\Oh(n)$.
The instance for the remaining fragments, of length at most $\omega$, contains $\Oh(n)$
samples, and thus its construction time is also $\Oh(n)$.
\end{proof}

We conclude with a full description of the data structure for \IPM for $\N$-patterns.
\begin{proposition}\label{prp:ipm}
  For every text $T$ of length $n$, there exists a data structure of size $\Oh(n)$ which answers \IPM in $\Oh(1)$ time provided that $x\in \N$. The data structure can be constructed in $\Oh(n)$ time.
\end{proposition}
\begin{proof}
  The core of our data structure is the samples family $\R$, constructed
  using \cref{prp:repr} along with a component for efficiently selecting a sample,
  plus the data structure answering \RIPM for $\R$, constructed using \cref{cor:locator}.
  Additionally, we include a data structure for \LCEQ (\cref{prop:lce}).
  Each of these ingredients takes $\Oh(n)$ space and $\Oh(n)$ time to build. 
  
  The query algorithm for fragments $x\in \N$ and $y$ works as follows.
  First, we use \cref{prp:repr} to obtain in $\Oh(1)$ time a sample $\rx\in \R$ contained in $x$ and satisfying $\per(\rx)=\Theta(|x|)$.
   Next, we query the component of \cref{cor:locator} to find all occurrences of $\rx$ contained in $y$;
   this takes $\Oh(\frac{|y|}{\per(\rx)})=\Oh(1)$ time.
  As a result, we obtain a constant number of candidate positions where $x$ may occur in $y$.
  We verify them using \LCEQ in $\Oh(1)$ time each.
  Recall that the occurrences of $x$ in $y$ form an arithmetic progression (\cref{fct:single}).
  \maybeqed \end{proof}

\section{\IPM with $\HP$-Patterns}\label{chp:per}

Our approach to \IPM with $\HP$-patterns relies on the structure of $\HP$-runs in the text. In order to answer a query, we look for runs that may arise as $\run(x')$ for the occurrences $x'$ of $x$ within $y$.
Due to the assumption $|y|<2|x|$, all these runs contain the middle position of $y$  (the position $T\big[\floor{\frac{\ell+r}{2}}\big]$ if  $y=T[\ell\dd r)$).
All such runs need to have length at least $|x|$ and period at most $\per(x)\le \frac13|x|$, which allows us to show that there are $\Oh(1)$ such runs.
In \cref{sec:finder}, we develop a component listing such runs in $\Oh(1)$ time after $\Oh(n)$-time preprocessing.

 Next, we look for the occurrences of $x$ contained in each candidate run $\gamma$.
 For $x$ to have any occurrence in $\gamma$, the periods of $x$ and $\gamma$ need to be equal;
 furthermore, the \emph{string periods} of $x$ and $\gamma$---that is, the prefixes of the two fragments of length equal to their periods---need to be cyclic rotations of each other.
 We then say that $x$ and $\gamma$ are \emph{compatible}.
 To check this condition, we compare the lexicographically minimal rotations of the string periods, called \emph{Lyndon roots}.
 We use techniques originating from a paper by Crochemore et al.~\cite{DBLP:journals/tcs/CrochemoreIKRRW14}, listed in \cref{sec:comp},
 to check compatibility of substrings of $T$ and list occurrences of an $\HP$-pattern in a compatible $\HP$-text;
 in this case, the set of occurrences forms an arithmetic progression.
 
 The query algorithm is described in \cref{sec:per_sum}.
 In brief, for each run $\gamma$ in $y$ that is compatible with~$x$, we can find all occurrences of $x$ in $y\cap \gamma$ using the toolbox of~\cite{DBLP:journals/tcs/CrochemoreIKRRW14} (recalled in \cref{sec:comp}).
 We may obtain $\Oh(1)$ arithmetic progressions representing the occurrences of $x$ in $y$,
 but \cref{fct:single} guarantees that they can be merged into a single progression.

\subsection{Special $\HP$-Runs}\label{sec:finder}
If $x$ is an $\HP$ fragment, then \cref{obs:char} yields the following useful characterization of $\run(x)$.

\begin{observation}
Consider a text $T$ and an $\HP$-fragment $x$ with $k=\lfloor \log |x|\rfloor$. 
Then, $\gamma=\run(x)$ satisfies the following condition:
\begin{equation}\label{eq:obs:sprun}
\per(\gamma)< \tfrac13 \cdot 2^{k+1},\qquad |\gamma|\ge 2^k,\quad\text{and}\quad\text{$\gamma$ is a $\HP$-run}.
\end{equation}
\end{observation} 

For a positive integer $k$, we say that a run $\gamma$ is \emph{$k$-special} if 
it satisfies condition \eqref{eq:obs:sprun} above.
Denote by $\SPEC_k(i)$ the set of $k$-special runs covering position $i$ in $T$.
Below, we develop a data structure for efficiently answering the following queries:

\defdsproblem{\textsc{Special Run Locating Queries}}{
Given a position $i$ of $T$ and an integer $k\in [0\dd \floor{\log n}]$, 
compute $\SPEC_k(i)$.
}

We first prove that the answers must be of constant size. Then, we develop a component for answering \textsc{Special Run Locating Queries} based on the fact that, even though there are $\Theta(n\log n)$ possible queries, it suffices to precompute answers to $\Oh(n)$ of them.
\begin{lemma}\label{lem:atmost5}
$|\SPEC_k(i)|\le 5$ 
for every integer $k\in [0\dd \floor{\log n}]$ and position $i$ in $T$.
\end{lemma}
\begin{proof}
For a proof by contradiction, suppose that there are at least six such runs $\gamma_j = T[\ell_j\dd r_j)$
with $p_j = \per(\gamma_j)$ for $j\in [1\dd 6]$ and $\ell_1\le \cdots \le \ell_6$.
For each $j\in [1\dd 6)$, \cref{fct:uni} yields 
\[|\gamma_j\cap \gamma_{j+1}|< p_j + p_{j+1} \le \tfrac13(2^{k+1}+|\gamma_{j+1}|)\le |\gamma_{j+1}|,\]
 so $\gamma_{j+1}$ is not contained in $\gamma_j$.
Observe also that $|\gamma_j|-2p_j\ge \tfrac13 |\gamma_j|\ge \tfrac13\cdot 2^k$.
Consequently, 
\begin{multline*}
\ell_{6}-\ell_1 \,=\, \sum_{j=1}^{5}(\ell_{j+1}-\ell_j) \,=\, \sum_{j=1}^5 (|\gamma_j| - |\gamma_j\cap \gamma_{j+1}|)\,>\, 
\sum_{j=1}^5 (|\gamma_j| - (p_j+p_{j+1}))\,=\,\\
 |\gamma_1|-(p_1+p_6) + \sum_{j=2}^5 (|\gamma_j| -2 p_j) \ge\; |\gamma_1| - 2\cdot \tfrac13 \cdot 2^{k+1} + \sum_{j=2}^5 \tfrac13\cdot 2^k
\,=\, |\gamma_1| - \tfrac23\cdot 2^{k+1} + \tfrac43\cdot 2^k \,=\, |\gamma_1|.
\end{multline*}
We derived $\ell_{6}-\ell_1   > |\gamma_1|$, which contradicts $\gamma_1\cap \gamma_6 \ne \emptyset$.
\end{proof}

\begin{lemma}\label{lem:run_finder}
  For every text $T$, there exists a data structure that answers \textsc{Special Run Locating Queries} in $\Oh(1)$ time,
  takes $\Oh(n)$ space, and can be constructed in $\Oh(n)$ time.
\end{lemma}
\begin{proof}
  For every integer $k\in [0\dd \floor{\log n}]$, the data structure contains the precomputed answers for all positions $i$ divisible by $2^k$.
  Each of these answers takes $\Oh(1)$ space by \cref{lem:atmost5}, so the total size of the data structure is $\Oh(n)$.

  As for the query algorithm, we note that if a $k$-special run $\gamma$ covers position $i$, then,
  due to $|\gamma|\ge 2^{k}$, it also covers position $2^k \floor{\frac{i}{2^k}}$ or $2^k \ceil{\frac{i}{2^k}}$.
  Thus, the query algorithm retrieves the answers for these two positions and, among the obtained $k$-special runs,
  reports those covering position $i$. The query time is constant by \cref{lem:atmost5}.

  As for the construction algorithm, we enumerate all runs using \cref{fct:runs}.
  For each $\HP$-run $\gamma=T[\ell\dd r)$ and each $k\in [\floor{\log(3\per(\gamma))}\dd \floor{\log |\gamma|}]$,
  we append $\gamma$ to the precomputed answers for all positions $i\in [\ell\dd r)$ that are multiples of $2^k$.
  Since there is at least one such position for each considered pair $(\gamma,k)$, the running time of this process
  is proportional to the time complexity of the algorithm of \cref{fct:runs} plus the total size of the precomputed answers, both of which are $\Oh(n)$.
\end{proof}

\subsection{Compatibility of Strings and Runs}\label{sec:comp}
A primitive string $w\in \Sigma^+$ is called a \emph{Lyndon word}~\cite{Lyndon1954,chen1958free} if $w\preceq w'$ for every rotation $w'$ of $w$.
Let $u$ be a string with the smallest period $\per(u)=p$.
The \emph{Lyndon root} $\lambda$ of $u$ is the lexicographically smallest rotation of the prefix $u[0\dd p)$.
We say that two strings are \emph{compatible} if they have the same Lyndon root.

A string $u$ with Lyndon root $\lambda$ can be uniquely represented as $\lambda'\lambda^k \lambda''$, where $\lambda'$ is a proper suffix of $\lambda$,
 $\lambda''$ is a proper prefix of $\lambda$, and $k\in \integ_{\ge 0}$ is a non-negative integer.
The \emph{Lyndon signature} of $u$ is defined as $(|\lambda'|,k,|\lambda''|)$.
Note that the Lyndon signature uniquely determines $u$ within its compatibility class.
This representation is very convenient for pattern matching if the text is compatible with the pattern.

\begin{lemma}\label{lem:lyndon}
Let $x$ and $y$ be compatible strings.
The set of positions where $x$ occurs in $y$ is an arithmetic progression that
can be computed in $\Oh(1)$ time given the Lyndon signatures of $x$ and $y$.
\end{lemma}
\begin{proof}
Let $\lambda$ be the common Lyndon root of $x$ and $y$
and let their Lyndon signatures be $(p, k, s)$ and $(p', k', s')$ respectively.
\cref{lem:synchr} (synchronization property) implies that $\lambda$ occurs in $y$
only at positions $i$ such that $i\equiv p'\pmod{|\lambda|}$.
Consequently, $x$ occurs in $y$ only at positions $i$ such that $i\equiv p'-p\pmod{|\lambda|}$.
Clearly, $x$ occurs in $y$ at all such positions $i$ within the interval $[0\dd |y|-|x|]$.
Therefore, it is a matter of simple calculations to compute the arithmetic progression of these positions.
\maybeqed \end{proof}

Crochemore et al.~\cite{DBLP:journals/tcs/CrochemoreIKRRW14} showed how to efficiently compute
Lyndon signatures of the runs of a~given~text.

\begin{fact}[Crochemore et al.~\cite{DBLP:journals/tcs/CrochemoreIKRRW14}]\label{fct:compat}
There exists an algorithm that, given a text $T$ of length~$n$, in $\Oh(n)$ time
computes Lyndon signatures of all runs in $T$.
\end{fact}

Finally, we note that the Lyndon signature of a periodic fragment $x$ 
can be derived from the Lyndon signature of $\run(x)$.
\begin{observation}\label{obs:compat}
Let $u$ be a fragment of a periodic string $w$ such that  $|u|\ge 2\per(w)$.
Then, $u$ is compatible with $w$.
Moreover, given the Lyndon signature of $w$, one can compute the Lyndon signature of $u$ in $\Oh(1)$ time.
\end{observation}

\subsection{Answering Queries}\label{sec:per_sum}

Our data structure consists of the set of runs $\RUNS(T)$ (\cref{fct:runs}), with each run accompanied by its period and Lyndon signature (\cref{fct:compat}), the data structure for \LCEQ (\cref{prop:lce}), and the data structure of \cref{lem:run_finder} for \textsc{Special Run Locating Queries}.
The entire data structure takes $\Oh(n)$ space and can be constructed in $\Oh(n)$ time.

As outlined at the beginning of \cref{chp:per}, the query algorithm consists of the following steps:

\begin{center}
\begin{minipage}{14cm}
\noindent {\bf Algorithm} answering \IPM with $\HP$-patterns

  \begin{enumerate}[label=(\Alph*)]
    \item\label{step:A} Compute the Lyndon signature of $x$.
    \item\label{step:B} Find all $\floor{\log |x|}$-special runs $\gamma$ containing the middle position of $y$,
      defined as $T[\lceil{\frac{\ell+r}{2}}\rceil]$ for $y=T[\ell\dd r)$, along with their Lyndon signatures.
    \item\label{step:C} Filter out runs $\gamma$ incompatible with $x$.
    \item\label{step:D}
	  For each of the compatible runs $\gamma$, compute an arithmetic progression representing the occurrences of $x$ in $y\cap \gamma$.
	  \item\label{step:E}
      Combine the resulting occurrences of $x$ in $y$ into a single arithmetic progression.
  \end{enumerate}
  \end{minipage}
\end{center}

\paragraph*{Correctness.} Clearly, a fragment matching $x$ (and contained in $y$) starts at each of the reported positions. 
It remains to prove that no occurrence $x'$ is missed.
Since $|y|<2|x|$, each occurrence $x'$ of $x$ in $y$ contains the middle point of $y$.
Therefore, $\gamma=\run(x')$ is among the runs found in step~\ref{step:B}.
By \cref{obs:compat}, $\gamma$~is compatible with $x'$ and, since $x$ and $x'$ match, $\gamma$ must be compatible with $x$. Hence, $\gamma$ is considered in step~\ref{step:D} and the starting position of $x'$ is reported in step~\ref{step:E}.

\paragraph*{Implementation.}
In step~\ref{step:A}, we use a \textsc{Special Run Locating Query} to list all $\floor{\log |x|}$-special runs containing the first position of $x$, and then we check if any of these runs contains the whole $x$ and has period not exceeding $\frac13|x|$. If so, this run is $\run(x)$ by \cref{obs:char}; otherwise, we raise an error to indicate that $x\in \N$.
We then use \cref{fct:compat,obs:compat} to retrieve the Lyndon signature of $\run(x)$ and $x$, respectively.
In step~\ref{step:B}, we use another \textsc{Special Run Locating Query} to identify all  $\floor{\log |x|}$-special runs $\gamma$ containing the middle position of $y$. \Cref{lem:atmost5} guarantees that we obtain at most five runs $\gamma$.
In step~\ref{step:C}, for each run $\gamma$, we use the Lyndon signatures 
to identify the Lyndon roots of $x$ and $\gamma$, represented as fragments of $T$,
and we ask an \textsc{LCE Query} to check if the Lyndon roots match.

 For the remaining (compatible) runs $\gamma$, we apply \cref{lem:lyndon} to find the occurrences of $x$ in $y\cap \gamma$.
 There is nothing to do if $|x|>|y\cap \gamma|$. 
Otherwise, $|y\cap\gamma|\ge |x|\ge 3\per(x)= 3\per(\gamma)$, so \cref{obs:compat} lets us retrieve the Lyndon signature of $y\cap \gamma$.
 We are left with at most five arithmetic progressions, one for each compatible run $\gamma$.
 As argued above, their union represents all occurrences of $x$ in $y$.
 By \cref{fct:single}, this set must form a single arithmetic progression.
 If the progressions are stored by (at most) three elements---the last one and the first two---it is easy to compute the union in constant time.

The above query algorithm also checks if the fragment $x$ is highly periodic. 
This concludes the proof of the following result:

\begin{proposition}\label{thm:per}
For every text $T$ of length $n$, there exists a data structure of size $\Oh(n)$ which answers \IPM in $\Oh(1)$ time provided that $x\in \HP$,
and reports an error whenever $x\in \N$.
The data structure can be constructed in $\Oh(n)$ time.
\end{proposition}

Combining \cref{thm:per} with \cref{prp:ipm}, we obtain an efficient data structure for \IPM over integer alphabets of polynomial size.
\begin{theorem}\label{thm:ipm0}
  For every text $T$ of length $n$, there exists a data structure of size $\Oh(n)$ which answers \IPM in $\Oh(1)$ time.
  The data structure can be constructed in $\Oh(n)$ time.
\end{theorem}

In the following section, we show how the data structure can be improved in the case of a small alphabet.

\section{\IPM in Texts over Small Alphabets}\label{sec:packed}
In this section, we assume that $\sigma \le \sqrt[19]{n}$; otherwise, \cref{thm:ipm} follows directly from \cref{thm:ipm0}.
We transform the string $T$ into a string $T'$ of length $\Oh(n/\tau)$, where $\tau = \floor{\frac1{19} \log_\sigma n}$.
Then, each \IPMOne in $T$ is reduced to a constant number of \IPM in $T'$.
Without loss of generality, we assume that $T$ starts and ends with unique characters to avoid degenerate cases.

A naive idea to construct the string $T'$ would be to partition $T$ into blocks of length $\tau$ and interpret each block as an integer with $\tau\cdot \ceil {\log \sigma}= \Oh(\log n)$ bits.
Unfortunately, this approach is not helpful: even if a pattern fragment $x$ has an occurrence in a text fragment $y$ of $T$,
the longest fragment of $x$ consisting of full blocks may have no ``aligned'' occurrence in the longest fragment of $y$ consisting of full blocks.
Therefore, we partition the string into blocks using the elements of a $\tau$-synchronizing set, which can be constructed in $\Oh(n / \log_\sigma n)$ time for the aforementioned value of $\tau$; see \cref{prp:synch}.

\subsection{Constructing Data Structure}
We use a $\tau$-synchronizing set 
$\S = \{s_0,\ldots, s_{m-1}\}$ of $T$, where $s_0 < \cdots < s_{m-1}$, and the corresponding set of synchronizing fragments $\SF=\{f_0,\ldots,f_{m-1}\}$, where
$f_i=T[s_i \dd s_i+2\tau)$. 
The aforementioned assumption on $T$ implies $s_0 \in [0\dd \tau)$ and $s_{m-1} \in (n-3\tau\dd n-2\tau]$; in particular, $\S \ne \emptyset$.

Denote $\Delta_i=s_{i+1}-s_i$ for $i\in [0\dd m-1)$.
We construct a length-$(2m-1)$ string $T'$ such that 
\[\begin{aligned}
T'[2i] &= f_i & \text{for }i\in [0\dd m),\\
T'[2i+1] &= \Delta_i & \text{for }i\in [0\dd m-1).
\end{aligned}
\]
Every character of $T'$ is either an integer in $[0\dd n]$ or a length-$2\tau$ substring of $T$ (which can be interpreted as an integer with $2\tau \ceil{\log n} = \Oh(\log n)$ bits); see \cref{fig:jak_chcesz}.

\begin{figure}[h]
  \centering
  \begin{tikzpicture}[xscale=0.63,yscale=0.74]
  \def\tauW{1.5}
  \def\perW{1.2}
  \def\boxH{0.5}

  \node at (0,0.5) [left] {$T=$};

  \def\curr{0}
  \foreach [count=\i]
    \w/\c
    in {1/A, 1/B, 2/C, 2/D, 2/C, 6/C, 4/D, 1/E} {
    \coordinate (l\i) at (\curr,0);
    \coordinate (b\i) at (\curr+\w,0);
    \coordinate (r\i) at (\curr+\w,0.8);
    \xdef\curr{\number\numexpr\curr+\w\relax}
  }
  \draw[thick] (l1) rectangle (r7) node[midway] {};  

  \begin{scope}[xshift=7cm,yshift=0.8cm]
    \clip (-\perW,0) rectangle (8.7,0.3);
    \foreach \dx in {0,...,10}{
      \coordinate (c\dx) at (\dx*\perW,0);
      \draw (\dx*\perW, 0) sin ((\dx*\perW+0.5*\perW, 0.3) cos ((\dx*\perW+1*\perW, 0);
    }
    \coordinate (b10) at ($(c0)+(-0.4+\tauW,-0.8)$);
    \coordinate (b11) at ($(c7)+(+0.5-\tauW,-0.8)$);
  \end{scope}

  \foreach \i/\y in {3/2, 2/1, 4/3} {
    \coordinate (bl\i) at ($(b\i)+(-\tauW,-\y)$);
    \coordinate (bll\i) at ($(b\i)+(-\tauW+0.2,-\y+0.2)$);
    \coordinate (br\i) at ($(b\i)+(\tauW,-\y)$);
    \coordinate (bc\i) at ($(b\i)+(0,-\y)$);
    \coordinate (bu\i) at ($(b\i)+(0,-\y+\boxH)$);

    \draw[thin, blue, -latex] (bll\i) -- (l\i -| bll\i);

    \draw[fill=white] (bl\i) rectangle ($(br\i)+(0,\boxH)$);
  }

  \foreach \i/\y in {10/1, 11/2} {
    \coordinate (bl\i) at ($(b\i)+(-\tauW,-\y)$);
    \coordinate (bll\i) at ($(b\i)+(-\tauW+0.2,-\y+0.2)$);
    \coordinate (br\i) at ($(b\i)+(\tauW,-\y)$);
    \coordinate (bc\i) at ($(b\i)+(0,-\y)$);
    \coordinate (bu\i) at ($(b\i)+(0,-\y+\boxH)$);

    \draw[thin, blue, -latex] (bll\i) -- (l1 -| bll\i);

    \draw[fill=white] (bl\i) rectangle ($(br\i)+(0,\boxH)$);
  }
  \draw[densely dotted] (br11)--($(br11 |- r1)+(0,0.2)$);

  \node[blue, above] at (bll2 |- b1) {\small $s_0$};
  \node[blue, above] at (bll3 |- b1) {\small $s_1$};
  \node[blue, above] at (bll4 |- b1) {\small $s_2$};
  \node[blue, above] at (bll10 |- b1) {\small $s_3$};
  \node[blue, above] at (bll11 |- b1) {\small $s_4$};

  \node[red, above] (j) at ($(bll4 |- b1) + (-1cm,0)$) {$j$};
  \draw[violet,-latex, bend left=45] (j |- r1) to (bll4 |- r1);
  \node[above] at ($(bll4 |- r1) + (-0.5cm,0.4cm)$) {\small $\BLOCK(j)=2$};

  \node[blue,yshift=0.17cm] at (bc2) {\small $f_0$};
  \node[blue,yshift=0.17cm] at (bc3) {\small $f_1$};
  \node[blue,yshift=0.17cm] at (bc4) {\small $f_2$};
  \node[blue,yshift=0.17cm] at (bc10) {\small $f_3$};
  \node[blue,yshift=0.16cm] at (bc11) {\small $f_4$};

\begin{scope}[yshift=-4cm,xshift=-2cm]
\node at (0,0.05) {$T'=$};
\foreach [count=\i] \l in {$f_0$, $f_1$, $f_2$, $f_3$, $f_4$} {
  \coordinate (q\i) at (-3.2cm+4.2cm*\i,0);
  \draw ($(q\i)+(0, -\boxH*0.5)$) rectangle +(2*\tauW,\boxH) node[midway,blue] {\small \l};
}
\foreach [count=\i] \l in {\Delta_0, \Delta_1, \Delta_2, \Delta_3} {
  \node at ($(q\i)+(2*\tauW,0)+(.6cm,0)$) {\small $,\textcolor{blue}{\l},$};
}
\end{scope}

\end{tikzpicture}
  \caption{A schematic view of the transformation from $T$ to $T'$, assuming $\S=\{s_0,\ldots,s_4\}$ and $\Delta_i=s_i-s_{i-1}$.
  The $\BLOCK$ operation returns the index of the nearest synchronizing position if it exists; see below.}\label{fig:jak_chcesz}
\end{figure}

We say that a fragment of $T$ is \emph{regular} if it is
of the form $T[s_i\dd s_j+2\tau)$ for $0 \le i \le j < m$. For such a fragment $z$, we denote the fragment $\code(z)=T'[2i\dd 2j]$.
\begin{fact}\label{fct:tp}
If $z$ and $z'$ are regular fragments, then \[z\cong z'\Longleftrightarrow \code(z)\cong \code(z').\]
\end{fact}
\begin{proof}
Let $z=T[s_i\dd s_j+2\tau)$ and  $z'=T[s_{i'}\dd s_{j'}+2\tau)$.

\paragraph*{(Implication $\Rightarrow$)}
Suppose that $T[s_i\dd s_j+2\tau)\cong T[s_{i'}\dd s_{j'}+2\tau)$.
We conclude from the consistency property of $\S$ that $j-i=j'-i'$ and $(\Delta_i,\ldots,\Delta_{j-1})=(\Delta_{i'},\ldots,\Delta_{j'-1})$.
This implies that $(f_i,\ldots,f_j)\cong (f_{i'},\ldots,f_{j'})$.
Hence, $T'[2i\dd 2j]\cong T'[2i'\dd 2j']$ holds as claimed.

\paragraph*{(Implication $\Leftarrow$)}
Suppose that $T'[2i\dd 2j]\cong T'[2i'\dd 2j']$.
Equivalently, we have $(f_i,\ldots,f_j)\cong (f_{i'},\ldots,f_{j'})$ and $(\Delta_i,\ldots,\Delta_{j-1})=(\Delta_{i'},\ldots,\Delta_{j'-1})$.
To obtain $T[s_i\dd s_j+2\tau)\cong T[s_{i'}\dd s_{j'}+2\tau)$, we only need to show that if $\Delta_{a}=\Delta_{a'}>2\tau$ holds for some $a$, $a'$ such that $0 \le a-i = a'-i' \le j-i$, then $T[s_{a} \dd s_{a+1}) = T[s_{a'} \dd s_{a'+1})$.
In this case, by the density property of $\S$, we have \[p:=\per(T(s_{a}\dd s_{a+1}+2\tau-1))\le \tfrac13\tau\quad\text{and}\quad
p':=\per(T(s_{a'}\dd s_{a'+1}+2\tau-1))\le \tfrac13\tau.\]
Since $T[s_{a}\dd s_{a}+2\tau)=f_a \cong f_{a'}=T[s_{a'}\dd s_{a'}+2\tau)$,
\cref{obs:char} implies $p=p'$ and therefore $T(s_{a}\dd s_{a+1}+2\tau-1)\cong T(s_{a'}\dd s_{a'+1}+2\tau-1)$.
As $f_a \cong f_{a'}$ also yields $T[s_a]=T[s_{a'}]$, this concludes the proof.
\end{proof}

Let $s_m=n$ and $s_{-1}=0$ be sentinels. 
The sequence $s_{-1},s_0,\ldots,s_{m-1},s_m$ of synchronizing positions
together with the sentinels partitions $[0\dd n)$ into
intervals called \emph{blocks}: 
\[[s_{-1} \dd s_0),\;[s_0\dd s_1),\; [s_1\dd s_2),\;
[s_2\dd s_3),\ldots, [s_{m-1}\dd s_m).\]
We define the following operation mapping each position $j\in [1\dd n)$ to the index of the block it belongs to;
see \cref{fig:jak_chcesz}:
\[\BLOCK(j)=i\quad \text{such that}\quad j\in [s_{i-1}\dd s_i).\]
We build a length-$n$ bitmask representing $\S$.
Then the $\BLOCK(j)=|\{i\in [0\dd m) : s_i \le j\}|$ operation can be viewed as a $\rank$ query on this bitmask.
These queries can be answered in $\Oh(1)$ time using a data structure of size $\Oh(n/\log n)$
that can be constructed in $\Oh(|\S|+n/\log n)=\Oh(n/\log_\sigma n)$ time~\cite{DBLP:conf/focs/Jacobson89,WaveletSuffixTree,DBLP:journals/tcs/MunroNV16}.

Recall that a run $\gamma$ is a $\tau$-run if $|\gamma| \ge \tau$ and $\per(\gamma) \le \frac13 \tau$.
We say that a $\tau$-run $\gamma$ is \emph{long} if $|\gamma| \ge 3\tau-1$ and use the following proposition for the considered $\tau$.

\begin{proposition}[{\cite[Section 6.1.2]{Kempa2019}}]\label{prp:taurons}
For a positive integer $\tau$, a string $T \in [0 \dd \sigma)^n$ contains $\Oh(n/\tau)$ long $\tau$-runs.
Moreover, if $\tau \le \frac19 \log_\sigma n$, we can compute all long $\tau$-runs in
$T$, compute their Lyndon signatures, and group the long $\tau$-runs by their Lyndon roots in $\Oh(n/\tau)$ time.
\end{proposition}

Let us note that a $\tau$-run is $\floor{\log\tau}$-special.
Hence, each position in $T$ belongs to at most five $\tau$-runs (\cref{lem:atmost5}).
Moreover, we can locate long $\tau$-runs using the same data structure as in \cref{lem:run_finder}, but constructed only for
$k=\floor{\log\tau}$ (and using \cref{prp:taurons} to compute all long $\tau$-tuns):
\begin{lemma}\label{lem:run_finder2}
  For every text $T$, there exists a data structure that can compute all long $\tau$-runs containing a given position in $\Oh(1)$ time,
  takes $\Oh(n/\tau)$ space, and can be constructed in $\Oh(n/\tau)$ time.
\end{lemma}

\noindent
We also perform preprocessing for \LCEQ (\cref{prop:lce}).

Finally, for every pair $(x,y)$ of strings such that $|x| < 8\tau$ and $|y| < 10\tau$, we precompute
the set of occurrences of $x$ in $y$, represented as an arithmetic progression.
The total number of such pairs $(x,y)$ is $\Oh(\sigma^{18\tau})=\Oh(n^{18/19})$, and for each such pair,
the output can be computed in $\Oh(\tau^{\Oh(1)})=n^{o(1)}$ time, so this preprocessing can be performed
in $\Oh(n/\tau)$ time.

\subsection{Answering Queries} 
We denote by $\MaxReg(u)$ the longest regular fragment 
contained in the fragment $u$ of $T$ and denote  $\Phi(u)=\code(\MaxReg(u))$.
Note that, for any two matching fragments $u\cong u'$, we have $\MaxReg(u) \cong \MaxReg(u')$ (by consistency of $\S$)
and $\Phi(u)\cong \Phi(u')$ (by \cref{fct:tp}); see \cref{fig:zz1}.

\begin{figure}[h]
 \centering
 \begin{tikzpicture}[xscale=1.2,yscale=1.2]
  \def\barW{1.0}
  
  \coordinate (l) at (0,0);
  \coordinate (lc) at (0, 0.2);
  \coordinate (r) at (10, 0.4);
  
  \node at ($(l)+(-1,0.25)$) [right] {$T=$};
  
  \draw[thick] (l) rectangle (r);
  
  \coordinate (x) at ($(l |- r)+(0.5,0)$);
  \coordinate (xr) at ($(x)+(3.2,0.3)$);
  \draw (x) rectangle (xr);
  \node at ([xshift=0.2cm,yshift=0.15cm]x) {$x$};
  
  \coordinate (y) at ($(l |- r)+(5.1,0)$);
  \coordinate (yr) at ($(y)+(4.5,0.3)$);
  \draw (y) rectangle (yr);
  \node at ([xshift=0.2cm,yshift=0.15cm]y) {$y$};

  \foreach \dx/\dy [count=\i from 0] in {-0.1/0.5, 0.85/1, 2.0/1, 1.2/0.5, 2.6/0.5} {
    \coordinate (sl\i) at ($(x |- l)+(\dx,-\dy-0.2)$);
    \coordinate (sr\i) at ($(sl\i)+(\barW, 0.2)$);
    \coordinate (su\i) at (sl\i |- l);
  
    \draw[blue,fill=lightgray] (sl\i) rectangle (sr\i);
    \draw[blue] (sl\i)--(su\i);
  }
  \draw[blue,fill=black!10!white] (sl0) rectangle (sr0);
  \draw[blue,fill=black!10!white] (sl4) rectangle (sr4);

  \node at (su1) [blue, above] { };
  \node at (su2) [blue, above] {};

  
  \coordinate (zl) at (sl1 |- xr);
  \coordinate (zr) at ($(sr2 |- xr)+(0,0.3)$);
  \draw[densely dotted] (sl1)--(zl |- xr);
  \draw[densely dotted] (sr2)--(zr |- xr);
  \draw (zl) rectangle (zr);
  \node at ([xshift=0.2cm,yshift=0.15cm]zl) {$z$};
  
  \coordinate (zc) at ($(zl)+(0,0.15cm)$);
  \draw[|<->|] (x |- zc)--(zc);

  \coordinate (y1) at (y |- l);
  \coordinate (y2) at ($(y1) + (1.1cm,0)$);
  \coordinate (y3) at ($(y2) + (0.7cm,0)$);
  
  \foreach \dx/\dy [count=\i] in {0.85/0.1, 2.0/0.1, 1.2/-0.35} {
    \coordinate (ssl\i) at ($(y3)+(\dx-0.85,\dy)$);
    \coordinate (ssr\i) at ($(ssl\i)+(\barW, 0.2)$);
  
    \draw[blue,fill=lightgray] (ssl\i) rectangle (ssr\i);
  }
  
  \node at (ssl1 |- ssr1) [above,yshift=0.03cm] { };
  \node at (ssl2 |- ssr2) [above] { };
  
  \coordinate (zzl) at ($(ssl1 |- zl) + (0, 0.3cm)$);
  \coordinate (zzr) at ($(zzl -| ssr2) + (0, 0.3cm)$);
  
  \coordinate (xxl) at ($(zzl) + (0, -0.3) - (0.85,0)$);
  \coordinate (xxr) at ($(xxl) + (3.2,0.3)$);

  \draw (zzl) rectangle (zzr);
  \node at ([xshift=0.2cm,yshift=0.17cm]zzl) { $z'$};
  
  \draw[densely dotted] (ssl1) -- (zzl);
  \draw[densely dotted] (ssr2) -- (zzr);

  \draw (xxl) rectangle (xxr);
  \node at ([xshift=0.2cm,yshift=0.17cm]xxl) { $x'$};
  
  \coordinate (zzc) at ($(zzl)+(0,0.15cm)$);
  \draw[|<->|] (xxl |- zzc)--(zzc);
  
  \end{tikzpicture}
 \caption{We have $z=\MaxReg(x)$ and $z'=\MaxReg(x')$.
The arrows correspond to the same distances due to synchronization. Each fragment $z'$ matching $z$ and contained in $y$ determines a single location of a candidate match $x'$.}\label{fig:zz1}
\end{figure}

\begin{observation}
After $\Oh(n/\tau)$-time preprocessing, given fragments $x$, $y$ of $T$,
we can compute the fragments $\MaxReg(x)$, $\MaxReg(y)$ in $T$ and 
their codes $\Phi(x)$, $\Phi(y)$ in $T'$ in $\Oh(1)$ time.
\end{observation}

First, we consider the case when $|x| \ge 8\tau$ and $\MaxReg(x)\ne \varepsilon$; the remaining corner cases will be addressed later.
In \IPM, we assume that the length of $y$ is proportional to the length of $x$. 
Here, we will make a stronger assumption that $|x| \le |y| \le \frac54|x|$. 
However, this assumption does not imply immediately that $|\Phi(y)|$ is proportional to $|\Phi(x)|$.
The latter condition is needed to apply \IPM to fragments $\Phi(x)$ and $\Phi(y)$ because $|\Phi(y)|$ could be too large compared with $|\Phi(x)|$.

Denote $y=T[\ell\dd r)$ and let $\Mid(y)=\floor{\frac{\ell+r}{2}}$ be the middle position of $y$.
Our approach, in this case, is to restrict the search to 
an appropriate fragment of length at most $2|\Phi(x)|+1$,
\[\Phi'(y)=T'[\,2i-|\Phi(x)|-1\dd 2i+|\Phi(x)|-1\,] \cap \Phi(y),\ \text{where}\ i=\BLOCK(\Mid(y)),\]
equal to an approximately the ``middle'' part of $\Phi(y)$; see \cref{new}.

\begin{figure}[h]
 \centering
 \begin{tikzpicture}[xscale=0.9,yscale=0.4]

\filldraw[white!85!black] (1,0) rectangle (6.5,1);
\draw[thick] (0,0) rectangle (8,1);
\draw (0,0.5) node[left] {$y=$};
\draw[-latex] (4,2) -- (4,1);
\draw (4,2) node[above] {$\Mid(y)=j$};
\draw[snake=brace] (6.5,-0.2) -- node[below] {$\MaxReg(y)$} (1,-0.2);

\begin{scope}[xshift=10cm]
\filldraw[white!85!black] (1.5,0) rectangle (5.5,1);
\draw[thick](0,0) rectangle (6.5,1);
\draw (0,0.5) node[left] {$x=$};
\draw[snake=brace] (5.5,-0.2) -- node[below] {$z=\MaxReg(x)$} (1.5,-0.2);
\end{scope}

\begin{scope}[yshift=-6cm]
\begin{scope}[xshift=1cm]
\filldraw[white!70!black] (0,0) rectangle (4.5,1);
\draw (0,0) rectangle (4.5,1);
\draw (0,0.5) node[left] {$\Phi(y)=$};
\draw[-latex] (2.75,2) -- (2.75,1);
\draw (2.75,2) node[above] {$i=\BLOCK(j)$};
\begin{scope}[yshift=-2cm]
\filldraw[white!70!black] (1.5,0) rectangle (4,1);
\draw (1.5,0.5) node[left] {$\Phi'(y)=$};
\draw (1.5,0) rectangle (4,1);
\draw (2.75,0) node[below] {$i$};
\draw[latex-latex] (1.5,-0.9) -- node[below] {$|\Phi(x)|$} (2.5,-0.9);
\draw[latex-latex] (3,-0.9) -- node[below] {$|\Phi(x)|$} (4,-0.9);
\end{scope}
\end{scope}
\begin{scope}[xshift=13cm]
\filldraw[white!70!black] (0,0) rectangle (1,1);
\draw (0,0) rectangle (1,1);
\draw (0.5,0) node[below] {$\Phi(x)=\code(z)$};
\end{scope}
\end{scope}
  
\end{tikzpicture}  
 \caption{Instead of searching for $\Phi(x)$ within $\Phi(y)$, we only search within a fragment $\Phi'(y)$ of length at most $2|\Phi(x)|+1$.
 Recall that each character of $\Phi(x)$ and $\Phi(y)$ fits in a single machine word.
}\label{new}
\end{figure}

\begin{fact}\label{fct:easy}
Let $x$ and $y$ be fragments of $T$ such that $8\tau \le |x| \le |y| \le \frac54|x|$, and
let $x'$ be a fragment matching $x$ and contained in $y$.
If $\MaxReg(x)\ne \eps$, then $\Phi(x')$ contains $T'[2i-2]$ or $T'[2i]$, where $i=\BLOCK(\Mid(y))$.
\end{fact}
\begin{proof}
Recall that $\MaxReg(x)\cong \MaxReg(x')$, so $\MaxReg(x')\ne \eps$.
Let $z'=\MaxReg(x')=T[s_{i'} \dd s_{j'}+2\tau)$ so that $\Phi(x')=T[2i'\dd 2j']$.
It suffices to prove that $\{i-1,i\} \cap [i'\dd j'] \ne \emptyset$.
\begin{itemize}
\item
If $\Mid(y)\in (s_{i'-1} \dd s_{j'+1})$, then $s_i \in [s_{i'}\dd s_{j'+1}]$, so $i \in [i' \dd j'+1]$ and $\{i-1,i\} \cap [i'\dd j'] \ne \emptyset$.

\smallskip
\item If $\Mid(y) \le s_{i'-1}$, then, since both $T[\Mid(y)]$ and $z'=T[s_{i'}\dd s_{j'}+2\tau)$ are contained in $x'$,
we conclude that $T[s_{i'-1}\dd s_{j'}+2\tau)$ is contained in $x'$, contradicting the definition of $\MaxReg(x')$.

\smallskip
\item Finally, if $\Mid(y)\ge s_{j'+1}$, then we claim that $T[s_{i}\dd s_{j'+1}+2\tau)$ is contained in $x'$,
contradicting the definition of $\MaxReg(x')$.
Indeed, $\tfrac12 |y|+2\tau \le \frac58|x| + \frac14|x| < |x|$,
so $x'$ contains both $T[s_{i'}\dd s_{j'}+2\tau)$ and $T[\Mid(y)\dd \Mid(y)+2\tau)$,
and thus also $T[s_{i'}\dd s_{j'+1}+2\tau)$. \qedhere
\end{itemize}
\end{proof}

In the query algorithm below, we assume without loss of generality that $|x| \le |y| \le \frac54|x|$.
(To handle $|y|\le 2|x|$, we combine the answers of up to four queries.)

\begin{center}
\begin{minipage}{14cm}
\noindent {\bf Algorithm} answering \IPM over small alphabets

  \begin{enumerate}[label=(\Alph*)]
    \item\label{Step:A} If $|x|<8\tau$, return a precomputed answer.
    \item\label{Step:B} If $\MaxReg(x)=\varepsilon$, apply the query algorithm for $\HP$-patterns of \cref{sec:per_sum}.
    \item\label{Step:C} Apply \IPM for pattern $\Phi(x)$ and text $\Phi'(y)$ (\cref{thm:ipm0}), obtaining
    at most two arithmetic progressions of occurrences.
    \item\label{Step:D} 
    Apply \cref{lem:peralg}\ref{it:permax} for $T$ and $T^R$ to test which of these occurrences extend to occurrences of $x$.
  \end{enumerate}
  \end{minipage}
\end{center}

\paragraph*{Correctness.}
In step~\ref{Step:B}, if $\MaxReg(x)=\varepsilon$, then the density of $\S$ implies $\per(x)\le \frac13\tau$, so indeed $x$ is an $\HP$-pattern.

In step~\ref{Step:C}, we have $|\Phi'(y)| \le 2|\Phi(x)|+1$, so \cref{fct:single} implies that there are up to two arithmetic progressions.

In step~\ref{Step:D}, we only care about occurrences of $\Phi(x)$ starting at even positions of $\Phi'(y)$.
Each arithmetic progression from step~\ref{Step:C} forms a periodic progression $\mathbf{p}=(p_i)_{i=0}^{k-1}$ in~$T$.
This is because, by \cref{fct:tp}, subsequent occurrences of $z=\MaxReg(x)$ start at all these positions.
Let $x=w z w'$ and $y=T[\ell \dd r)$.
We apply \cref{lem:peralg}\ref{it:permax} in $T$ to sequence $\mathbf{p}$, position $r$, and fragment $zw'$ to check which positions $p_i$ contain occurrences of $zw'$.
Then, we apply \cref{lem:peralg}\ref{it:permax} in $T^R$ to sequence
$\mathbf{p}'=(n-p_{k-1-i})_{i=0}^{k-1}$, position $n-\ell$, and fragment $w^R$
to check which positions $p_i-1$ in $T$ are endpoints of occurrences of $w$.
In each case, the lemma returns an integer interval of indices.
The intersection of the two intervals can be transformed into an arithmetic progression of positions.
The two resulting arithmetic progressions can be joined together to one progression by \cref{fct:single}.

\paragraph*{Implementation.}
Recall that $T$ is given in a packed representation.
In step~\ref{Step:A}, this lets us retrieve any fragment of length at most $10\tau$, encoded in a single machine word, in $\Oh(1)$ time.
Then we can use the precomputed answer.

In step~\ref{Step:B}, the preprocessing of the query algorithm of \cref{sec:per_sum} takes only $\Oh(n/\tau)$ time and space as we use \cref{lem:run_finder2} to find long $\tau$-runs and the LCE queries of \cref{prop:lce}.

In step~\ref{Step:C}, we have $|T'|=\Oh(n/\tau)$ and $T'$ can be extracted from the packed representation of $T$ in $\Oh(n/\tau)$ time.
The preprocessing of \IPM of \cref{thm:ipm0} on $T'$ takes $\Oh(n/\tau)$ time and space.

\medskip
\noindent
This concludes the description of our data structure for \IPM.

\thmipm*

\section{Applications of IPM and LCE Queries}\label{chp:app}
We present some applications of our data structure for \IPMFull.
This includes answering \PQ (\cref{sec:BQ}), \FC (\cref{sec:FC}), and variants of \LSC (\cref{sec:GSC}).
Before that, we prove \cref{lem:peralg}, which is useful in all our applications, as discussed in \cref{sec:techniques}.

\subsection{Proof of \cref{lem:peralg}}\label{sec:app:comb}

Let us recall that a sequence $\mathbf{p}$ of positions $p_0< p_1 < \cdots < p_{k-1}$ in a string $w$
is a  \emph{periodic progression} of length $k$ (in $w$) if $w[p_0 \dd  p_{1})\cong \cdots \cong w[p_{k-2}\dd  p_{k-1})$.
If $k\ge 2$, we call the string  $v\cong w[p_i \dd  p_{i+1})$ the \emph{(string) period} of $\mathbf{p}$, while its length $p_{i+1}-p_i$ is the \emph{difference} of $\mathbf{p}$.
Periodic progressions $\mathbf{p},\mathbf{p}'$ are called \emph{non-overlapping} if the last term of $\mathbf{p}$ is smaller than the first term of $\mathbf{p}'$ or vice versa: the last term of $\mathbf{p}'$ is smaller than the first term of $\mathbf{p}$.
Every periodic progression is an arithmetic progression and, consequently, can be represented by three integers,
e.g., the terms $p_0$, $p_1$, and $p_{k-1}$ (with $p_1$ omitted if $k=1$, i.e., if $p_{k-1}=p_0$).

All our applications of \IPM rely on the structure of the values $\LCE(p_i,q)$ for a periodic progression $(p_i)_{i=0}^{k-1}$.
In \cref{lem:peralg} below, we give a combinatorial characterization of this structure (in a slightly more general form) amended with its immediate algorithmic applications.
Let us start with a simple combinatorial result.
\begin{fact}\label{fct:lcp}
For strings $u,v\in \Sigma^*$ and $\rho\in\Sigma^+$, let $d_u = \lcpinf{\rho}{u}$ and $d_v=\lcpinf{\rho}{v}$.
\begin{enumerate}[label=(\alph*)]
  \item If $d_u \ne d_v$, then $\lcp(u,v)=\min(d_u,d_v)$.
  \item If $d_u = d_v$, then $\lcp(u,v)\ge d_u=d_v$.
\end{enumerate}
\end{fact}
\begin{proof}
Let $d=\min(d_u,d_v)$. Note that $u[0\dd d)\cong (\rho^\infty)[0\dd d)\cong v[0\dd d)$, so $\lcp(u,v)\ge d$.
If $d=d_u < d_v$, then $v[d]=(\rho^\infty)[d]\ne u[d]$, so $\lcp(u,v)=d$. The case of $d=d_v < d_u$ is symmetric.
\end{proof}

\begin{fact}[Applications of \LCEQ]\label{fct:lce}
Assume that we have access to a text $T$ equipped with a data structure answering \LCEQ in $\Oh(1)$ time.
Given fragments $x,y$ of $T$, the values $\lcp(x,y)$ and $\lcpinf{x}{y}$
can be computed in $\Oh(1)$ time.
\end{fact}
\begin{proof}
If $x=T[i_x\dd j_x)$ and $y=T[i_y\dd j_y)$,
then $\lcp(x,y)=\min(\LCE(i_x,i_y),|x|,|y|)$.
Hence, $\lcp(x,y)$ can be computed in constant time.

If $\lcp(x,y)< |x|$, i.e., $x$ is not a prefix of $y$, then $\lcpinf{x}{y}=\lcp(x,y)$.
Otherwise, consider a fragment $y'=T[i_y+|x|\dd j_y)$. A simple inductive proof shows that 
$\lcpinf{x}{y} = |x|+\lcpinf{x}{y'}=|x|+\lcp(y,y')$.
In either case, $\lcpinf{x}{y}$ can be computed in constant time.
\end{proof}

\begin{figure}[ht]
\begin{center}
\begin{tikzpicture}[scale=0.7]


    \draw (3.2,0) rectangle (10, 0.6);

    \draw (2.7, .3) node[anchor=mid] {\footnotesize$u_0$};
    
    \draw (4.1,-.6) rectangle (10, 0);
    \draw (3.6, -.3) node[anchor=mid] {\footnotesize$u_1$};
    \draw (5,-1.2) rectangle (10, -0.6);
    \draw (4.5, -.9) node[anchor=mid] {\footnotesize$u_2$};
    
    \filldraw[fill=black!20] (3.2,0) rectangle (8.2,0.6);
    \filldraw[fill=black!20] (4.1,-.6) rectangle (9.5,0);
    \filldraw[fill=black!20] (5,-1.2) rectangle (9.1,-0.6);

    \begin{scope}
        \clip (3.2,0) rectangle (9.1, 1);
        \foreach \x in {0, 0.9, ..., 10} {
          \draw (3.2+\x, .6) sin (3.2+\x+.45, .85) cos (3.2+\x+.9, .6);
        }
    \end{scope}

        \draw[<->] (3.2, 1.5) -- node[anchor=mid,above=-0.09]{\footnotesize $d_u$} (9.1, 1.5);
         \draw[dotted] (3.2,0.6)--(3.2,1.7)   (9.1, -1.2)--(9.1, 1.7);
        \draw (3.65, 1.05) node {\footnotesize $\rho$};

    \begin{scope}[xshift = 14 cm]
    \foreach \x in {0,1,2}{
        \draw (0,0-\x*0.6) rectangle (8,.6-\x*0.6);
        \draw (-.5, .3-\x*0.6) node[anchor=mid] {\footnotesize$v$};
        }
        
        \filldraw[fill=black!20] (0,0) rectangle (5,0.6);
    \filldraw[fill=black!20] (0,-.6) rectangle (5.4,0);
    \filldraw[fill=black!20] (0,-1.2) rectangle (4.1,-0.6);

        \draw[<->] (0, 1.5) -- node[anchor=mid,above=-0.09]{\footnotesize $d_v$} (5, 1.5);
        \draw (0.45, 1.05) node {\footnotesize $\rho$};

        \draw[dotted] (0,0.6)--(0,1.7)   (5, .6)--(5, 1.7);

        \begin{scope}
            \clip (0,0) rectangle (5, 1);
            \foreach \x in {0, 0.9, ..., 11} {
                \draw (\x, .6) sin (\x+.45, .85) cos (\x+.9, .6);
            }
        \end{scope}

    \end{scope}

\end{tikzpicture}
\end{center}
\caption{An illustration of notions used in \cref{lem:peralg}.
Shaded rectangles represent the common prefixes of $u_i$ and $v$.
In this case, $\frac{d_u-d_v}{|\rho|}=1$.}\label{fig:perlcp}
\end{figure}

Although our applications use \cref{fct:lcp} for two different purposes,
the overall scheme is the same each time. Consequently, we group two application-specific queries
in a single algorithmic lemma; see \cref{fig:perlcp}.

\peralg*
\begin{proof}
There is nothing to do for $k=0$.
For $k=1$, \cref{fct:lce} lets us check if $\lcp(u_0,v) = |u_0|$, that is, whether $u_0$ matches a prefix of $v$.
Moreover, $i=0$ maximizes $\lcp(u_i,v)$.

Henceforth, we shall assume $k\ge 2$. In this case, we retrieve an occurrence $T[p_0\dd p_1)$ of the string period $\rho$ of $\mathbf{p}$ and apply \cref{fct:lce} to determine $d_u = \lcpinf{\rho}{u_0}$ and $d_v = \lcpinf{\rho}{v}$.
We also compute $i_{t} = \frac{d_u-d_v}{|\rho|}$.

Let us observe that for $i \in [0 \dd k)$, $u_0 = \rho^i u_i$, so $\lcpinf{\rho}{u_i}=\lcpinf{\rho}{u_0}-i|\rho|=d_u -i|\rho|$.
If $i_t\in [0 \dd k)$, then $\lcpinf{\rho}{u_{i_t}}=d_v$.
Hence, by \cref{fct:lcp}, we have $\lcp(u_i,v)=d_v$ for $i<i_t$,
$\lcp(u_i,v)\ge d_v$ for $i=i_t$ (if $i_t \in [0 \dd k)$),
and $\lcp(u_i,v) = d_u - i|\rho|<d_v$ for $i>i_t$.

\ref{it:perpref}
We shall report $i\in [0\dd k)$ such that $\lcp(u_i,v)=|u_i|$.
For $i<i_t$, we have $\lcp(u_i,v)=d_v < d_u-i|\rho|\le |u_0|-i|\rho|=|u_i|$, so these indices are never reported.
If $i_t\in [0\dd k)$, we compute $\lcp(u_{i_t},v)$ using \cref{fct:lce} and this index may need to be reported.
For $i>i_t$, we have $\lcp(u_i,v)= d_u-i|\rho|$ and $|u_i|=|u_0|-i|\rho|$, so we either report all these indices (if $d_u=|u_0|$) or none of them (otherwise).

\ref{it:permax}
If $i_t \in [0\dd k)$, we check whether $\lcp(u_{i_t},v)>d_v$ using \cref{fct:lce}.
If so, we report $i_t$ as the only index maximizing $\lcp(u_i,v)$ because $\lcp(u_i,v)\le d_v$ holds unless $i=i_t$.
Otherwise, the maximum of $\lcp(u_i,v)$ is~$d_v$, attained for all $i\in [0\dd k)$ such that $i \le i_t$ (if $i_t \ge 0$),
or~$d_u$, attained for $i=0$ (if $i_t \le 0$).
\end{proof}

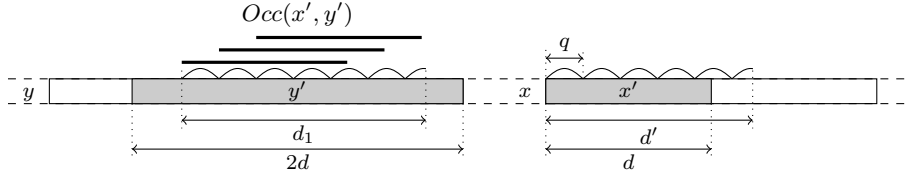
\subsection{Prefix-Suffix Queries and Their Applications}\label{sec:BQ}
In this section, we show the solutions for \BQ, \PQ, and \PEQ using \IPM.

We start with \BQ. Assume that $|x|,|y|\ge d$; otherwise, there are no suffixes to be reported.
Let $x'$ be the prefix of $x$ of length $d$ and $y'$ be the suffix of $y$ of length $\min(2d-1,|y|)$.
Suppose that a suffix $z$ of $y$ matches a prefix of $x$. If $|z|\ge d$, then $z$ must start with a fragment matching~$x'$.
Moreover, if $|z|\le 2d-1$, then $z$ is a suffix of $y'$, so this yields an occurrence of $x'$ in $y'$.
We find all such occurrences with a single \textsc{IPM Query} and then use \cref{lem:peralg} to find out which of them
can be extended to the sought suffixes $z$ of $y$.

\begin{figure}[ht]
\begin{center}
\begin{tikzpicture}[scale=0.7]

    \draw[dashed] (-1, 0)--(21, 0)  (-1, 0.6)--(21, 0.6);

    \draw (0,0) rectangle (10, 0.6);

    \draw (-.5, .3) node[anchor=mid] {\footnotesize$y$};

    \filldraw[fill=black!20] (2,0) rectangle (10, .6);
    \draw (6, .3) node[anchor=mid] {\footnotesize$y'$};

    \draw[dotted] (2, 0.6) -- (2, 1.2) (10, 0.6) -- (10, 1.2);
        \draw[latex-latex] (2, 1) -- node[anchor=mid,above=-0.09]{\footnotesize $2d-1$} (10, 1);
    

    \foreach \i in {0,1,2} {
        \draw (3.2+0.9*\i, -0.6-0.6*\i) rectangle (10,-0.6*\i);
        \filldraw[fill=black!20] (3.2+0.9*\i, -0.6-0.6*\i) rectangle node[anchor=mid]{\footnotesize $x'$}(7.2+0.9*\i,-0.6*\i);
        \draw (2.7+0.9*\i, -0.3-0.6*\i) node[anchor=mid]{\footnotesize $y_\i$};
    }

    \begin{scope}[xshift = 12 cm]
        \draw (0,0) rectangle (8,.6);
        \draw (-.5, .3) node[anchor=mid] {\footnotesize$x$};

        \filldraw[fill=black!20] (0,0) rectangle (4, .6);
        \draw (2, .3) node[anchor=mid] {\footnotesize$x'$};

        \draw[latex-latex] (0, 1) -- node[anchor=mid, above=-0.09]{\footnotesize $d$} (4, 1);

        \draw[dotted] (0,0.6)--(0,1.2)  (4,0.6)--(4,1.2);

    \end{scope}

\end{tikzpicture}
\end{center}
\caption{The notions used in the algorithms answering \BQ and \BLCP (for the latter, see \cref{sec:GSC}).}\label{fig:app-blcp}
\end{figure}

By \cref{fct:single}, the starting positions of the occurrences of $x'$ in $y'$ form a periodic progression in~$T$.
Let $y_i$ be the suffix of $y$ starting with the $i$th occurrence of $x'$; see \cref{fig:app-blcp}.
We need to check which of the fragments $y_i$ occur as prefixes of $x$.
This is possible using \cref{lem:peralg}\ref{it:perpref}, which lets us find all indices $i$ such that $y_i$ is a prefix of $x$.
The result is an integer interval of indices, which can be transformed into an arithmetic progression of lengths $|y_i|$.
Consequently, the data structure of \cref{thm:ipm} (which already contains the component of \cref{prop:lce} for \LCEQ) can answer \BQ in $\Oh(1)$ time.
Hence, we obtain the following results.

\thmappbq*

\thmapppq*

\begin{proof}
\PQ can be answered using \BQ for $y=x$.
To compute all periods of $x$, we use \BQ to find all borders of $x$ of length within $[2^k\dd 2^{k+1})$ for each $k\in [0\dd \floor{\log |x|}]$.
The lengths of borders can be easily transformed to periods since $x$ has period $p$ if and only if it has a border of length $|x|-p$.
\maybeqed \end{proof}

\thmrun*
\begin{proof}
  \PEQ can be answered using \BQ and \LCEQ in $T$ and $T^R$.
  Given a fragment $x=T[\ell\dd r)$, we use a \BQone to find the longest proper border of $x$, provided that its length is at least $\tfrac12|x|$.
  If no such border exists, we report that $\run(x)=\bot$.
  Otherwise, the length of the longest border yields the period $p=\per(x)\le \tfrac12|x|$.
  In this case, \[\run(x)=T[\ell - \lcs(T[0\dd \ell),T[0\dd \ell+p))\dd \ell + p + \lcp(T[\ell\dd n),T[\ell+p\dd n))),\]
  where $\lcs(uv,w)$ denotes the length of their longest common suffix of strings $v$ and $w$.
\maybeqed \end{proof}

\subsection{Cyclic Equivalence Queries}\label{sec:FC}

Recall that, for a non-empty string $w\in \Sigma^n$, we define a string $\Rot(w) = w[n-1]w[0]\cdots w[n-2]$.
First, we prove that the sought set $\ROT(x,y)=\{j\in \integ\,:\, y= \Rot^j(x)\}$
indeed forms an arithmetic progression.

\begin{fact}\label{fct:FC}
  If $\ROT(x,y)\ne \emptyset$, then $\ROT(x,y)$ is an infinite arithmetic progression whose difference divides~$|x|$.
\end{fact}
\begin{proof}
  Let us note that if $j,j'\in \ROT(x,x)$, then $j-j'\in \ROT(x,x)$ and thus also $\gcd(j,j')\in \ROT(x,x)$.
  Consequently, $\ROT(x,x)$ consists of multiples of some integer $m$.
  Due to $|x|\in \ROT(x,x)$, this integer $m$ is a divisor of $|x|$.

  Next, observe that if $j\in \ROT(x,y)$, then $\ROT(x,y)=\{j + j' : j'\in \ROT(x,x)\}$.
  Hence, if $\ROT(x,y)\ne \emptyset$, then $\ROT(x,y)$ is an infinite arithmetic progression whose difference divides $|x|$.
\end{proof}

While answering \FC, we can assume that $|x|=|y|$;
we denote the common length of $x$ and $y$ by $d$.
Our query algorithm is based on the following characterization of $\ROT(x,y)$:

\begin{observation}\label{obs:rot}
Let $x,y$ be strings of common length $d$. For every $j\in [0\dd d]$,
we have $j\in \ROT(x,y)$ if and only if $y[0\dd j) \cong x[d-j\dd d)$ and $y[j\dd d)\cong x[0\dd d-j)$.
\end{observation}

Below, we provide an algorithm that computes $\ROT(x,y)\cap [\ceil{\frac{d}{2}}\dd d)$.
By \cref{obs:rot}, $j\in \ROT(x,y)$ if and only if $d-j\in \ROT(y,x)$, so running this algorithm both for $(x,y)$
and $(y,x)$ lets us retrieve $\ROT(x,y)\cap [1\dd d)$.
This is sufficient to determine $\ROT(x,y)$ because an \textsc{LCE Query} lets us easily check if $x \cong y$, i.e., whether $0\in \ROT(x,y)$, and \cref{fct:FC} yields $\ROT(x,y) = \{j \in \integ : j\bmod d \in \ROT(x,y) \cap [0\dd d)\}$.

By \cref{obs:rot}, if $j\in [\ceil{\frac{d}{2}}\dd d)\cap \ROT(x,y)$,
then the length-$j$ suffix of $x$ matches a prefix of $y$.
Since $d-1 \le 2\ceil{\frac{d}{2}}-1$, all lengths $j_0< \cdots < j_{k-1}$ satisfying the latter condition form an arithmetic progression and can be retrieved with a single 
\textsc{Prefix-Suffix Query}. 
Moreover, \cref{obs:rot} yields $[\ceil{\frac{d}{2}}\dd d)\cap \ROT(x,y) = \{j_i : i\in [0\dd k)\text{ and }y[j_i\dd d) \cong x[0\dd d-j_i)\}$.
Hence, it suffices to check for which indices $i$ the suffix $y_i := y[j_i\dd d)$ of $y$ matches a prefix of $x$.
For this, we note that $(j_i)_{i=0}^{k-1}$ is a periodic progression in $y$:
for each $i\in [1\dd k)$, the string $y[j_{i-1}\dd j_{i})$ matches a suffix of $x$
whose length is the difference of the arithmetic progression $(j_i)_{i=0}^{k-1}$.
Hence, \cref{lem:peralg}\ref{it:perpref} lets us retrieve an integer interval
consisting of indices $i$ such that $j_i \in \ROT(x,y)$, and this interval can be easily
transformed into an arithmetic progression of the corresponding values $j_i$.
Consequently, the data structure of \cref{thm:ipm} (which already contains the component of \cref{prop:lce} for \LCEQ and which can answer \BQ in $\Oh(1)$ time; cf.\ \cref{thm:app-bq})
can also answer \FC in $\Oh(1)$ time.
 
\thmfc*

\subsection{Queries Related to Lempel--Ziv Compression}\label{sec:GSC}
\textsc{Substring Compression Queries} are internal queries asking for a compressed representation of a substring or the (exact or approximate) size of this representation.
This family of problems was introduced by Cormode and Muthukrishnan~\cite{DBLP:conf/soda/CormodeM05},
and some of their results were later improved by Keller et al.~\cite{DBLP:journals/tcs/KellerKFL14}.
\textsc{Substring Compression Queries} have a fairly direct motivation: 
Consider a server holding a long repetitive text $T$ and clients asking for substrings of~$T$ (e.g., chunks that should be displayed).
A limited capacity of the communication channel justifies compressing these substrings.

The aforementioned papers~\cite{DBLP:conf/soda/CormodeM05,DBLP:journals/tcs/KellerKFL14} apply the classic LZ77 compression scheme~\cite{DBLP:journals/tit/ZivL77}.
Among other problems, they consider internal queries for the LZ factorization of a given fragment $x$
and for the generalized LZ factorization of one fragment $x$ in the context of another fragment $y$.
The latter is defined as the part representing $x$ in the LZ factorization of a string $y\# x$, 
where $\#$ is a special sentinel symbol not present in the text. 
A server can send such a generalized LZ factorization to a client who requests $y$ and has previously received $x$.

Thus, in this section, we consider variants of \LSC.
Let us first recall Lempel--Ziv LZ77 algorithm~\cite{DBLP:journals/tit/ZivL77} and range successor queries that will be used in our solution.

\paragraph{LZ77 Compression}
Consider a string $w\in \Sigma^*$. We say that a fragment $w[\ell\dd r)$ has a \emph{previous occurrence} (or is a \emph{previous fragment}) if $w[\ell\dd r)\cong w[\ell'\dd r')$ for some positions $\ell'<\ell$ and $r'<r$. The fragment $w[\ell\dd r)$ has a \emph{non-overlapping previous occurrence} (or is a \emph{non-overlapping previous fragment}) if additionally $r'\le \ell$.

The Lempel--Ziv factorization $\LZ(w)$ is a factorization $w=f_1 \cdots f_k$ into fragments (called \emph{phrases}) such that 
each phrase $f_i$ is the longest previous fragment starting at position $|f_1\cdots f_{i-1}|$,
or a single letter if there is no such previous fragment.
The non-overlapping Lempel--Ziv factorization $\LZ_N(w)$ is defined analogously, allowing for non-overlapping previous fragments only.
Both factorizations (and several closely related variants) are useful for compression
because a previous fragment can be represented using a reference to the previous occurrence (e.g., the positions of its endpoints).

Strings $w\in \Sigma^*$ are sometimes compressed with respect to a \emph{context string} (or \emph{dictionary string}) $v\in \Sigma^*$.
Essentially, there are two ways to define the factorization $\LZ(w \mid v)$ of $w$ with respect to $v$.
In the \emph{relative LZ factorization}~\cite{DBLP:journals/tit/ZivM93,DBLP:conf/spire/KuruppuPZ10} $\LZ_R(w \mid v)$, each phrase is the longest fragment of $w$ which starts at the given position and occurs in $v$ (or a single letter if there is no such fragment).
An alternative approach is to allow both substrings of $v$ and previous fragments of $w$ as phrases. This results in the \emph{generalized LZ factorization}, denoted $\LZ_G(w \mid v)$; see~\cite{DBLP:conf/soda/CormodeM05,DBLP:journals/tcs/KellerKFL14}.
Equivalently,  $\LZ_G(w \mid v)$ can be defined as the suffix of $\LZ(v\#w)$ corresponding to $w$, where $\#$ is a special symbol that is present neither in $v$ nor in $w$.
The previous fragments in the non-overlapping generalized LZ factorization $\LZ_{NG}(w \mid v)$ must be non-overlapping.
\begin{example}
Let $w=\mathtt{aaaabaabaaaa}$ and $v=\mathtt{baabab}$.
We have
\begin{align*}
\LZ(w)\,&=\, \mathtt{a}\cdot \mathtt{aaa}\cdot \mathtt{b}\cdot \mathtt{aabaa}\cdot\mathtt{aa},&
\LZ_N(w)\,&=\, \mathtt{a}\cdot \mathtt{a}\cdot \mathtt{aa}\cdot \mathtt{b}\cdot \mathtt{aab}\cdot\mathtt{aaaa},\\
\LZ_R(w\mid v)\,&=\, \mathtt{aa}\cdot\mathtt{aaba}\cdot \mathtt{aba}\cdot \mathtt{aa}\cdot \mathtt{a},&
\LZ_{G}(w\mid v)\,&=\, \mathtt{aa}\cdot \mathtt{aaba}\cdot \mathtt{abaa}\cdot\mathtt{aa},\\
\LZ_{GN}(w\mid v)\,&=\, \mathtt{aa}\cdot \mathtt{aaba}\cdot \mathtt{aba} \cdot \mathtt{aaa}.
\end{align*}
\end{example}

\paragraph{Range Successor Queries}
We define the \emph{successor} of an integer $x$ in a set $A$ as
$\suc_{A}(x)=\min\{y\in A : y > x\}$.
Successor queries on a range $A[\ell \dd r]$ of an array $A$ are defined as follows.
\defdsproblem{\textsc{Range Successor Queries (Range Next Value Queries)}}{
  \textbf{Input}: An array $A$ of $n$ integers.\\
\textbf{Queries} Given a range $[\ell \dd r]$ and an integer $x$, compute $\suc_{A[\ell \dd r]}(x)$ (and an index $j\in [\ell\dd r]$ such that $A[j]=\suc_{A[\ell \dd r]}(x)$, if any).}

The following three trade-offs describe the current state of the art for such queries.
\begin{proposition}\label{prp:range_successor_ub}
For any constant $\eps>0$ and the functions $S_{\rsucc}$, $Q_\rsucc$, and $C_\rsucc$ specified below, there is a data structure of size 
 $S_{\rsucc}(n)$ that answers range successor queries in $Q_\rsucc(n)$ time and can be constructed in $C_\rsucc(n)$ time:
 \begin{enumerate}[label=(\alph*)]
   \item\label{it:succn} $S_{\rsucc}(n)=\Oh(n)$, $Q_{\rsucc}(n)=\Oh(\log^{\eps}{n})$,
   and $C_\rsucc(n)=\Oh(n\sqrt{\log n})$~\cite{DBLP:conf/swat/NekrichN12,DBLP:conf/soda/BelazzouguiP16};
   \item\label{it:succloglog} $S_{\rsucc}(n)=\Oh(n\log\log n)$, $Q_{\rsucc}(n)=\Oh(\log\log n)$, and $C_\rsucc(n)=\Oh(n\sqrt{\log n})$~\cite{DBLP:journals/ipl/Zhou16,Gao2020};
   \item\label{it:succeps} $S_{\rsucc}(n)=\Oh(n^{1+\eps})$, $Q_{\rsucc}(n)=\Oh(1)$,
   and $C_\rsucc(n)=\Oh(n^{1+\eps})$~\cite{DBLP:journals/tcs/CrochemoreIKRTW12}.
 \end{enumerate}
\end{proposition}

\paragraph{LZ Substring Compression Queries}
We consider the following types of queries.
\defdsproblem{\textsc{(Non-Overlapping) LZ Substring Compression Queries}}{
Given a fragment $x$ of $T$, compute the (non-overlapping) LZ factorization of $x$, i.e., $LZ(x)$ ($LZ_N(x)$, respectively).
}
\defdsproblem{\textsc{Relative LZ Substring Compression Queries}}{
Given two fragments $x$ and $y$ of $T$, compute the relative LZ factorization of $x$ with respect to $y$, i.e., $LZ_R(x \mid y)$.
}

\defdsproblem{\textsc{Generalized (Non-Overlapping) LZ Substring Compression Queries}}{
Given two fragments $x$ and $y$ of $T$, compute the generalized (non-overlapping) LZ factorization of $x$ with respect to $y$, i.e., $LZ_G(x \mid y)$
($LZ_{GN}(x \mid y)$, respectively).
}

Our query algorithms heavily rely on the results of Keller et al.~\cite{DBLP:journals/tcs/KellerKFL14} for \LSC and \GSC. 
The main improvement is a more efficient solution for the following auxiliary problem:
\defdsproblem{\BLCPFull}{Given two fragments $x$ and $y$ of $T$, find the longest prefix $p$ of $x$ which  occurs in $y$.}
The other, easier auxiliary problem defined in~\cite{DBLP:journals/tcs/KellerKFL14} is used as a black box.
\defdsproblem{\ILCPFull}{Given a fragment $x$ of $T$ and an interval $[\ell\dd r]$ of positions in $T$, 
find the longest prefix $p$ of $x$ which occurs in $T$  at some position within $[\ell\dd r]$.}

The data structure for \ILCP uses range successor queries, so we state the complexity in an abstract form.
This convention gets propagated to further results in this section.

\begin{lemma}[Keller et al.~\cite{DBLP:journals/tcs/KellerKFL14}]\label{lem:app-ilcp}
For a text $T$ of length $n$, there exists a data structure of size $\Oh(n+S_{\rsucc}(n))$ that answers \ILCP\ in $\Oh(Q_{\rsucc}(n))$ time.
The data structure can be constructed in $\Oh(n+C_{\rsucc}(n))$ time.
\end{lemma}

As observed in~\cite{DBLP:journals/tcs/KellerKFL14}, the decision version of \IPM easily reduces to \ILCP. For $x=T[\ell_x\dd r_x)$
and $y=T[\ell_y\dd r_y)$, it suffices to check if the longest prefix of $x$ occurring at some position in $[\ell_y\dd r_y-|x|]$ of $T$ is $x$ itself.%
\begin{corollary}[Keller et al.~\cite{DBLP:journals/tcs/KellerKFL14}]\label{lem:any}
For a text $T$ of length $n$, there exists a data structure of size $\Oh(n+S_{\rsucc}(n))$ that, given fragments $x,y$ of $T$, can decide in $\Oh(Q_{\rsucc}(n))$ time whether $x$ occurs in $y$.
The data structure can be constructed in $\Oh(n+C_{\rsucc}(n))$ time.
\end{corollary}

We proceed with our solution for \BLCP.
Let $x = T[\ell_x\dd r_x)$ and $y = T[\ell_y\dd r_y)$.
First, we search for the largest $k$ such that the prefix of $x$ of length~$2^k$ (i.e., $T[\ell_x\dd \ell_x+2^k)$) occurs in $y$.
We use a variant of the binary search involving exponential search (also called galloping search),
which requires $\Oh(\log K)$ steps, where $K$ is the optimal value of $k$.
At each step, for a fixed $k$ we need to decide if $T[\ell_x\dd \ell_x+2^k)$ occurs in $y$.
This can be done in $\Oh(Q_{\rsucc}(n))$ time using \cref{lem:any}.
At this point, we have an integer $K$ such that the optimal prefix $p$ has length $|p|\in [2^K\dd 2^{K+1})$.
The running time is $\Oh(Q_{\rsucc}(n) \log K) = \Oh(Q_{\rsucc}(n) \log\log |p|)$ so far.

Let $p'$ be the prefix obtained from an \textsc{Interval Longest Common Prefix Query} for $x$ and $[\ell_y\dd r_y-2^{K+1}]$.
We have $|p'|<2^{K+1}$ and thus the occurrence of $p'$ starting in $[\ell_y\dd r_y-2^{K+1}]$ lies within $y$.
Consequently, $|p|\ge |p'|$; moreover, if $p$ occurs at a position within $[\ell_y\dd r_y-2^{K+1}]$,
then $p=p'$.

The other possibility is that $p$ only occurs near the end of $y$, i.e., within the suffix of $y$ of length $2^{K+1}-1$, which we denote as $y'$.
We use a similar approach as for \BQ with $d=2^K$ to detect $p$ in this case.
We define $x'$ as the prefix of $x$ of length $2^K$.
An occurrence of $p$ must start with an occurrence of $x'$, so we find all occurrences of $x'$ in $y'$.
If there are no such occurrences, we conclude that $p=p'$.

Otherwise, we define $y_i$ as the suffix of $y$ starting with the $i$th occurrence of $x$; see~\cref{fig:app-blcp}.
Next, we apply \cref{lem:peralg}\ref{it:permax} to compute $\max_i \lcp(y_i,x)$. By the discussion above,
this must be the length of the longest prefix of $x$ which occurs in $y'$.
We compare its length to $|p'|$ and choose the final answer $p$ as the longer of the two candidates.

Thus, the data structure for \IPM, accompanied by the components of \cref{lem:app-ilcp,lem:any,prop:lce},
yields the following result:

\thmblcp*

Finally, we generalize the approach of~\cite{DBLP:journals/tcs/KellerKFL14} to support multiple types of \LSC using \cref{thm:blcp} to improve the running time.

\begin{theorem}\label{thm:lsc}
For every text $T$ of length $n$, there is a data structure of size $\Oh(n+S_{\rsucc}(n))$
that answers:
\begin{enumerate}[label=(\alph*)]
  \item \textsc{Non-Overlapping} \LSC,
  \item \textsc{Relative} \LSC,
  \item \GSC, and
    \item \textsc{Generalized Non-Overlapping} \LSC,
\end{enumerate}
each in $\Oh\bigl(F\cdot Q_{\rsucc}(n)\log\log\frac{|x|}{F}\bigr)$ time,
where $F$ is the number of phrases reported.
The data structure can be constructed in $\Oh(n+C_{\rsucc}(n))$ time.
\end{theorem}
\begin{proof}
Let $x=T[\ell_x\dd r_x)$ and suppose that we have already factorized $x'=T[\ell_x\dd m)$,
i.e., the next phrase needs to be a prefix of $x''=T[m\dd r_x)$.
Depending on the factorization type, it is chosen among the longest prefix of $x''$ that is a previous fragment of $x$ (i.e., has an occurrence starting within $[\ell_x\dd m)$), the longest prefix of $x''$ that is a non-overlapping previous fragment of $x$ (i.e., occurs in $x'$),
or the longest prefix of $x''$ that occurs in $y$.
The first case reduces to an \textsc{Interval Longest Common Prefix Query}, while the latter two---to \BLCP.
For each factorization type, we compute the relevant candidates and choose the longest one as the phrase;
if there are no valid candidates, the next phrase is a single letter, i.e., $T[m\dd m]$.

Thus, regardless of the factorization type, we report each phrase $f_i$ of the factorization $x=f_1\cdots f_F$
in $\Oh(Q_{\rsucc}(n) \log \log |f_i|)$ time. 
This way, the total running time is $\Oh\big(\sum_{i=1}^F Q_{\rsucc}(n)\log\log |f_i|\big)$, which is $\Oh\bigl(F\cdot Q_{\rsucc}(n)\log\log\frac{|x|}{F}\bigr)$  due to Jensen's equality applied to the concave $\log \log $ function.
\maybeqed \end{proof}

Let us note that in the case of ordinary \LSC the approach presented in \cref{thm:lsc} would result in $\Oh(F\cdot Q_{\rsucc}(n))$ query time
because only \ILCP would be used; this is exactly the algorithm for \LSC provided in~\cite{DBLP:journals/tcs/KellerKFL14}.

Hence, despite our improvements, there is still an overhead for using variants of the LZ factorization other than the standard one.
Nevertheless, the overhead disappears if we use the state-of-the-art $\Oh(n)$-size data structure for range successor queries. This is because the $\Oh(\log^{\eps} n)$ time complexity lets us hide $\log^{o(1)} n$ factors
by choosing a slightly greater $\eps$.
Formally, \cref{thm:lsc,prp:range_successor_ub} yield the following result:
\begin{corollary}\label{cor:app-blcp}
For every text $T$ of length $n$ over an alphabet $[0 \dd n^{\Oh(1)})$ and constant $\eps>0$, there is a data structure of size $\Oh(n)$
that answers \BLCP\ in $\Oh(\log^{\eps }n)$ time
and \LSC (for all five factorization types defined above) in $\Oh(\log^{\eps }n)$ time per phrase reported.
Moreover, the data structure can be constructed in $\Oh(n\sqrt{\log n})$ time.
\end{corollary}

\section*{Acknowledgement}
The authors wish to thank Dominik Kempa for helpful discussions regarding synchronizing sets
and Moshe Lewenstein for a suggestion to work on the Generalized Substring Compression problem.
Jakub Radoszewski was supported by the Polish National Science Center, grant no.\ 2022/46/E/ST6/00463.

\bibliographystyle{alphaurl}
\bibliography{ipm}

\appendix

\section{Proof of \cref{lem:maxcut}}\label{app:AMDC}

\AMDC*
\begin{proof}
First, we preprocess $G$ so that each $v\in V$ stores both incoming and outgoing arcs.
For $A,B\sub V$, we denote $E(A,B)\sub E$ to be the set of arcs leading from $A$ to $B$;
recall that the goal is to make sure that $|E(L,R)|\ge \tfrac14|E|$.
Given $v\in V$ and $A\sub V$, define $\deg^+_A(v):=|E(\{v\},A)|$
and $\deg^-_A(v):= |E(A,\{v\})|$.

We maintain a partition $V = L \cup M \cup R$ into three disjoint classes.
Initially, $M=V$ and, as long as $M\ne \emptyset$, we pick an arbitrary vertex $v\in M$
and move $v$ to $L$ or $R$, depending on whether \[2\deg^+_R(v)+\deg^+_M(v)\ge 2\deg^-_L(v)+\deg^-_M(v)\]
or not.
This decision can be implemented in $\Oh(1+\deg^+_V(v)+\deg^-_V(v))$ time, which yields a total running time of $\Oh(|V|+|E|)$.

As for correctness, we shall prove that \[\Phi:=4|E(L,R)| + 2|E(L,M)| +2|E(M,R)|+|E(M,M)|\]
cannot decrease throughout the algorithm.
Consider the effect of moving $v$ from $M$ to $L$ on the four terms of~$\Phi$ (recall that there are no self-loops $v\to v$):
\begin{itemize}
  \item $|E(L,R)|$ increases by $\deg^+_R(v)$;
  \item $|E(L,M)|$ increases by $\deg^+_M(v)$ and decreases by $\deg^-_L(v)$;
  \item $|E(M,R)|$ decreases by $\deg^+_R(v)$;
  \item $|E(M,M)|$ decreases by $\deg^+_M(v)$ and decreases by $\deg^-_M(v)$.
\end{itemize}
Overall, $\Phi$ increases by 
  \begin{align*}4\cdot \deg^+_R(v) + 2\cdot (\deg^+_M(v) - \deg^-_L(v)) + 2\cdot (-\deg^+_R(v)) + (- \deg^+_M(v) - \deg^-_M(v))\\  = 2\deg^+_R(v) + \deg^+_M(v) - 2\deg^-_L(v) - \deg^-_M(v),\end{align*}
and this quantity is non-negative when the algorithm decides to move $v$ to $L$. 

Similarly, if $v$ is moved from $M$ to $R$, then $\Phi$ does not decrease.
Upon the end of the algorithm, we have $\Phi=4|E(L,R)|$ due to $M=\emptyset$,
whereas, initially, $\Phi = |E(M,M)|=|E|$ due to $M=V$. 
Since $\Phi$ is non-decreasing, we conclude that $4|E(L,R)|\ge |E|$ holds as claimed. 
\end{proof}

\end{document}